%% file: main.tex
\documentclass[12pt]{article}
\usepackage{amsmath,amssymb,amsthm,amsfonts,setspace,fullpage,graphicx,color,paralist,nicefrac,comment, hyperref,booktabs,multirow, float, caption}
\hypersetup{
	colorlinks=true,
	linkcolor=red,
	citecolor=blue,
	filecolor=magenta,      
	urlcolor=cyan,
}
\usepackage[multiple]{footmisc}
\usepackage{mathtools}
\usepackage{thmtools}
\usepackage{thm-restate}
\usepackage{caption}
\usepackage{bbm}
\usepackage[multiple]{footmisc}
\usepackage{natbib}
\usepackage{bibunits}

\defaultbibliographystyle{ecta}
\defaultbibliography{reference}

\renewcommand{\bibsection}{%
  \section*{References}
  \vspace{-0.5em}
}

\usepackage{apptools}
\usepackage[margin = 1.1in]{geometry}
\usepackage{graphicx}
\usepackage{tikz}
\usepackage{pgfplots}  
\usepgflibrary{shapes.geometric}  
\usetikzlibrary{patterns,calc}
\usepackage[position=bottom]{subfig}

\usepackage{pgfplots}
\pgfplotsset{compat=1.18} 
\usetikzlibrary{arrows.meta,positioning}

\newcommand{\de}{\mathop{}\!\mathrm{d}}

\def\w{\omega}
\def\l{\lambda}
\def\m{\mu}
\def\t{\tau}
\def\e{\varepsilon}

\def\D{\Delta}
\def\W{\Omega}

\def\R{\mathbb{R}}

\def\BV{\mathbf{V}}

\def\TT{\mathcal{T}}
\def\VV{\mathcal{V}}

\newcommand{\set}[1]{\{ #1 \}} 

\DeclareMathOperator*{\cav}{cav}
\DeclareMathOperator*{\co}{co}
\DeclareMathOperator{\Cov}{Cov}

\DeclareMathOperator{\Var}{Var}

\DeclareMathOperator*{\interior}{int}

\DeclareMathOperator{\marg}{marg}

\DeclareMathOperator*{\Gr}{Gr}
\DeclareMathOperator*{\argmin}{argmin}
\DeclareMathOperator*{\argmax}{argmax}

\DeclareMathOperator{\supp}{supp}

\newcommand{\oV}{\mkern2mu \overline{\mkern-2mu V \mkern-2mu} \mkern2mu}
\newcommand{\uV}{\mkern2mu \underline{ V \mkern-3mu} \mkern3mu}

\newtheorem{theorem}{Theorem}

\newtheorem{claim}{Claim}

\newtheorem{corollary}{Corollary}

\newtheorem{lemma}{Lemma}

\newtheorem{proposition}{Proposition}

\theoremstyle{definition}
\newtheorem{definition}{Definition}

\newtheorem{Remark}{Remark}

\newtheorem{thmalt}{Theorem}[theorem]
\newenvironment{thmp}[1]{
  \renewcommand\thethmalt{#1}
  \thmalt
}{\endthmalt}

\newtheorem{defalt}[definition]{Definition}
\newenvironment{defp}[1]{
  \renewcommand\thedefalt{#1}
  \defalt
}{\enddefalt}

\renewcommand{\blacksquare}{\rule{0.5em}{0.5em}}
\makeatletter
\renewenvironment{proof}[1][\proofname]{\noindent  \par\pushQED{\qed}\normalfont  \topsep6\p@\@plus6\p@\relax
  \trivlist\item[\hskip\labelsep\bfseries#1\@addpunct{.}]  \ignorespaces
}{\ \hfill\blacksquare\endtrivlist\@endpefalse
}
\makeatother

\makeatletter
\newcommand{\pushright}[1]{\ifmeasuring@#1\else\omit\hfill$\displaystyle#1$\fi\ignorespaces}
\newcommand{\pushleft}[1]{\ifmeasuring@#1\else\omit$\displaystyle#1$\hfill\fi\ignorespaces}
\makeatother

\newcommand\blfootnote[1]{%
   \begingroup
   \renewcommand\thefootnote{}\footnote{#1}%
   \addtocounter{footnote}{-1}%
   \endgroup
}

\begin{document}

\newpage

\setcounter{page}{1}

   \title{The Bounds of Mediated Communication}
    
    \author{Roberto Corrao\thanks{Department of Economics, Stanford, \href{mailto:rcorrao@stanford.edu}{rcorrao@stanford.edu}} \\ Stanford \\\and Yifan Dai\thanks{Department of Economics, MIT, \href{mailto:yfdai@mit.edu}{yfdai@mit.edu}} \\ MIT
}
     \date{\vspace{-0.5em} June\ 2026 \vspace{-1em}}
\onehalfspacing

\maketitle
\thispagestyle{plain}

\begin{abstract}
We study sender--receiver games with transparent motives, where an uninformed, sender-aligned mediator commits to a communication mechanism but cannot verify the sender's report. We compare mediated communication with cheap talk and Bayesian persuasion. A belief-value distribution is implementable by mediation exactly when it satisfies Bayes plausibility, receiver obedience, and zero covariance between posterior beliefs and sender values. This formulation separates the roles of commitment and verifiability. We show that mediation attains the persuasion value only when cheap talk also does; thus, whenever persuasion strictly outperforms cheap talk, unverifiable reports create a strict loss. We characterize when mediation strictly improves on cheap talk: the sender's value must exhibit countervailing effects along some direction in belief space. In finite-action, binary-state environments, this characterization yields a tractable description of optimal mediation. Applications cover platforms, lobbying, acceptance games, and a reinterpretation of the model as matching with externalities, where we study the efficiency-fairness tradeoff.
\end{abstract}

\quad JEL codes: D82, D83

\quad
\vspace{-0.5cm}Keywords: Sender-receiver games; communication equilibrium; value of mediation

\blfootnote{We are particularly indebted to Drew Fudenberg, Stephen Morris, and Alexander Wolitzky for many suggestions that substantially improved the paper. We are grateful to Charles Angelucci, Ian Ball, Alessandro Bonatti, Simone Cerreia-Vioglio, Laura Doval, Bob Gibbons, Marina Halac, Nicole Immorlica, Navin Kartik, Andrew Koh, Haruki Kono, Giacomo Lanzani, Xiao Lin, Elliot Lipnowski, Joel Flynn, Teddy Mekonnen, David Pearce, Doron Ravid, Sivakorn Sanguanmoo, Vasiliki Skreta, Alex Smolin, Juuso Toikka, Giacomo Weber, and Leeat Yariv for helpful comments and conversations. Roberto Corrao gratefully acknowledges the Gordon B. Pye Dissertation Fellowship for financial support. An earlier version of this paper was circulated under the title ``Mediated Communication with Transparent Motives."}

\thispagestyle{empty} 
\clearpage

\newpage

\setcounter{page}{1}

\section{Introduction}
Information intermediaries mediate communication between an informed party and an uninformed decision maker. A platform may stand between a seller who knows which buyers a product is best suited for and buyers who must decide whether to purchase the product. Similarly, a trade association may channel information from interested groups to policymakers who must choose among competing policy alternatives. In other cases, the same role is played not by a literal third party, but by a communication device or automated algorithm that replicates the commitment normally provided by an intermediary. In many of these examples, the informed party has an agenda over the receiver's actions that does not itself depend on the state. For example, a seller wants to maximize profits, and a lobbyist wants new favorable policies to be implemented. In addition, information is often soft: the mediator can process and disclose what the sender reports, but cannot directly verify the underlying state.

This paper studies the limits of mediated communication in these settings. Consider sender--receiver games with transparent motives \citep{LR20}: the sender is privately informed about the state, the receiver takes an action, and the sender's payoff depends on the receiver's action but not directly on the state. We compare three communication protocols. In cheap talk \citep{CS82}, the sender communicates directly with the receiver and cannot commit. In mediated communication \citep{myerson1982optimal,forges1986approach}, an uninformed mediator can commit to a stochastic communication mechanism that maps the sender's report into messages for the receiver. 
In Bayesian persuasion \citep{kamenica2011bayesian}, the sender can commit to an experiment about the true state, or equivalently, a mediator can freely verify that the sender's report corresponds to the true state. Assuming that the sender-preferred equilibrium is selected, persuasion is weakly preferred to mediation, which is weakly preferred to cheap talk.

Both persuasion and mediation rely on commitment to randomized mechanisms; what differs is whether the rule is applied to the true state or to the sender's unverifiable report. We interpret such randomized rules broadly, as algorithmic disclosure policies, smart contracts, or institutional procedures fixed in advance. Under this interpretation, the possible gaps between persuasion, mediation, and cheap talk yield two questions. First, when can communication based on unverifiable reports replicate what direct verification would achieve, that is, when is \emph{verifiability valuable}? Second, when can a non-verifying mediator improve on unmediated cheap talk, that is, when is \emph{mediation valuable}?

We first give a belief-based characterization of mediated communication: a joint distribution of the receiver's beliefs and the sender's interim values is implementable via a communication mechanism if and only if it satisfies three conditions: Bayes plausibility, receiver obedience, and a zero-covariance restriction between beliefs and values. This condition is the belief-space form of the sender's incentive constraint. It says that the sender's value may vary across the receiver's realized posterior beliefs, but this variation cannot be systematically correlated with the sender's report. Differently, cheap talk imposes the stronger restriction that the sender's value is constant across all realized posterior beliefs, and persuasion removes any incentive constraint on the sender's side.

Our second result uses this characterization to show that verifiability is valuable if and only if commitment is valuable. Thus, whenever commitment to a verifiable experiment strictly improves on cheap talk, the inability to verify the sender's report necessarily entails a strict loss relative to persuasion. The result rules out the possibility that mediation closes the entire gap between cheap talk and persuasion.

\begin{figure}[t]
\centering
\resizebox{0.4\textwidth}{!}{%
\input{theorem3_simplex_figure_input.tex}
}
\caption{Directional improvability.}
\begin{minipage}{0.93\textwidth}
\small \emph{Note.} The triangle is the belief simplex over three states. $p$ is the prior; $\mu^+$ is a closer, more favorable posterior, and $\mu^-$ is a farther, less favorable one. The vertical levels are cheap-talk values: $v^+$ at $\m^+$, $\overline V_{CT}(p)$ at $p$, and $v^-$ at $\mu^-$, with \(v^+>\overline V_{CT}(p)>v^-\).
\end{minipage}
\label{fig:directional-improvability}
\end{figure}

Our third result characterizes when mediation improves on cheap talk through a condition we call directional improvability. Geometrically, this is a failure of \emph{weak single crossing} of the cheap-talk value correspondence along some one-dimensional direction through the prior. Specifically, along such a direction,
there is a favorable posterior near the prior yielding a value above the cheap-talk value at the prior, and a farther unfavorable one yielding a value below it (Figure \ref{fig:directional-improvability}).
The unfavorable posterior disciplines the sender's report, allowing the mediator to sustain the favorable one that direct communication could not credibly induce. The relative positions of the two posteriors ensure that this incentive balance leaves a gain over cheap talk. In finite-action, binary-state environments, the same construction yields a tractable characterization of sender-optimal mediation, reducing the problem to comparing finitely many candidate mechanisms.

These results show that the value of mediation must come from a specific kind of noise. The mediator exploits \emph{countervailing effects} in how the receiver's beliefs affect the sender's payoff: some belief changes are attractive to the sender, while others are unattractive but useful for credibility. Directional improvability is the condition that these countervailing effects exist and can be balanced in the sender's favor. Furthermore, we show that the same logic applies to costly signaling games with transparent motives (cf. \cite{koessler2024belief}). Here, the mediator has an even greater edge over unmediated communication due to the possibility of randomizing over both cheap messages and costly signals.

We develop two direct applications of our results. The first is lobbying through an intermediary. A policymaker chooses between a safe status quo and several active policies whose payoffs depend on the state. The lobbyist favors higher active policies and dislikes the status quo, whereas the policymaker wants the chosen policy to fit the state but reverts to the status quo when uncertainty is sufficiently high.  Countervailing effects arise because the lobbyist prefers the policymaker to be certain enough to avoid the status quo, while favoring higher active policies. We show that the set of priors where both verifiability and mediation are valuable expands when the policymaker's cost of choosing the wrong policies increases. This rationalizes the use of political mediators under severe negative consequences of choosing the wrong policy.

The second application is a class of acceptance games in which the receiver accepts or rejects a risky prospect and has a privately known outside option.
A non-monotonicity condition on the receiver's value of the prospect creates the countervailing effects that make mediation valuable. Under log-concavity of the outside-option distribution, this translates into a strict ex-ante Pareto improvement for both players.

We finally reinterpret the model as a reduced-form matching problem with externalities and study the tradeoff between efficiency and fairness therein. Agents with heterogeneous traits are assigned to groups, and the value generated by a group depends on its composition. We assume that all agents matched in the same group receive the value, that is, there are \emph{aligned preferences}. In this language, a distribution over posteriors becomes a distribution over group compositions, and the sender's interim value becomes the value generated by a group. We consider Pareto efficiency and two notions of fairness.
 Ex-ante fairness requires that the value generated by a group cannot covary with the group composition in a way that systematically benefits some traits. Ex-post fairness requires that all realized groups generate the same value. The distinction between mediation and cheap talk becomes the distinction between ex-ante and ex-post fairness constraints. Extending our results, we show that a feasible allocation is Pareto optimal and ex-ante fair if and only if it is Pareto optimal and ex-post fair: Under ex-ante fairness, high-value groups cannot systematically contain different traits from low-value groups. Hence, if realized group values are not already equal, one can put more weight on high-value groups and less weight on low-value groups without changing the aggregate trait distribution, making some traits better off and no trait worse off.

\paragraph{Literature review}
This paper studies mediated communication \citep{myerson1982optimal,forges1986approach} using a belief-based formulation, as in Bayesian persuasion \citep{kamenica2011bayesian} and cheap talk with transparent motives \citep{LR20}.\footnote{\cite{aumann1995repeated} and \cite{aumann2003long} use belief-based arguments to study, respectively, zero-sum repeated games with asymmetric information and long cheap talk. More recently, \cite{koessler2024belief} apply a belief-based approach to characterize PBE payoffs in signaling games.} The belief-based formulation allows us to characterize feasible distributions over posterior beliefs under mediation and compare mediation directly with persuasion and cheap talk. Related duality arguments appear in the Bayesian persuasion literature \citep{DM19,DK22,kolotilin2022persuasion}.

A recent literature compares mediation with other communication protocols in the uniform-quadratic environment of \cite{CS82}. \cite{blume2007noisy} compare noisy cheap talk with cheap talk, while \cite{goltsman2009mediation} compare mediation, cheap talk, and delegation. We instead compare persuasion, mediation, and cheap talk under state-independent sender preferences, but without imposing additional parametric assumptions.  Relatedly, \cite{rudov2026extreme} study when correlated equilibria improve on Nash equilibria in general complete information games. Unlike this paper, we focus on cheap-talk games with asymmetric information.

The closest paper to ours is \cite{salamanca2021value}, which studies mediated communication in \emph{finite} games through a recommendation-based formulation close to \cite{myerson1982optimal}. Our analysis differs in three respects. First, the models are not nested: we focus on transparent motives, but allow for infinitely many actions and states. Second, we work directly with posterior beliefs rather than recommendations. Importantly, this different approach is what we leverage to explicitly compare mediation with persuasion and cheap talk. Third, the main results are different. While \cite{salamanca2021value} establishes strong duality for the recommendation-based mediation problem, we use perturbation arguments to characterize when verifiability and mediation are valuable.\footnote{\cite{salamanca2021value} provides a binary-state example with transparent motives in which both verifiability and mediation are valuable, but does not characterize when this happens.}

The paper is also related to work on Bayesian persuasion with limited commitment \citep{LL22,lipnowski2022persuasion,koessler2021information,le2024mediated}. Like mediation, these protocols can be viewed as intermediate between persuasion and cheap talk. Under transparent motives, however, some of them are equivalent to these two benchmarks.\footnote{For example, the credible information structures in \cite{LL22} coincide with those feasible under persuasion, and the sender's optimal payoff in long cheap-talk (\cite{aumann2003long}, \cite{LR20}) equals the single-round cheap-talk one.} By contrast, mediation can deliver a sender value strictly between the persuasion and cheap-talk bounds; we characterize when this occurs.

When interpreted as a matching model, our contribution lies at the intersection of several literatures. Aligned preferences have been applied, among other things, to study stability and efficiency in both partnerships, e.g., \cite{farrell1988partnerships}, and one-to-one matching markets, e.g., \cite{ferdowsian2025strategic}. Unlike these papers, we allow for arbitrary externalities/spillover effects and focus on the efficiency-fairness tradeoff.\footnote{Notably, \cite{echenique2024stable} consider aligned preferences in one-to-one matching markets and connect both stability and fairness (in a Rawlsian sense) to solutions of optimal transport problems. Again, they do not consider externalities and nonlinear distributional effects.} More closely related to our paper, \cite{doval2021information} and \cite{kolotilin2022persuasion} start from a pure information-design problem to derive implications on efficiency, welfare, and assortative patterns in models with spillover effects. Differently, we focus on aligned preferences and the interpretation of the incentive constraints of mediation and cheap talk as fairness conditions.\footnote{\cite{kolotilin2022persuasion} generalize the models of spillover effects in firms of \cite{saint2001distribution} and in schools of \cite{epple1998competition}.}
Finally, in the context of algorithmic design, \cite{liang2026algorithm} have also employed information-design tools to study the accuracy-fairness tradeoff.

\paragraph{Outline of the paper}
Section \ref{Subsec: Illustrative Example} presents a platform-mediated disclosure example that previews our results. Section \ref{sec:model} formally introduces the model. Section \ref{sec:imple} develops the belief-based characterization of mediated communication. Section \ref{sec:main_results} characterizes when verifiability and mediation are valuable, and solves the optimal mediation problem in binary-state environments. Section \ref{sec: model interpretation} discusses interpretations of the model. Section \ref{sec:applications} applies the results to lobbying through an intermediary and to acceptance games. Section \ref{sec:matching} reinterprets the framework as a matching problem with externalities and derives implications for the efficiency-fairness tradeoff. Finally, Section \ref{sec:signaling} extends the analysis to signaling games with transparent motives. All proofs are relegated to the Appendix.

\section{Illustration: Platform-Mediated Disclosure}
\label{Subsec: Illustrative Example}

A seller lists a product on an online platform whose appeal depends on horizontal fit, such as a laptop that may be better suited to designers or to gamers, or software that may be better suited to general users or to specialists. Consider a binary-state model with three buyers: two of type $0$, and one of type $1$. The seller privately observes which buyer type the product is better matched to, summarized by a state $\w\in \{0,1\}$. A state-$0$ product is better suited to type-$0$ buyers, while a state-$1$ product is better suited to type-$1$ buyers. Buyers' common prior is \(p=\Pr(\omega=1) \in (0,1)\). The price is normalized to
$1/3$, and the payoffs from buying are $(1-\w) -1/3$ for type-$0$ buyers and $\w-1/3$ for type-$1$ buyers.\footnote{Section \ref{sec:signaling} allows the seller to choose prices endogenously, turning the game into a costly signaling one.}

A platform mediates communication by committing in advance to a public information policy, possibly stochastic, that maps the seller's report into a public message. The platform does not observe the state or buyers' types and relies on the seller's unverifiable report.\footnote{Since communication is public and the price is common, it is without loss to assume the platform does not elicit buyers' types.} We compare the seller's optimal payoff under mediation with two benchmarks. Under cheap talk, the seller communicates directly with buyers. Under persuasion, the platform can verify the state directly, so it does not rely on the seller's report. Thus, the value of mediation is the gain from platform commitment, while the value of verifiability is the additional gain from allowing the platform to verify the state. 

In all cases, communication is public: there are no targeted messages to different buyers. Buyers observe public information about the product and update their belief from \(p\) to a common posterior \(\mu \in [0,1]\). Buying is a best response for a type-$1$ buyer if and only if $\mu \geq 1/3$, and for a type-$0$ buyer if and only if $1-\mu \geq 1/3$. Since the price is fixed, we rescale the seller's payoff to be the number of purchasing buyers. The seller's indirect utility correspondence $\BV(\mu)$ in terms of the buyers' posterior is presented in Figure \ref{Fig: illustrative}. The key feature is that \(\BV\) is non-monotone: a posterior that makes the product more attractive to one consumer type can make it less attractive to the other. This countervailing effect creates a wedge between cheap talk and mediation.

The comparison depends on the prior, as shown in Figure \ref{Fig: illustrative}(a). If the prior is intermediate,
\(p\in[1/3,2/3]\), both buyer groups already buy, so all three protocols attain the maximal payoff. If the prior is low, \(p\in(0,1/3)\), both mediation and verifiability are valuable. If the prior is high, \(p\in(2/3,1)\), verifiability is valuable, but mediation is not.

We focus on the case $p \in (0, 1/3)$, where both mediation and verifiability are valuable. 
Under cheap talk, every message used in equilibrium must give the seller the same payoff. In the belief-based representation, the seller's payoff must be constant across all induced posteriors, whose distribution must average to $p$. Hence, the seller's optimal payoff under cheap talk at $p$ is $2$, which corresponds to the quasiconcave envelope of \(\BV\) evaluated at \(p\).

\begin{figure}[t]
\centering
\caption{Illustrative Example}
\label{Fig: illustrative}

\def\Vscale{1} 
\definecolor{mutwocol}{RGB}{190,45,45}     
\definecolor{muthreecol}{RGB}{55,105,180}  

\tikzset{
    frontier/.style={black,line width=1.08pt,line cap=round,line join=round},
    guide/.style={black!30,densely dotted,line width=0.65pt},
    point/.style={only marks, mark=*, mark size=2.35pt},
    hintone/.style={draw=black,line width=0.95pt,-stealth},
    hinttwo/.style={draw=mutwocol,line width=0.95pt,-stealth},
    hinttwovert/.style={draw=mutwocol,line width=0.95pt,-stealth},
    hintthree/.style={draw=muthreecol,line width=0.95pt,-stealth}
}
\begin{minipage}[t]{0.48\textwidth}
\centering
\begin{tikzpicture}
\begin{axis}[
    width=\textwidth,
    height=5.9cm,
    xmin=0, xmax=1.08,
    ymin=0.6*\Vscale, ymax=3.3*\Vscale,
    axis lines=middle,
    axis line style={black,-{Latex[length=1.8mm]}},
    xtick={0,0.33,0.66,1}, 
    xticklabels={}, 
    ytick={\Vscale,2*\Vscale,3*\Vscale},
    yticklabels={$1$,$2$,$3$}, 
    tick label style={font=\scriptsize},
    xlabel={$\mu$}, 
    ylabel={$\BV(\mu)$}, 
    xlabel style={ at={(axis description cs:0.985,0)}, anchor=west, font=\scriptsize, color=black }, 
    ylabel style={ at={(axis description cs:0.02,1.005)}, anchor=south, rotate=0, font=\scriptsize, color=black }, 
    tick style={black}, 
    clip=false
]

\draw[frontier]
    (axis cs:0,2)       -- (axis cs:0.33,2)
    -- (axis cs:0.33,3) -- (axis cs:0.66,3)
    -- (axis cs:0.66,1) -- (axis cs:1,1);

\node[font=\scriptsize, anchor=base, yshift=-2ex, inner sep=0pt] at (axis cs:0,0.6*\Vscale) {$0$};
\node[font=\scriptsize, anchor=base, yshift=-2ex, inner sep=0pt] at (axis cs:0.33,0.6*\Vscale) {$1/3$};
\node[font=\scriptsize, anchor=base, yshift=-2ex, inner sep=0pt] at (axis cs:0.66,0.6*\Vscale) {$2/3$};
\node[font=\scriptsize, anchor=base, yshift=-2ex, inner sep=0pt] at (axis cs:1,0.6*\Vscale) {$1$};

\addplot[
    orange!85!black, very thick,
    domain=0:0.333333,
    samples=200
] {\Vscale*((2-3*x)/(1-2*x))};

\addplot[
    orange!85!black, very thick
] coordinates {
    (0.333333,{\Vscale*3})
    (0.666667,{\Vscale*3})
};

\addplot[
    orange!85!black, very thick
] coordinates {
    (0.666667,{\Vscale*1})
    (1,{\Vscale*1})
};

\addplot[
    blue!80!black, very thick, dashed, domain=0:0.75, samples=2
] coordinates {(0,2) (0.33,3)};
\addplot[
    blue!80!black, very thick, dashed, domain=0.75:1, samples=100
] coordinates {(0.33,3) (0.66,3)};
\addplot[
    blue!80!black, very thick, dashed, domain=0.75:1, samples=100
] coordinates {(0.66,3) (1,1)};

\addplot[
    red!75!black, very thick, dashed
] coordinates {(0,2) (0.33,2)};
\addplot[
    red!75!black, very thick, dashed
] coordinates {(0.33,3) (0.66,3)};
\addplot[
    red!75!black, very thick, dashed
] coordinates {(0.66,1) (1,1)};

\node[
    font=\scriptsize,
    text=blue!80!black,
    fill=white,
    inner sep=1pt
] at (axis cs:0.15,2.7) {BP};

\node[
    font=\scriptsize,
    text=orange!85!black,
    inner sep=1pt
] at (axis cs:0.25,2.25) {MD};

\node[
    font=\scriptsize,
    text=red!75!black,
    inner sep=1pt
] at (axis cs:0.2,1.8) {CT};

\node[font=\scriptsize] at (axis cs:0.2,1.3)
    {\textcolor{blue!80!black}{BP} $>$ \textcolor{orange!85!black}{MD} $>$ \textcolor{red!75!black}{CT}};

\node[font=\scriptsize] at (axis cs:0.5, 3.3)
    {\textcolor{blue!80!black}{BP} $=$ \textcolor{orange!85!black}{MD} $=$ \textcolor{red!75!black}{CT}};

\node[font=\scriptsize] at (axis cs:0.9,1.3)
    {\textcolor{blue!80!black}{BP} $>$ \textcolor{orange!85!black}{MD} $=$ \textcolor{red!75!black}{CT}};

\end{axis}
\end{tikzpicture}

\vspace{0.35em}
{\small (a) Comparison of Bayesian persuasion, mediation, and cheap talk}
\end{minipage}
\hfill
\begin{minipage}[t]{0.48\textwidth}
\centering
\begin{tikzpicture}

\def\pstar{0.166667}
\def\mup{0.333333}
\def\mum{1}

\begin{axis}[
    width=\textwidth,
    height=5.9cm,
    xmin=0, xmax=1.08,
    ymin=0.6*\Vscale, ymax=3.3*\Vscale,
    axis lines=middle,
    axis line style={black,-{Latex[length=1.8mm]}},
    xtick={0,0.166667,0.333333,0.666667,1}, 
    xticklabels={}, 
    ytick={\Vscale,2*\Vscale,3*\Vscale},
    yticklabels={$1$,$2$,$3$}, 
    tick label style={font=\scriptsize},
    xlabel={$\mu$}, 
    ylabel={$\BV(\mu)$}, 
    xlabel style={ at={(axis description cs:0.985,0)}, anchor=west, font=\scriptsize, color=black }, 
    ylabel style={ at={(axis description cs:0.02,1.005)}, anchor=south, rotate=0, font=\scriptsize, color=black }, 
    tick style={black}, 
    clip=false
]

\node[font=\scriptsize, anchor=base, yshift=-2ex, inner sep=0pt] at (axis cs:0,0.6*\Vscale) {$0$};
\node[font=\scriptsize, anchor=base, yshift=-2ex, inner sep=0pt] at (axis cs:0.166667,0.6*\Vscale) {$p$};
\node[font=\scriptsize, anchor=base, yshift=-2ex, inner sep=0pt] at (axis cs:0.333333,0.6*\Vscale) {$\mu^+$};
\node[font=\scriptsize, anchor=base, yshift=-2ex, inner sep=0pt] at (axis cs:1,0.6*\Vscale) {$\mu^-$};

\draw[frontier]
    (axis cs:0,{2*\Vscale}) 
    -- (axis cs:0.333333,{2*\Vscale})
    -- (axis cs:0.333333,{3*\Vscale}) 
    -- (axis cs:0.666667,{3*\Vscale})
    -- (axis cs:0.666667,{1*\Vscale}) 
    -- (axis cs:1,{1*\Vscale});

\draw[black!35, densely dotted]
    (axis cs:\pstar,{0.65*\Vscale}) -- (axis cs:\pstar,{3.25*\Vscale});

\draw[black!35, densely dotted]
    (axis cs: 0 , 2*\Vscale) -- (axis cs: 1 , 2*\Vscale);
    
\addplot[
    point,
    mark options={draw=black,fill=black}
] coordinates {(0,{2*\Vscale})};

\addplot[
    point,
    mark options={draw=black,fill=mutwocol}
] coordinates {(\mup,{3*\Vscale})};

\addplot[
    point,
    mark options={draw=black,fill=muthreecol}
] coordinates {(\mum,{1*\Vscale})};

\node[
    font=\scriptsize,
    align=center,
    text=mutwocol,
    inner sep=1pt
] at (axis cs:0.46,{2.55*\Vscale})
    {closer\\favorable\\posterior};

\node[
    font=\scriptsize,
    align=center,
    text=muthreecol,
    inner sep=1pt
] at (axis cs:0.85,{1.4*\Vscale})
    {farther\\unfavorable\\posterior};

\draw[
    decorate,
    decoration={brace,amplitude=3.5pt},
    black!55
]
    (axis cs:\mup,{3.08*\Vscale}) -- (axis cs:\mum,{3.08*\Vscale});

\node[
    font=\scriptsize,
    align=center,
    text=black
] at (axis cs:0.7,{3.28*\Vscale})
    {zero covariance};

\draw[
    decorate,
    decoration={brace,amplitude=3.5pt},
    black!55
]
    (axis cs:0,{3.33*\Vscale}) -- (axis cs:0.6667,{3.33*\Vscale});

\node[
    font=\scriptsize,
    align=center,
    text=black
] at (axis cs:0.38,{3.53*\Vscale})
    {Bayes plausibility};

\end{axis}
\end{tikzpicture}

\vspace{0.35em}
{\small (b) Construction of a strictly improving mediation plan}
\end{minipage}

\par\medskip
\begin{minipage}{0.93\textwidth}
\small
\emph{Notes.}
Panel (a) illustrates the seller's optimal payoff from Bayesian persuasion (blue dashed), mediation (orange solid), and cheap talk (red dashed, coincides with the upper envelope of $\BV$). Panel (b) illustrates the three-posterior mediation construction in the text with $p \in (0,1/3)$.
\end{minipage}

\end{figure}

Mediated communication relaxes this constraint. The seller need not be indifferent across all realized posteriors, but only across reports to the platform. Theorem \ref{Thm: Implementability} shows that a mediated outcome can be represented by a distribution over buyer beliefs and seller values $(\mu,v)$ with $v\in \BV(\mu)$. It satisfies Bayes plausibility, and the seller's incentive constraint becomes
\[
    \Cov[\mu, v] = \mathbb{E}\big[(\mu-p)v\big]=0 .
\]
Intuitively, state-1 reports make high posteriors (on $\w=1$) more likely, while state-0 reports make low posteriors more likely. If the seller's value $v$ were systematically higher at high posteriors, the seller would strictly prefer the state-1 report; if it were lower, the seller would prefer the state-0 report. Truth-telling therefore requires $v$ to be uncorrelated with $\mu$, but not constant across posteriors.

This additional flexibility is valuable when the platform can use an unfavorable posterior $\mu^-$ to discipline a favorable one $\mu^+$. In this example, since the sender's payoff under cheap talk is $2$, it is enough to find \((\mu^+,v^+)\) and $(\mu^-, v^-)$ in the graph of $\BV$ such that
\[
    v^+> 2 >v^-
    \qquad\text{and}\qquad
    \mu^+\in(p,\mu^-).
\]
The first condition allows gains and losses to offset each other in the seller's incentive constraint. The second makes it possible to satisfy zero covariance while placing relatively more weight on the favorable posterior, because it is closer to the prior \(p\). 

In words, the platform sometimes sends messages that reduce expected demand for the seller, i.e., they induce $\m^-$.
These unfavorable messages are the credibility cost that makes
truthful reporting incentive-compatible. With this, the platform
can also send favorable messages that cheap talk could not sustain, i.e., they induce $\m^+$. The posterior at $0$ restores Bayes plausibility. Mediation is valuable
exactly when the favorable effect dominates the credibility cost, that is, when $\mu^+\in(p,\mu^-)$. In this case, 
 relative to direct communication, the seller has a strictly positive ex-ante willingness to pay for access to the platform.\footnote{For simplicity, we assume that the seller chooses whether to use the platform before observing \(\w\). Under transparent motives, however, the highest sender's equilibrium payoff in the augmented game where the sender decides whether to hire a mediator after observing \(\w\) coincides with the sender's highest communication-equilibrium payoff. See \cite{KS2024infomediated}.} The construction is illustrated in Figure \ref{Fig: illustrative}(b).

Beyond the binary-state example, Theorem \ref{Thm: MD vs CT} shows that such one-dimensional perturbations characterize when mediation strictly benefits the sender in general. In binary-state environments, a similar construction also delivers the sender-optimal mediation plan (Proposition \ref{prop: opt MD binary}). In this example, the optimum is attained by a distribution supported on posteriors \(\{0,\mu^+=1/3,\mu^-= 1\}\). The posterior $\mu^+=1/3$ is the closest posterior that raises the payoff above 2; moving farther right does not increase demand, but makes the incentive constraint harder to satisfy. The posterior $\mu^-=1$ is the farthest posterior with a payoff below $2$, and is therefore the most effective credibility cost.

Now consider persuasion, where the platform can verify the product's state and hence implement any Bayes-plausible distribution of posteriors. For $p\in (0, 1/3)$, the optimal verifiable experiment mixes between posteriors $0$ and $1/3$, strictly improving on cheap talk. Without verification, this is not implementable because there is a positive covariance between the posterior and seller value. Conversely, if the optimal verifiable experiment does not exploit the covariance, as for $p\in [1/3,2/3]$, cheap talk can attain the same payoff. This illustrates Theorem \ref{Thm: BP = MD iff BP = CT}: commitment to an experiment strictly improves on cheap talk if and only if verifiability is valuable.

Mediation can also benefit buyers. Aggregate buyer utility is convex and piecewise linear in the common posterior. Seller-optimal mediation induces posteriors supported on $\{0,1/3,1\}$, while an informative seller-optimal cheap-talk equilibrium induces posteriors supported on $\{0,1/3\}$. The former is a mean-preserving spread of the latter, so relative to this equilibrium, mediation benefits both sides. For $p \in (0,1/6)$, mediation benefits buyers relative to every cheap-talk equilibrium. Section \ref{ssec:accepta_pareto} extends this Pareto observation to a general class of games with a privately informed receiver; see Proposition \ref{pro:accep_Pareto}.

\section{The Model}
\label{sec:model}
We present our model and main results in the context of mediated communication.
In Section \ref{sec:matching}, we rephrase the model as matching with externalities and aligned preferences, and use our results to study the efficiency-fairness tradeoff therein.

We consider sender-receiver games under transparent motives. The sender is privately informed about the state $\w \in \W$ drawn from a full-support prior $p\in\D(\Omega)$, where $\W$ is finite with $|\W|= n$.
The receiver is uninformed about $\w$ and takes an action $a \in A$. The sender has a \emph{state-independent} utility function $u_S: A \to \R$ and the receiver has utility $u_R: \W \times A \to \R$. The action set $A$ is a compact metric space, and $u_S$ and $u_R$ are continuous. The receiver can be interpreted as either an individual player or as multiple players with common information about the state, separate actions, and no strategic externalities (e.g., Sections \ref{Subsec: Illustrative Example} and \ref{ssec:accepta_pareto}).\footnote{We consider finitely many states only for simplicity, and we show in Appendix \ref{Sec: infinite state} that Theorem \ref{Thm: Implementability}, Theorem \ref{Thm: BP = MD iff BP = CT}, Point (i) of Theorem \ref{Thm: MD vs CT}, and Theorem \ref{Thm: multidim quasiconvex} hold when $\W$ is an arbitrary compact metric space. The transparent-motive assumption has bite, and we defer its discussion to Section \ref{sec: model interpretation}.}

The sender and receiver communicate through an uninformed mediator, who commits to a communication mechanism $\sigma: M_S \to \D(M_R)$. $M_S$ is the reporting space for the sender, and $ M_R$ is the space of messages for the receiver; both are rich enough. After observing $\w$, the sender sends a report $m_S\in M_S$ to the mediator. Given the report, the receiver observes a message $m_R \in M_R$ drawn according to $\sigma$ and takes an action $a \in A$. 

 We consider the communication game induced by $\sigma$ and focus on its Bayes-Nash equilibria (BNE), also known as the \emph{communication equilibria} (see \cite{myerson1982optimal} and  \cite{forges1986approach}).
A mechanism $\sigma$ and a communication equilibrium induce an outcome distribution $\pi \in \D(\W\times A)$. The Revelation Principle \citep{myerson1982optimal, forges1986approach} implies that 
    an outcome distribution $\pi$ is induced by some communication equilibrium if and only if:
    \begin{itemize}
        \item [(i)] Consistency: $ \marg_{\W}{\pi}=p$
 
        \item [(ii)]  Obedience: For $\pi$-almost all $a \in A$, $\mathbb{E}_{\pi^a}[u_R(\w,a)]=\max_{a' \in A} \mathbb{E}_{\pi^a}[u_R(\w,a')]$,
         where $\pi^a \in \D(\W)$ is a version of the conditional probability given $a\in A$;
        \item [(iii)]  Honesty: For all $\w,\w' \in \W$, 
        $\mathbb{E}_{\pi^{\w}}[u_S(a)] \ge \mathbb{E}_{\pi^{\w'}}[u_S(a)]$, where $\pi^{\w} \in \D(A)$ is the conditional probability given $\w\in \W$.
    \end{itemize}
We say that $\pi$ is a communication equilibrium (CE) outcome if (i), (ii), and (iii) hold. In this paper, we select the CE outcome that maximizes the sender's ex-ante payoff.

\paragraph{Cheap talk and Bayesian persuasion}
We compare mediated communication with cheap talk (i.e., unmediated communication) and Bayesian persuasion.

Under cheap talk, there is a sufficiently rich message space $M$ and the following timing: First, the sender observes $\w$ and sends a message $m\in M$ to the receiver, then the latter takes an action $a$. No player can commit to a contingent strategy, and we select the BNE that maximizes the sender's ex-ante payoff. The equilibrium outcomes are those that satisfy (i), (ii), and a strengthening of (iii) where the sender has no strict incentive to deviate from sending the equilibrium message.\footnote{See \cite{CS82} and \cite{LR20} for the formal definition of BNE under cheap talk.}

Under Bayesian persuasion, before observing the state, the sender designs an experiment $\sigma : \W \to \Delta(M)$ that takes as input the true state and commits to revealing its realization $m$ to the receiver. The receiver takes an action $a$ after observing $m$. We select the BNE of this game that maximizes the sender's ex-ante payoff. The feasible equilibrium outcomes are those that satisfy (i) and (ii) without any incentive compatibility requirement for the sender.

The sender weakly prefers persuasion to mediation, and mediation to cheap talk. We call the gaps in the sender's payoff respectively the \emph{value of verifiability} and the \emph{value of mediation}. We say that verifiability, respectively mediation, is \emph{valuable} at $p$ when the corresponding gap is strictly positive. We next present our main results and then, in Section \ref{sec: model interpretation}, discuss the interpretation of the model and its assumptions.

\section{Belief-based Approach to Mediation}
\label{sec:imple}
In this section, we characterize the feasible distributions of receiver beliefs and sender values under mediation, as well as the sender's maximum payoff.
We start by defining the interim value correspondence  $\mathbf{V}: \D(\W) \rightrightarrows \R$ by
        \begin{equation}
        \label{eq:value_corresp}
            \mathbf{V}(\m):=\co\left(u_S \left(\argmax_{a \in A}\mathbb{E}_{\m}[u_R(\w,a)]\right )\right).
        \end{equation}
For every posterior $\mu \in \D(\W)$, the set $\mathbf{V}(\m)$ collects all the possible (expected) sender payoffs that can be attained by some (potentially mixed) best response of the receiver at posterior $\m$.\footnote {It is standard to show that $\mathbf{V}$ is a Kakutani correspondence. Alternatively, we can take any Kakutani correspondence $\BV$ as a primitive capturing the set of the sender's continuation values given each posterior. This setting is strictly more general than the one presented here, and we will exploit it in Section \ref{ssec:accepta_pareto}.} Define the graph of this correspondence $\Gr(\mathbf{V}) \subseteq \Delta(\W) \times \mathbb{R}$ and  the functions $\oV(\m)=\max\BV(\m)$ and $\uV(\m)=\min\BV(\m)$.

Any CE outcome $\pi$ induces a distribution $\eta^\pi \in \Delta(\Delta(\W) \times \mathbb{R})$ over pairs $(\m,v) \in \D(\W) \times \mathbb{R}$ of receiver posterior beliefs and sender expected values. Formally,
\begin{equation*}
    \eta^\pi (S) = \int_{\W\times A} \mathbb{I}[(\pi^a, u_S(a))\in S] \de \pi(\w,a)
\end{equation*}
for every Borel measurable $S \subseteq \Delta(\W) \times \mathbb{R}$.
\begin{definition}   \label{def: induced by CE}
   A distribution of posteriors and values $\eta\in \D(\D(\W)\times \R)$ is \emph{induced by} some CE outcome $\pi \in \D(\W \times A)$ if $\eta = \eta^\pi$. 
\end{definition}

Our first result characterizes implementable distributions over pairs of posteriors and interim values in terms of conditions parallel to Consistency, Obedience, and Honesty. 

\begin{theorem} \label{Thm: Implementability}
If a distribution of receiver beliefs and sender values $\eta \in \D(\D(\W)\times \R)$ is induced by some CE outcome, then it satisfies
        \begin{itemize}
        \item [(i)]  Consistency*:
            \begin{equation}
            \label{Eq: Bayes-plausibility}
                \mathbb{E}_{\eta}[\m]=p;
                \tag{BP}
            \end{equation}
            \item [(ii)] Obedience*:
            \begin{equation}
            \label{Eq: Obedience}
                \eta(\Gr(\mathbf{V}))=1;
                \tag{OB}
            \end{equation}
            \item [(iii)] Honesty*:
            \begin{equation}
            \label{eq:zero_cov}
                \Cov_{\eta}[v,\m] =\mathbf{0}, \tag{zeroCov}
            \end{equation}  
            that is, $\Cov_\eta[v,\mu(\w)] = 0$ for every $\w \in \W$.
        \end{itemize}
Conversely, if $\eta$ satisfies (i), (ii), and (iii), then there exists a CE outcome $\pi \in \D(\W \times A)$ such that $\mathbb{E}_{\eta}[v] = \mathbb{E}_{\pi}[u_S]$.
\end{theorem}

The zero-covariance condition \eqref{eq:zero_cov} has a simple incentive interpretation. Conditional on $\w$, Bayes' rule tilts the distribution over induced posteriors and values toward those that put higher probability on $\omega$: $\de \eta^{\w}(\mu, v) = \frac{\mu(\w)}{p(\w)}\de \eta(\mu, v)$. Thus, if $\Cov_\eta[v,\mu(\omega)]>0$ under the unconditional distribution $\eta$, then the outcomes that are more likely after report $\w$ are also the outcomes that are better for the sender. Equivalently, report $\w$ gives the sender an above-average continuation payoff:
\begin{align} \label{Eq: alternative zerocov}
    \Cov_{\eta}[v,\mu(\w)] = \mathbb{E}_{\eta}[v\mu(\w)] - \mathbb{E}_\eta[v] \mathbb{E}_{\eta}[\mu(\w)] = p(\w) (\mathbb{E}_{\eta^\w}[v] - \mathbb{E}_\eta[v]).
\end{align}
Since the ex ante payoff is the $p$-weighted average of the payoffs from all reports, some other report $\omega'$ must then give a lower payoff. The sender in state $\omega'$ would prefer to report $\omega$, violating Honesty. Therefore, Honesty requires zero covariance for every $\omega\in\Omega$.

Applying our Theorem \ref{Thm: Implementability}, we can rewrite the mediator's problem in the belief space. The mediator chooses $\eta \in \D(\D(\W)\times \R)$ to maximize the sender's expected payoff:
\begin{align} \label{Eq: MD problem}
    \VV_{MD}(p) \coloneqq \sup_{\eta \in \D(\D(\W)\times \R)} & \int_{\D(\W)\times \R} v \, \de \eta(\mu,v) \tag{MD}\\
    \textnormal{subject to:} &\; \eqref{Eq: Bayes-plausibility}, \eqref{Eq: Obedience}, \eqref{eq:zero_cov}. \notag
\end{align}

\begin{Remark}\label{Rmk: Existence}
    \eqref{Eq: MD problem} admits a solution $\eta^*$ supported on no more than $2n-1$ points. The solution exists since the feasible set is compact and the objective function is continuous in the weak topology. The constraints \eqref{Eq: Bayes-plausibility} and \eqref{eq:zero_cov} are in the form of moment conditions \`a la \cite{Win88}, which implies that optimal mediation can be achieved with no more than $2n-1$ messages. We discuss the technical details in Appendix \ref{App: existence and value}.
\end{Remark}

Bayesian persuasion and cheap talk can also be analyzed via the belief-based approach. The sender's preferred Bayesian persuasion value $\VV_{BP}(p)$ is obtained by maximizing over $\eta \in \D(\D(\W)\times \R)$ that satisfies \eqref{Eq: Bayes-plausibility} and \eqref{Eq: Obedience}, without the zero-covariance constraint. Under cheap talk, the sender's highest equilibrium payoff $\VV_{CT}(p)$ is obtained by maximizing over $\eta \in \D(\D(\W)\times \R)$ that satisfies \eqref{Eq: Bayes-plausibility} and \eqref{Eq: Obedience}, while replacing \eqref{eq:zero_cov} by the stronger zero-variance constraint $\Var_\eta[v] = 0$.\footnote{This zero-variance condition is equivalent to the requirement that the sender's value is constant over the support of the distribution of receiver beliefs, that is, the condition in \cite{LR20}.} This is equivalent to requiring that there is no stochastic dependence between $v$ and $\m$, a stricter condition than the zero-covariance constraint that in turn only imposes no \emph{linear} dependence between $v$ and $\m$.
\begin{Remark}[Analysis of variance]
\label{rk: anova}
Given a distribution of beliefs and values \(\eta\) satisfying \eqref{Eq: Bayes-plausibility} and \eqref{Eq: Obedience}, the law of total variance gives
\[
    \Var_\eta[v]
    =
    \Var_p\!\left[\mathbb{E}_{\eta^{\omega}}[v]\right]
    +
    \mathbb{E}_p\!\left[\Var_{\eta^{\omega}}[v]\right].
\]
The first term is the variation in the sender's expected value across state reports; the second is the average variation in value induced by the mediator after a report. Cheap talk imposes \(\Var_\eta[v]=0\). Mediation instead imposes only $\Var_p\!\left[\mathbb{E}_{\eta^{\omega}}[v]\right]=0$,
so the sender's value cannot vary across reports, but it may vary across the receiver's realized posteriors.
\end{Remark}

Thus, mediation can improve on cheap talk only by creating payoff dispersion $\Var_{\eta^{\omega}}[v]>0$ that is orthogonal to reporting incentives: some posterior realizations give the sender more than the cheap-talk value, others give less, and the mediator balances them so that no type wants to misreport. The next section formalizes this idea through \emph{directional improvability}, which characterizes when such offsetting variation can be constructed.

\section{Comparison of communication protocols}
\label{sec:main_results}
We start with a fundamental result that leverages the assumption of transparent motives to simplify the comparison of persuasion, mediation, and cheap talk in the rest of the section. While the following result is simple in nature, we shall see that it has far-reaching consequences, even beyond communication in games (see Section \ref{sec:matching}).

\begin{theorem} \label{Thm: BP = MD iff BP = CT} Verifiability is valuable at $p$ if and only if commitment is valuable at $p$, that is,
$\VV_{BP}(p) > \VV_{MD}(p) $ if and only if $\VV_{BP}(p) > \VV_{CT}(p)$.
\end{theorem} 
The only if direction is immediate. For the if direction, fix a feasible distribution under mediation and perturb its probabilities in proportion to the payoff deviations from the mean. This puts more weight on above-average values and less weight on below-average values. Condition \eqref{eq:zero_cov} says that this redistribution does not affect the posterior mean. Such a perturbation would raise the sender's payoff unless the sender's payoff is constant almost surely, yielding the desired result.

\paragraph{Full-dimensionality.}
We introduce a local full-dimensionality condition for cheap talk that (i) holds generically over priors in finite games, and (ii) when it does, it significantly simplifies the comparison of communication protocols.

Let $\oV_{CT},\uV_{CT}:\D(\W)\to \R$ respectively denote the quasiconcave  envelope of $\oV$ and quasiconvex envelope of $\uV$, and define the cheap-talk correspondence by $\BV_{CT}(\mu)=[\uV_{CT}(\m),\oV_{CT}(\m)]$. The sender's value under a cheap-talk equilibrium lies in $\BV_{CT}$.\footnote{See \cite{LR20} Appendix C.2.1, which defines the quasiconcave and quasiconvex envelopes with an extra semi-continuity assumption. Our definition is the same since our state space $\W$ is finite. Alternatively, as shown in \cite{aumann2003long} and \cite{LR20}, $\BV_{CT}$ is the correspondence whose graph coincides with the di-convexification of the graph of $\BV$.} For all $\mu, \mu' \in \D(\W)$, let $(\mu, \mu'] \coloneqq \{ (1-t) \mu + t \mu' : t \in (0,1]\}$ denote the half-open line segment and $(\mu, \mu')$ denote the corresponding open line segment.

\begin{definition}
\label{def:CT_hull}
    The cheap talk hull at $p$ is defined as
    \begin{equation*}\label{eq: cheap talk hull}
        H^*(p) \coloneqq \{\mu \in \D(\W): \exists \mu_0 \in \D(\W) \text{ such that }  \oV_{CT}(p) \in \BV_{CT}(\mu_0) \text{ and } p\in (
        \mu_0,\mu]\}.
    \end{equation*}
     The full-dimensionality condition holds at $p$ if $H^*(p)=\D(\W)$. In particular, $p\in H^*(p)$.
\end{definition}
The set $H^*(p)$ collects beliefs $\mu$ such that the cheap-talk value at $p$ remains attainable when the prior is slightly perturbed in the opposite direction of $\m$. 
Thus, the full-dimensionality condition requires that $\oV_{CT}(p)$ can still be attained when the prior is slightly perturbed toward \emph{any} arbitrary direction.

\begin{Remark}\label{Lem: full dim cond}
The full-dimensionality condition trivially holds at $p$ if $\oV_{CT}$ is constant around $p$. When the state is binary, this happens when $\oV_{CT}(p)>\oV(p)$, that is, when direct communication is valuable at $p$. For an arbitrary state space, if the action set $A$ is finite, then $\oV_{CT}$ is locally constant for almost every prior (Corollary 2, \cite{LR20}), hence the full-dimensionality condition holds generically.
\end{Remark}

\subsection{Valuable verifiability}
\label{SEC:BP vs MD}
Theorem \ref{Thm: BP = MD iff BP = CT} immediately yields that verifiability is valuable if and only if the concave and quasiconcave envelopes of the sender's value function do not coincide at the prior. Therefore,  if the sender cannot achieve the optimal persuasion value using single-round cheap talk, then they cannot attain it via any communication mechanism without sender commitment (e.g., multiple-round cheap talk, noisy cheap talk, or in general, mediation). 

When $\W$ is non-binary, comparing the concave envelope and the quasiconcave envelope is not always easy. Thus, we take a constructive approach and provide a sufficient condition for persuasion to strictly outperform mediation, which also becomes necessary when the full-dimensionality condition holds.

\begin{proposition}\label{Prop: sufficient cond for BP > MD}
    The following hold:
\begin{enumerate}
    \item If $\max_{\mu \in  H^*(p)}\oV_{CT}(\mu) > \oV_{CT}(p)$, then verifiability is valuable at $p$.
    \item If verifiability is valuable at $p$ then  $\max\oV > \oV_{CT}(p)$.
\end{enumerate}
 Moreover, if the full-dimensionality condition holds at $p$, then verifiability is valuable at $p$ if and only if $\max\oV>\oV_{CT}(p)$.
\end{proposition}

This result establishes that verifiability is valuable under very weak conditions. In fact, it suffices to show that there exists a cheap-talk equilibrium at some prior $\mu  \in H^*(p)$ that the sender strictly prefers to the optimal cheap-talk equilibrium with prior $p$. This condition is particularly relevant in finite games.

\begin{corollary}
    Assume that $A$ is finite. For almost all priors $p$, verifiability is valuable at $p$ if and only if $p \notin \co \argmax \oV$.
\end{corollary}

The proof of Proposition \ref{Prop: sufficient cond for BP > MD}  is constructive and relies on Theorem \ref{Thm: BP = MD iff BP = CT}. Starting from $\mu$, one can construct a distribution $\eta \in \D(\Gr(\mathbf{V}))$ whose average over beliefs is $p$ (due to $\mu \in H^*(p)$) and whose average over values is larger than  $\oV_{CT}(p)$ (due to $\oV_{CT}(\mu) > \oV_{CT}(p)$). This shows that commitment is valuable, hence, that verifiability is valuable.

\subsection{Valuable mediation} \label{Sec: MD vs CT}
This section provides separate sufficient and necessary conditions for the mediator to strictly outperform direct communication. These conditions collapse under full dimensionality, yielding a tight geometric characterization of when mediation is valuable.
\begin{definition}
\label{def:improv}
    Cheap talk is directionally improvable at $p$ if there exist $(\mu^+,v^+), (\mu^-,v^-) \in \Gr(\mathbf{V}_{CT})$ such that 
    \begin{equation*}
       \m^+ \in (p,\m^-) \qquad \text{and} \qquad  v^+>\oV_{CT}(p)>v^-.
    \end{equation*}
    If in addition, $\mu^- \in H^*(p)$, then cheap talk is hull-directionally improvable at $p$.
\end{definition}

Geometrically, Definition \ref{def:improv} amounts to the failure of a form of \emph{weak single crossing} of the cheap-talk correspondence $\mathbf{V}_{CT}$ at $\oV_{CT}(p)$ along a line segment through $p$. Reading the segment from the side opposite to $\mu^-$ toward $\mu^-$, the cheap-talk correspondence is first weakly below $\oV_{CT}(p)$, then strictly above it at $\mu^+$, and finally strictly below it at $\mu^-$. Thus, the above-benchmark payoffs arise in the middle of the segment rather than on one side of $p$. This formalizes the countervailing effects that mediation exploits: the closer belief $\mu^+$ raises the sender's payoff relative to $\oV_{CT}(p)$, while the farther belief $\mu^-$ provides the offset needed for incentive compatibility. The next result shows that countervailing effects allow us to construct a strictly improving mediation plan.

\begin{theorem} \label{Thm: MD vs CT}
The following hold:
\begin{enumerate}
    \item If cheap talk is hull-directionally improvable at $p$, then mediation is valuable at $p$. 
    \item If mediation is valuable at $p$, then cheap talk is directionally improvable at $p$.
\end{enumerate}
 Moreover, if the full-dimensionality condition holds at $p$, then mediation is valuable at $p$ if and only if cheap talk is directionally improvable at $p$.
\end{theorem}
Theorem \ref{Thm: MD vs CT} reduces the comparison between mediation and cheap talk with arbitrarily many states to a simple one-dimensional problem, \emph{as if} the state is binary. Similarly to Proposition \ref{Prop: sufficient cond for BP > MD}, this theorem yields the following immediate corollary for finite games.
\begin{corollary} \label{Cor: MD vs CT finite}
        Assume that $A$ is finite. For almost all priors $p$, mediation is valuable at $p$ if and only if cheap talk is directionally improvable at $p$.
\end{corollary}

It is natural to ask whether, under the sufficient condition of Theorem \ref{Thm: MD vs CT}, mediation also strictly improves the expected utility of the receiver. This is indeed the case provided that the sender's and receiver's indirect payoffs are aligned enough in the sense that the receiver's indirect utility is an increasing and convex transformation of $\overline{V}$. In Section \ref{ssec:accepta_pareto}, we analyze a class of games satisfying this condition.

\paragraph{Intuition of Theorem \ref{Thm: MD vs CT}.}
Suppose first that cheap talk is
hull-directionally improvable at \(p\). Then there are beliefs
\(\mu^- \in H^*(p)\) and \(\mu^+\in(p,\mu^-)\), together with cheap-talk
equilibria at priors \(\mu^+\) and \(\mu^-\), that yield sender payoffs
\(v^+>\oV_{CT}(p)>v^-\), respectively. Since \(\mu^-\in H^*(p)\), there is
also a belief \(\mu_0\) on the opposite side of \(p\) and a cheap-talk
equilibrium at prior \(\mu_0\) that yields a payoff \(\oV_{CT}(p)\).

The mediator randomizes over these three cheap-talk equilibria. The benchmark equilibrium at $\mu_0$ restores Bayes plausibility. The unfavorable equilibrium at $\mu^-$ provides the incentive discipline needed to sustain the favorable equilibrium at $\mu^+$. Since \(\mu^+\) lies between \(p\) and \(\mu^-\), the unfavorable equilibrium has more leverage in the covariance constraint: it can discipline the favorable equilibrium at an expected-payoff cost smaller than the gain generated by the favorable equilibrium. Hence, the mediator can choose weights so that the zero-covariance constraint is satisfied while the expected payoff remains strictly above $\oV_{CT}(p)$. 

Conversely, suppose that mediation is valuable at $p$. Then any improving mediation
plan must sometimes generate sender values strictly above \(\oV_{CT}(p)\).
The zero-covariance constraint implies that these high-value realizations
cannot stand alone: they must be balanced by realizations with values strictly
below \(\oV_{CT}(p)\). Weighting the high-value realizations by their gains and the low-value realizations by their losses, zero covariance implies that the two weighted average beliefs lie on a common line through $p$, with the high-value average closer to $p$.
Thus an improving mediation plan necessarily reveals a direction along
which the cheap-talk correspondence crosses the level \(\oV_{CT}(p)\) from
above to below. This is exactly directional improvability. Therefore, if cheap
talk is not directionally improvable at \(p\), no mediation plan can improve on
cheap talk.

Finally, under full dimensionality,
directional improvability and hull-directional improvability coincide, and the
sufficient and necessary conditions collapse to an equivalence.

\subsection{Optimal mediation with binary states} \label{ssec: opt MD binary state}
When cheap talk is directionally improvable, the construction in the proof of Theorem \ref{Thm: MD vs CT} yields a communication mechanism that outperforms any cheap-talk equilibrium. Even though this construction need not yield the optimal communication mechanism in general, we next show that a similar one does in the binary-state setting. This is intuitive since the idea of our construction is to convert the problem from a multidimensional one to a one-dimensional one and then find an improving perturbation. When the environment is one-dimensional to begin with, this approach yields optimality overall.

Let $\W = \{\w_0,\w_1\}$ be binary and let mediation strictly improve on cheap talk, which implies that cheap talk is directionally improvable at $p$. We let $\mu \in [0,1]$ denote the posterior probability of $\w_1$. Throughout, we assume without loss that 
\[
\oV(\mu) \leq \oV_{CT}(p) \qquad \textrm{for all } \mu \in [0, p].
\]
We next provide a characterization of the solution of \eqref{Eq: MD problem}. To improve on cheap talk, mediation must sometimes generate values above $\oV_{CT}(p)$. The truth-telling constraint then requires these high-value realizations to be balanced by lower-value realizations to ensure zero covariance. In the binary-state case, this tradeoff can be represented using three posterior-value pairs $\left\{(\mu_i,v_i)\right\}_{i=1}^3$ ordered as $\mu_1<p<\mu_2<\mu_3$.

The point $(\mu_2,v_2)$ delivers the value above $\oV_{CT}(p)$, while $(\mu_3,v_3)$ delivers the lower value needed to satisfy the truth-telling constraint. The point $(\mu_1,v_1)$ provides the left-side posterior needed for Bayes plausibility. Optimality pushes these selected values to the relevant boundary of $\Gr(\BV)$: at $\mu_1$ and $\mu_2$, the selected values are maximal, $v_i=\oV(\mu_i)$ for $i=1,2$; at $\mu_3$, the selected value is minimal, $v_3=\uV(\mu_3)$. Moreover, the support points can be chosen to satisfy the corresponding frontier conditions: no point further left than $\mu_1$ reaches $v_1$ on the upper envelope, no point between $p$ and $\mu_2$ reaches $v_2$ on the upper envelope, and no point to the right of $\mu_3$ falls below $v_3$ on the lower envelope.

\begin{proposition} \label{prop: opt MD binary}
\eqref{Eq: MD problem} admits a solution $\eta^*$ with
\begin{equation*}
    \supp(\eta^*)=\{(\mu_1,\oV(\mu_1)),(\mu_2,\oV(\mu_2)),(\mu_3,\uV(\mu_3))\}
\end{equation*}
where
\[
\mu_1<p<\mu_2<\mu_3,\qquad 
\uV(\mu_3) < \oV(\mu_1) \le \oV_{CT}(p)< \oV(\mu_2),
\]
and
\[
\oV(\mu)<\oV(\mu_1)\ \ \forall\,\mu<\mu_1,\qquad
\oV(\mu)<\oV(\mu_2)\ \ \forall\,\mu\in(p,\mu_2),\qquad
\uV(\mu)>\uV(\mu_3)\ \ \forall\,\mu>\mu_3.
\]
\end{proposition}

The frontier condition intuitively follows from \eqref{eq:zero_cov}. The favorable posterior $\m_2$ raises the sender's payoff, but also induces a positive covariance between beliefs and payoffs; for a given payoff gain, this effect is smaller when $\m_2$ is closer to the prior. The unfavorable posterior $\m_3$ supplies the offsetting negative covariance; for a given payoff loss, this offset is stronger when $\m_3$ is farther from the prior. $\m_1$ restores Bayes plausibility; placing it farther from the prior requires less mass on it. In the example of Section \ref{Subsec: Illustrative Example}, these forces select a unique triple $\m_2=1/3$, $\m_3=1$, and $\m_1=0$, see Figure \ref{fig:binary-frontier-panels}(a).

\begin{figure}[t]
\centering
\caption{
Optimal Mediation with Binary States
}
\label{fig:binary-frontier-panels}

\definecolor{mutwocol}{RGB}{190,45,45}     
\definecolor{muthreecol}{RGB}{55,105,180}  

\tikzset{
    frontier/.style={black,line width=1.08pt,line cap=round,line join=round},
    guide/.style={black!30,densely dotted,line width=0.65pt},
    point/.style={only marks, mark=*, mark size=2.35pt},
    hintone/.style={draw=black,line width=0.95pt,-stealth},
    hinttwo/.style={draw=mutwocol,line width=0.95pt,-stealth},
    hinttwovert/.style={draw=mutwocol,line width=0.95pt,-stealth},
    hintthree/.style={draw=muthreecol,line width=0.95pt,-stealth}
}

\begin{minipage}[t]{0.48\textwidth}
\centering
\begin{tikzpicture}[>=stealth]
\begin{axis}[
    width=\textwidth,
    height=5.05cm,
    xmin=0, xmax=1.07,
    ymin=-1.85, ymax=3.35,
    axis x line*=bottom,
    axis y line*=left,
    axis line style={black,-{Latex[length=1.8mm]}},
    xtick=\empty,
    ytick={-1,1,3},
    yticklabels={$1$,$2$,$3$},
    tick style={draw=none},
    yticklabel style={font=\scriptsize},
    clip=false
]

\node[font=\scriptsize, anchor=south, rotate=0]
    at (axis description cs:0,1) {$\BV(\mu)$};
\node[font=\scriptsize, anchor=west]
    at (axis description cs:1,0) {$\mu$};

\draw[guide] (axis cs:0,1) -- (axis cs:1.02,1);

\draw[guide] (axis cs:0.1667,-1.75) -- (axis cs:0.1667,3);
\node[font=\scriptsize, anchor=north, inner sep=1pt]
    at (axis cs:0.1667,-2) {$p$};

\draw[frontier]
    (axis cs:0,1)       -- (axis cs:0.33,1)
    -- (axis cs:0.33,3) -- (axis cs:0.66,3)
    -- (axis cs:0.66,-1) -- (axis cs:1,-1);

\addplot[
    point,
    mark options={draw=black,fill=black}
] coordinates {(0,1)};

\addplot[
    point,
    mark options={draw=black,fill=mutwocol}
] coordinates {(0.33,3)};

\addplot[
    point,
    mark options={draw=black,fill=muthreecol}
] coordinates {(1,-1)};

\node[font=\scriptsize, anchor=north, inner sep=1pt]
    at (axis cs:0.02,-2) {$\mu_1$};
\node[font=\scriptsize, anchor=north, inner sep=1pt]
    at (axis cs:0.33,-2) {$\mu_2$};
\node[font=\scriptsize, anchor=north, inner sep=1pt]
    at (axis cs:1.02,-2) {$\mu_3$};

\draw[hintone]   (axis cs:0.125,1.2) -- (axis cs:0.025,1.2);
\draw[hinttwo]   (axis cs:0.46,3.2) -- (axis cs:0.36,3.2);
\draw[hinttwovert] (axis cs:0.36,1.8) -- (axis cs:0.36,2.8);
\draw[hintthree] (axis cs:0.86,-0.8) -- (axis cs:0.96,-0.8);

\end{axis}
\end{tikzpicture}

\vspace{0.25em}
{\small (a) Frontier conditions pin down the triple}
\end{minipage}
\hfill
\begin{minipage}[t]{0.48\textwidth}
\centering
\begin{tikzpicture}[>=stealth]
\begin{axis}[
    width=\textwidth,
    height=5.05cm,
    xmin=0, xmax=1.07,
    ymin=-1.85, ymax=2.50,
    axis x line*=bottom,
    axis y line*=left,
    axis line style={black,-{Latex[length=1.8mm]}},
    xtick=\empty,
    ytick={-1,0,1,2},
    yticklabels={$1$,$2$,$3$,$4$},
    tick style={draw=none},
    yticklabel style={font=\scriptsize},
    clip=false
]

\node[font=\scriptsize, anchor=south, rotate=0]
    at (axis description cs:0,1) {$\BV(\mu)$};
\node[font=\scriptsize, anchor=west]
    at (axis description cs:1,0) {$\mu$};

\draw[guide] (axis cs:0,0) -- (axis cs:1.02,0);

\draw[guide] (axis cs:0.125,-1.75) -- (axis cs:0.125,2);
\node[font=\scriptsize, anchor=north, inner sep=1pt]
    at (axis cs:0.125,-2) {$p$};

\draw[frontier]
    (axis cs:0,0)        -- (axis cs:0.25,0)
    -- (axis cs:0.25,1)  -- (axis cs:0.50,1)
    -- (axis cs:0.50,2)  -- (axis cs:0.75,2)
    -- (axis cs:0.75,-1) -- (axis cs:1,-1);

\addplot[
    point,
    mark options={draw=black,fill=black}
] coordinates {(0,0)};

\addplot[
    point,
    mark options={draw=black,fill=mutwocol}
] coordinates {(0.25,1) (0.50,2)};

\addplot[
    point,
    mark options={draw=black,fill=muthreecol}
] coordinates {(1,-1)};

\node[font=\scriptsize, anchor=north, inner sep=1pt]
    at (axis cs:0.02,-2) {$\mu_1$};
\node[font=\scriptsize, anchor=north, inner sep=1pt]
    at (axis cs:0.25,-2) {$\mu_2$};
\node[font=\scriptsize, anchor=north, inner sep=1pt]
    at (axis cs:0.50,-1.9) {$\mu_2'$};
\node[font=\scriptsize, anchor=north, inner sep=1pt]
    at (axis cs:1.02,-2) {$\mu_3$};

\end{axis}
\end{tikzpicture}

\vspace{0.25em}
{\small (b) Finite search over candidate triples}
\end{minipage}

\par\medskip
\begin{minipage}{0.93\textwidth}
\small
\emph{Notes.}
Grey dotted lines mark the prior \(p\) and the benchmark \(\oV_{CT}(p)=2\). Dots are frontier points allowed by Proposition \ref{prop: opt MD binary}: black for \(\mu_1\), red for \(\mu_2\), and blue for \(\mu_3\). Arrows in panel (a) indicate the objective-improving directions before the frontier conditions bind. In panel (b), the two red points are alternative candidates for the second support point, so it remains to compare the two induced triples $\{\mu_1,\mu_2,\mu_3\}$ and $\{\mu_1,\mu_2',\mu_3\}$.
\end{minipage}
\end{figure}

When $A$ is finite, Proposition \ref{prop: opt MD binary} turns the binary-state mediation problem into a finite search. There is a finite partition of $[0,1]$ such that $\BV$ is constant on the interior of each interval. The frontier conditions imply that an optimal solution can be chosen with support on the boundaries of these intervals. 
Hence, there are finitely many candidate triples of posterior-value pairs. For each feasible triple, the Bayes-plausibility and covariance constraints uniquely determine the probabilities, so the optimum is obtained by comparing the induced mediation values. 
Figure \ref{fig:binary-frontier-panels}(b) illustrates this reduction: there are two candidates for the second support point, while the other two support points are fixed, and the problem reduces to comparing the two induced distributions.

For singleton-valued smooth payoff functions, differentiability further yields first-order necessary conditions for the three support points; see Proposition \ref{Prop: opt MD cont} in Appendix \ref{sec: OLA Binary}.

\section{Discussion and interpretation of the model} 
\label{sec: model interpretation}
In this section, we discuss the model's assumptions and interpretations.

\paragraph{Communication and correlated equilibria}
The mediator need not be interpreted as a benevolent or strategic third party. It can represent any communication technology, institutional procedure, certification protocol, advisory process, or correlation device that can commit ex ante to how messages and recommendations are generated. In fact, CE outcomes provide the broadest reduced-form benchmark for what can be achieved through communication, since they correspond exactly to the payoff outcomes induced by correlated equilibria of the strategic-form cheap talk game (\cite{forges2020games}). 

From this perspective, understanding when mediation is valuable allows us to determine when correlation is beneficial by enhancing strategic communication (cf. \cite{rudov2026extreme}). Likewise, studying the value of verifiability tells us when correlation is still not enough to attain what full direct commitment by the sender could attain.

\paragraph{Mediated communication as noise design}
Mediated communication can also be interpreted as a cheap talk game in which the sender or a third party commits ex ante to a noisy channel between the sender's original message and what the receiver observes.\footnote{Formally, starting from the cheap talk game above with rich message space $M$, assume that a designer can commit to a mechanism $\sigma :M \to \Delta(M')$, for some set $M'$, and the receiver knows $\sigma$ and observes only the realization of $m'$.}
It is well-known that adding an \emph{exogenous} noisy channel to a cheap talk game can expand the set of equilibrium payoffs (see \cite{blume2007noisy}). Instead, we study when the sender is willing to pay to add a designed, and hence \emph{endogenous}, noisy channel to their communication with the receiver.

This interpretation also clarifies the relation between persuasion and mediation. In both cases, a designer commits to a possibly random mapping from inputs to messages for the receiver. The difference is that, under mediation, the mapping is applied to a strategically chosen report, whereas under persuasion it is applied to the true state. Thus, the distinction is not technological, since both require commitment to randomization, but concerns \emph{verifiability}. In persuasion, unlike in mediation, it is \emph{as if} the mediator can \emph{perfectly} verify that the input is the true state.

\paragraph{Smart contracts and informed mediated communication}
The assumption that an individual can commit to a randomized mechanism, regardless of the input, may seem strong. One possible justification comes from long-run reputation arguments (see, for example, \cite{best2024persuasion}). Another is to interpret the mediator as an automated system that can be programmed to execute, possibly with randomness, the transmission of information or more general transactions. In this sense, the sender-aligned mediator can be viewed as a \emph{smart contract} designed by the sender.

Recent papers such as \cite{brzustowski2023smart} and \cite{drakopoulos2022blockchain} propose a similar interpretation of the Myersonian mediator. The mediator has the three main features of a smart contract: \emph{immutability}, since the program cannot be manipulated once designed; \emph{automatic execution}, since the output cannot be manipulated given the input; and \emph{encryption}, since the input cannot be recovered from the output. In our baseline model, the mediator is a smart contract whose only role is to convey information. In our extension to signaling games (Section \ref{sec:signaling}), it also conveys payoff-relevant outcomes. In both cases, we study when the sender is willing to pay to design a smart contract rather than interact directly with the receiver.

This interpretation also points to a natural robustness question. Suppose the sender can reprogram the smart contract after learning the state but before choosing the input, that is, immutability fails.\footnote{This is exactly the case considered in \cite{brzustowski2023smart}.} Then the model becomes one of \emph{informed mediation} as in \cite{KS2024infomediated}. In principle, this requires a separate analysis. However, \cite{KS2024infomediated} show that in sender--receiver games with transparent motives, equilibrium payoffs under mediation and informed mediation coincide. It follows that all our results extend unchanged to that case.\footnote{This follows from the fact that, under transparent motives, the Honesty condition (iii) implies that every type of the sender gets the same interim payoff. A similar argument also applies when the sender can choose whether to hire an aligned mediator \emph{after} learning the state.} The value of mediation then corresponds to the sender's willingness to pay for programmable, but not necessarily immutable, smart contracts.

The value of verifiability also has a natural interpretation in this perspective. An experiment can be viewed as a smart contract connected to an \emph{oracle} that perfectly verifies the truthfulness of the input by accessing external data. Thus, the value of verifiability is the sender's willingness to pay to append such an oracle to the smart contract.

\paragraph{Transparent Motives}
Even though in the main model we assume that the sender's payoff is state-independent, we can relax this assumption by requiring that, for every state $\omega \in \Omega$, the \emph{ordinal preference} of the sender over the receiver's mixed actions is the same. This amounts to saying that there exist vectors $\alpha,\beta \in \mathbb{R}^{\W}$, with $\alpha$ strictly positive, such that $u_S(\w,a)=\alpha(\w)u_S(a)+\beta(\omega)$. In this case, the Honesty condition is completely unaffected by $\alpha$ and $\beta$. Perhaps more surprisingly, the mediator's maximization problem is also unaffected by this change. To see this, observe that the constraints on $\eta$ imply
\begin{equation*}
   \mathbb{E}_{\eta}[\langle\alpha,\mu\rangle v+\langle\beta,\mu\rangle]=\mathbb{E}_{\eta}[\langle\alpha,\mu\rangle]\mathbb{E}_{\eta}[v] +\mathbb{E}_{\eta}[\langle\beta,\mu\rangle]=\langle \alpha, p\rangle \mathbb{E}_{\eta}[v]+\langle \beta, p\rangle
\end{equation*}
where the first equality follows by \eqref{eq:zero_cov} and the second by \eqref{Eq: Bayes-plausibility}. With this, the mediator's objective function is $\mathbb{E}_{\eta}[v]$, as in our main baseline model, and all our results apply. By contrast, when the sender's ordinal preferences depend on the state, the mediator can exploit this heterogeneity to elicit the state truthfully, implying that the value of mediation is often higher. In Appendix \ref{subApp: trilemma} we show how to extend Theorem \ref{Thm: Implementability} to the general state-dependent case and, via a simple example, that Theorem \ref{Thm: BP = MD iff BP = CT} \emph{does not} extend to this general case. We leave the full comparison among persuasion, mediation, and cheap talk in this case for future research.

\section{Applications}
\label{sec:applications}
In this section, we apply our results to show (i) when a lobbyist strictly benefits from intermediation, and (ii) when sender-preferred mediated communication strictly improves both the sender's and receiver's payoffs in a class of binary-action games with a privately informed receiver that we call \emph{acceptance games}.
\subsection{Lobbying through an intermediary}
\label{ssec:lobbying}
There are $n\geq 3$ possible states of the world, $\W=\{1,\dots,n\}$, and the receiver chooses an action in $A=\{0,1,\dots,n\}$. Action $0$ is a safe outside option, while each action $a\neq 0$ is an active policy whose payoff depends on how well it matches the state. The receiver's payoff is
\[
u_R(\omega,a)=\lambda-\phi(|\omega-a|)\quad\text{for }a\neq 0,
\qquad
u_R(\omega,0)=0,
\]
where $\lambda>0$ and $\phi:\mathbb R_+\to\mathbb R_+$ is increasing, with $\phi(0)=0$ and $\phi(1)>2\lambda$. Thus, active policies can be beneficial when they fit the state, but become worse than the status quo whenever there is a mismatch. By contrast, the sender prefers higher actions: $u_S(a)=a$.

A natural interpretation is a lobbying environment with an interested sender, a policymaker, and a communication intermediary. The sender is a firm or lobbyist with private information about which policy is most favorable to its interests. The receiver is a policymaker or regulator who values adopting a policy well-suited to the underlying state, but can always keep the status quo. The mediator may be interpreted, in reduced form, as a communication intermediary, such as a chamber of commerce or a think tank.\footnote{For example, \citet{thrall2025informational} studies how American Chambers of Commerce aggregate and transmit commercially relevant information to diplomats, providing a useful reduced-form analogue of an organized communication intermediary on the sender side. Relatedly, \citet{FraussenHalpin2017} study think tanks as contributors to policy advisory systems, offering a useful analogue of organizations that structure and channel policy-relevant information to policymakers.}

The example can also be interpreted as mediation in veto bargaining, with the sender playing the role of proposer and the receiver the role of veto player. The difference is the communication technology. Rather than giving the proposer direct control over information disclosure, as in \citet{kim2025persuasion}, or giving the veto player direct discretion over a menu of policies, as in \citet{kartik2021delegation}, communication here is structured through an intermediary that can support equilibrium outcomes unavailable under direct cheap talk.\footnote{We could formally allow for an endogenous proposal from the lobbyist, as in classical veto-bargaining models, through the extension to signaling games in Section \ref{sec:signaling}. The same results as in this section would go through. A further difference from the classical veto-bargaining model is that here the status quo is not merely a reference policy but a safe outside option for the policymaker. }

The receiver would choose an active policy $i > 0$ if and only if the posterior places sufficiently high probability on state $i$. For every active policy $i > 0$ and every state $k \neq i$, let $\beta_i^{k} \in \D(\W)$ denote the belief supported on $\{i,k\}$ with 
\[
\beta_{i}^k(k) = \frac{\l}{\phi(|i-k|)} , \qquad \beta_{i}^k(i) = 1 - \frac{\l}{\phi(|i-k|)}.
\]
At $\beta_i^k$, the receiver is indifferent between policy $i$ and the status quo. Since $\phi(1) > 2\l$, this belief is well-defined and assigns probability less than $1/2$ to the mismatched state $k\neq i$. For every $i > 0$, define 
\[
I_i \coloneqq \co \{\beta_i^k : k \neq i\}, \qquad D_i \coloneqq \co(\{\delta_i\} \cup I_i).
\]
The set $I_i$ is the set of posterior beliefs where the receiver is indifferent between policy $i$ and the status quo, while $D_i$ is the set of posteriors where policy $i$ is a best response. The condition $\phi(1)>2\lambda$ also implies that the active-policy regions are disjoint: $D_i\cap D_j=\emptyset$ for every $i\neq j$. Therefore, the sender's indirect payoff correspondence is 
\begin{align*}
        \BV(\mu) = 
        \begin{cases}
            \{i\} & \text{if } \mu \in D_i \setminus I_i  \text{ for some $i\in \W$}\\
            [0,i] & \text{if } \mu \in I_i  \text{ for some $i\in \W$}\\
            \{0\} & \text{otherwise}
        \end{cases}
    \end{align*}
Applying the quasi-concavification argument as in \cite{LR20}, the sender can secure a payoff of at least $i > 0$ under cheap talk if and only if the belief is in 
\[
S_i = \co\left(\bigcup_{j=i}^n D_j\right).
\]
The sets $S_i$ are nested: $S_{i+1}\subseteq S_i$ for every $i > 0$. The sender's optimal cheap talk payoff under prior $p$ is thus $\oV_{CT}(p) = \max\{i : p \in S_i\}$. 

We say that $p$ is non-boundary if it is not on the boundary of $S_i$ for every $i > 0$. At every non-boundary prior, $\oV_{CT}$ is locally constant, so the full-dimensionality condition holds. Applying Proposition \ref{Prop: sufficient cond for BP > MD} and Theorem \ref{Thm: MD vs CT}, verifiability is valuable if and only if the cheap talk value at $p$ is not globally maximal, while mediation is valuable if and only if cheap talk is directionally improvable at $p$.

\begin{proposition}\label{prop: lobbying}
For every non-boundary prior $p$:
\begin{enumerate}
    \item[(i)] If $p\in S_n$, then neither verifiability nor mediation is valuable.
    \item[(ii)] If $p\in S_{n-1}\setminus S_n$, then verifiability is valuable but mediation is not.
    \item[(iii)] If $p \notin S_{n-1}$, then both verifiability and mediation are valuable.
\end{enumerate}
\end{proposition}

\newcommand{\drawfourstatepanel}[5]{%
\begin{scope}[shift={(#1,#2)}]

  \coordinate (#4d1) at ( 3.50,-1.45);
  \coordinate (#4d3) at (-3.10,-1.65);
  \coordinate (#4d4) at ( 0.20, 5.35);

  \coordinate (#4d2) at (-0.55, 0.45);

  \pgfmathsetmacro{\tau}{#3}

  \pgfmathsetmacro{\tfortyone}{1/(2+3*\tau)}   
  \pgfmathsetmacro{\tfortytwo}{1/(2+2*\tau)}   
  \pgfmathsetmacro{\tfortythree}{1/(2+\tau)}   
  \pgfmathsetmacro{\tthirtyone}{1/(2+2*\tau)}  
  \pgfmathsetmacro{\tthirtytwo}{1/(2+\tau)}    

  \coordinate (#4b41) at ($(#4d4)!\tfortyone!(#4d1)$);
  \coordinate (#4b42) at ($(#4d4)!\tfortytwo!(#4d2)$);
  \coordinate (#4b43) at ($(#4d4)!\tfortythree!(#4d3)$);

  \coordinate (#4b31) at ($(#4d3)!\tthirtyone!(#4d1)$);
  \coordinate (#4b32) at ($(#4d3)!\tthirtytwo!(#4d2)$);


  \fill[green!7] (#4d1)--(#4d3)--(#4d4)--cycle;
  \fill[pattern=north east lines, pattern color=green!30!black]
      (#4d1)--(#4d3)--(#4d4)--cycle;

  \fill[green!35, opacity=.12]
      (#4b32)--(#4b31)--(#4b41)--(#4b42)--cycle;

  \fill[blue!8] (#4d3)--(#4b31)--(#4b41)--(#4d4)--cycle;
  \fill[pattern=north west lines, pattern color=blue!42]
      (#4d3)--(#4b31)--(#4b41)--(#4d4)--cycle;

  \fill[red!8] (#4d4)--(#4b41)--(#4b43)--cycle;
  \fill[pattern=horizontal lines, pattern color=red!55]
      (#4d4)--(#4b41)--(#4b43)--cycle;

  \draw[black!45, dashed, line width=.40pt] (#4d2)--(#4d1);
  \draw[black!45, dashed, line width=.40pt] (#4d2)--(#4d3);
  \draw[black!45, dashed, line width=.40pt] (#4d2)--(#4d4);

  \draw[blue!55!black, dashed, line width=.33pt] (#4b32)--(#4b31);
  \draw[blue!55!black, dashed, line width=.33pt] (#4b32)--(#4b42);

  \draw[red!65!black, dashed, line width=.33pt] (#4b42)--(#4b41);
  \draw[red!65!black, dashed, line width=.33pt] (#4b42)--(#4b43);

  \draw[blue!70!black, line width=.70pt]
      (#4d3)--(#4b31)--(#4b41)--(#4d4)--cycle;

  \draw[red!75!black, line width=.70pt]
      (#4d4)--(#4b41)--(#4b43)--cycle;

  \draw[black, line width=.82pt] (#4d1)--(#4d3)--(#4d4)--cycle;

  \fill[white] (#4d2) circle (1.20pt);
  \draw[black!65, dashed, line width=.40pt] (#4d2) circle (1.20pt);

  \fill (#4d1) circle (1.00pt);
  \fill (#4d3) circle (1.00pt);
  \fill (#4d4) circle (1.00pt);

  \node[innerlabel] at ($(#4d2)+(-0.42,0.16)$) {$\delta_2$};

  \node[right, font=\scriptsize] at (#4d1) {$\delta_1$};
  \node[below, font=\scriptsize] at (#4d3) {$\delta_3$};
  \node[above, font=\scriptsize] at (#4d4) {$\delta_4$};

  \node[paneltitle] at (0.15,6.55) {#5};

\end{scope}
}

\begin{figure}[t]
\centering

\resizebox{\textwidth}{!}{%
\begin{tikzpicture}[
  x=0.58cm,
  y=0.58cm,
  line join=round,
  line cap=round
]

  \tikzset{
    innerlabel/.style={
      font=\scriptsize,
      inner sep=1.0pt,
      fill=white,
      fill opacity=.94,
      text opacity=1,
      rounded corners=.7pt
    },
    paneltitle/.style={
      font=\footnotesize
    }
  }

  \drawfourstatepanel{0.0}{0.0}{0.5}{A}{(a) Low cost}
  \drawfourstatepanel{8.3}{0.0}{1}{B}{(b) Intermediate cost}
  \drawfourstatepanel{16.6}{0.0}{2}{C}{(c) High cost}

  \begin{scope}[shift={(8.3,-3.95)}]

    \fill[red!8] (-7.35,0.05) rectangle (-6.80,0.32);
    \fill[pattern=horizontal lines, pattern color=red!55]
      (-7.35,0.05) rectangle (-6.80,0.32);
    \draw[black!75] (-7.35,0.05) rectangle (-6.80,0.32);
    \node[right, font=\scriptsize] at (-6.58,0.185)
      {$\mathrm{BP}=\mathrm{MD}=\mathrm{CT}$};

    \fill[blue!8] (-1.75,0.05) rectangle (-1.20,0.32);
    \fill[pattern=north west lines, pattern color=blue!42]
      (-1.75,0.05) rectangle (-1.20,0.32);
    \draw[black!75] (-1.75,0.05) rectangle (-1.20,0.32);
    \node[right, font=\scriptsize] at (-0.98,0.185)
      {$\mathrm{BP}>\mathrm{MD}=\mathrm{CT}$};

    \fill[green!7] (4.10,0.05) rectangle (4.65,0.32);
    \fill[pattern=north east lines, pattern color=green!30!black]
      (4.10,0.05) rectangle (4.65,0.32);
    \draw[black!75] (4.10,0.05) rectangle (4.65,0.32);
    \node[right, font=\scriptsize] at (4.88,0.185)
      {$\mathrm{BP}>\mathrm{MD}>\mathrm{CT}$};

  \end{scope}

\end{tikzpicture}%
}
\caption{Comparative statics in mismatch costs.}
\label{fig:four-state-comparative-statics}
\smallskip
\small
$n=4$, $\lambda=1$, $\phi(d) = 2 + kd$, $k=1/2$ in (a), $k=1$ in (b), $k=2$ in (c). 
\end{figure}

In particular, both verifiability and mediation are valuable when the prior cannot be decomposed into posteriors that justify the sender's two most preferred policies. As mismatch costs increase, the receiver demands more confidence before adopting any active policy, so the regions $S_{n-1}$ and $S_n$ shrink, making $p\notin S_{n-1}$ easier to satisfy. Figure \ref{fig:four-state-comparative-statics} illustrates this with a four-state example. Given the mismatch cost, increasing the number of states (and policies) has a similar effect: the high-policy region $S_{n-1}$ occupies a smaller share of the simplex. Therefore, strict value from verifiability and mediation becomes more prevalent when active policies are harder to justify or when the decision problem is more complex.

\subsection{Acceptance games and Pareto improving mediation}
\label{ssec:accepta_pareto}
In this section, we apply our results to a class of \emph{acceptance games}
nesting the salesperson examples in
\citet{chakraborty2010persuasion} and \citet{LR20}. We provide sufficient
conditions under which sender-preferred mediation is strictly (ex-ante) Pareto
improving relative to \emph{any} cheap-talk equilibrium.

The receiver has a binary choice: whether to accept or reject a risky prospect. The prospect's value depends on a multidimensional moment function $T:\D(\Omega)\to\R^k$ with $k \geq 2$. Given  a posterior belief
$\mu \in \D(\W)$, let $x=T(\mu)\in X$ denote the relevant moment vector, where $X$ is the set of all relevant moments which has dimension $k$. The receiver's value of the prospect is $R: X \to \R$, which is continuous and strictly quasiconvex.

The receiver compares $R(x)$ to an outside option with value
$ r \in\R$, which is their private information and is drawn from a
strictly increasing and continuously differentiable CDF $G$. Assume that
$R(X)\subseteq \supp(G)$, so the outside option is competitive. The receiver's expected payoff given moment $x$ is
\[
v_R(x)\coloneqq H(R(x)),
\qquad
H(y)\coloneqq \int \max\{ r ,y\}\,dG( r ).
\]
Given moment
 $x$, the sender's payoff is the probability that the receiver accepts
the prospect, that is,
\[
v(x)\coloneqq G(R(x)).
\]
Hence, $\BV(\m) = \{v(T(\m))\}$ for all $\m\in\D(\W)$. For any $\eta\in\D(\D(\Omega)\times\R)$ that is feasible under
\eqref{Eq: MD problem}, only the marginal distribution $\tau = \marg_{\D(\Omega)}\eta$
matters for payoffs, and we can equivalently describe outcomes in this section
through distributions of posteriors.

We say that mediation is \emph{strictly (ex-ante) Pareto improving} at prior
$p$ if there exists a distribution $\tau^*\in\D(\D(\Omega))$ feasible under
mediation such that
\[
\mathbb{E}_{\tau^*}[v\circ T]>\mathbb{E}_\tau[v\circ T]
\qquad\text{and}\qquad
\mathbb{E}_{\tau^*}[v_R\circ T]>\mathbb{E}_\tau[v_R\circ T]
\]
for every distribution $\tau$ feasible under cheap talk.

Even though both players rank posterior moments through the same index $R(x)$,
the sender's incentives can still be non-monotone in posterior beliefs. We
capture this by requiring non-monotonicity of $R$ along the edges issuing
from a worst degenerate belief. Let $X_T\coloneqq \{T(\delta_\w): \w \in \W\}$
denote the set of moment vectors induced by degenerate beliefs. To rule out the
trivial case in which unmediated communication already attains the global
maximum payoff, we maintain throughout that $\min_{x\in X_T}R(x)<\max_{x\in X}R(x)$.

\begin{definition}\label{def: minimally edge non-mon}
We say that $R$ is \emph{minimally edge non-monotone given $T$} if there
exists $\underline x\in \arg\min_{\tilde x\in X_T}R(\tilde x)$
such that, for every $x\in X_T\setminus\{\underline x\}$, the one-dimensional
function
\[
\hat R_x(\lambda):=R\bigl(\lambda x+(1-\lambda)\underline x\bigr)
\]
is not weakly increasing in $\lambda\in[0,1]$.
\end{definition}

Our next result combines this condition with log-concavity of the outside option distribution to obtain a Pareto comparison between mediation and cheap talk. Recall that $G$ is log-concave if $\log G$ is concave. Log-concavity is a standard assumption satisfied by many standard unimodal distributions on the real line.

\begin{proposition}
\label{pro:accep_Pareto}
    If $R$ is minimally edge non-monotone given $T$, and $G$ is log-concave, then
    there exists an $(n-1)$--simplex $\tilde{\D} \subseteq \D(\W)$ such that verifiability is valuable and mediation is strictly Pareto improving at $p$ for all $p \in \interior \tilde{\D}$.
\end{proposition}

Since $G$ is strictly increasing, the sender's payoff
$v(x)=G(R(x))$ inherits the non-monotonicity and strict quasiconvexity of $R$. This ensures the sender is strictly
better off under mediation than under cheap talk for priors in $\interior \tilde{\D}$. At the same time, the receiver's payoff can be written as
\[
v_R(x)=\psi(v(x)),
\qquad
\psi(z):=H(G^{-1}(z)),
\]
and log-concavity of \(G\) implies that \(\psi\) is convex. Therefore, whenever mediation raises the sender's ex-ante payoff, Jensen's inequality
implies that it also raises the receiver's ex-ante payoff relative to cheap talk, where the sender's payoff must be constant across posterior realizations.

The assumptions on \(R\) are only needed to obtain a strict sender improvement on a
non-trivial set of priors. Once such a strict improvement exists, log-concavity of $G$ upgrades it to a strict Pareto improvement. The same logic also applies to the comparison between Bayesian persuasion and cheap talk. For example, in the
salesperson model of \citet{LR20}, when $G$ is log-concave, and their Claim 3 applies, commitment is strictly (ex-ante) Pareto improving.

\paragraph{Uncertain projects and costly implementation.}
We close this section by giving a concrete illustration of an acceptance game and of our result.
An innovator (the sender) proposes to a firm (the receiver) a bundle of \(k\)
projects with uncertain payoffs \(\w\in\W\subseteq [0,1]^k\setminus\{\mathbf 0\}\).
The coordinate \(\w_i\) is the payoff generated by project \(i\) if that project
is successfully implemented. The
firm chooses whether to accept or reject it. If it rejects, it obtains an outside
option \( r \), where \( r \) is privately observed by the firm and
drawn from a log-concave distribution \(G\). If it accepts, it chooses
implementation efforts \(a\in[0,1]^k\), incurs cost
\(\frac12\sum_{i=1}^k a_i^2\), and project \(i\) succeeds with probability \(a_i\).

Given posterior belief \(\mu\in\D(\W)\), only the posterior mean $T(\mu) = \mathbb{E}_{\mu}[\w] \in \R^k$ matters for the firm's decision. The set of posterior means $X = T(\D(\W)) = \co(\W)$ has dimension at least $2$, that is, posterior means vary in at least two
affinely independent directions. Since \(\W\subseteq [0,1]^k\), every posterior
mean \(x=T(\mu)\) also lies in \([0,1]^k\), so the firm's optimal effort choice is
\(a_i=x_i\) for each \(i\). Hence, the firm's expected payoff from accepting is
\[
R(x)
=
\max_{a\in[0,1]^k}\sum_{i=1}^k a_i x_i-\frac12\sum_{i=1}^k a_i^2
=
\frac12\sum_{i=1}^k x_i^2
=
\frac12\|x\|^2.
\]
Thus $R$ is strictly convex, and therefore strictly quasiconvex. We maintain the standing assumptions that $R(X) \subseteq \supp G$ and $\min_{x\in X_T}R(x) < \max_{x\in X}R(x)$. In this example, the latter condition is equivalent to $\|\w\|^2$ being non-constant over $\w\in \W$, which rules out the degenerate case in which all states induce the same acceptance value for the firm.

\begin{proposition}\label{prop: project investment}
\leavevmode
\begin{enumerate}[(i)]
    \item $R$ is minimally edge non-monotone given \(T\) if and only if there exists \(\underline{\w}\in\arg\min_{\w\in\W}\|\w\|\) such that
    \[
    \langle \w,\underline{\w}\rangle<\|\underline{\w}\|^2
    \qquad\text{for every }\w\neq\underline{\w}.
    \]

    \item If the condition in (i) holds, then there exists an $(n-1)$-simplex
    $\tilde{\D}\subseteq\D(\W)$ such that mediation is strictly Pareto improving
    for every $p\in\interior\tilde{\D}$.

    \item The condition in (i) holds if each state $\w \in \W$ has exactly one project with a positive payoff, and there exists  $\underline{\w}\in\arg\min_{\w\in\W}\|\w\|$ that is not Pareto dominated.
\end{enumerate}
\end{proposition}

Part (i) shows that, in this example, edge non-monotonicity is a \emph{misalignment} condition. Suppose the firm were certain that the state is $\underline{\omega}$, it would then choose effort $\underline{\omega}$, yielding acceptance value $R(\underline{\omega})=\frac12\|\underline{\omega}\|^2$. The condition in (i) says that every other state $\omega$ is sufficiently misaligned with this effort profile: if the true state were $\omega$ but the firm still chose the effort tailored to $\underline{\omega}$, the resulting expected payoff would be strictly lower than under state $\underline{\omega}$ itself. Therefore, a small amount of prior uncertainty around $\underline{\omega}$ can reduce the value of accepting, because the firm may incur effort costs on projects that are poorly matched to the true payoff profile.

Part (iii) gives the cleanest illustration of this force. If only one project pays off in each state, then investing in the wrong project generates no return but still entails an effort cost. Thus, when the worst state is not simply a lower-payoff version of another state, the other states are naturally misaligned with it. In this case, mediation exploits the countervailing effects implied by the misalignment condition to generate a strict Pareto improvement for all players over cheap talk.

\section{Matching with externalities: Efficiency vs. fairness}
\label{sec:matching}
In this section, we give a more general interpretation of our framework in terms of assigning agents with heterogeneous traits to units or groups and describe the implications of our results therein. Many matching environments feature \emph{externalities} such as peer or spillover effects: the value generated by a firm, classroom, neighborhood, or hospital unit depends not only on the characteristics of one agent, but on the entire composition of the group. This creates a basic tension between efficiency and fairness. Efficient allocations typically exploit variation in group composition in order to generate more value, whereas fair allocations seek to equalize the benefits received by different types of agents. A reinterpretation of our model and results allows us to study this tension in a tractable reduced-form model that isolates the role of externalities.

Our model is intended to capture environments in which agents are grouped into productive or socially relevant units with \emph{aligned preferences}, that is, agents assigned to the same unit obtain the same payoff.\footnote{The assumption of aligned preferences here mirrors the transparent-motives assumption in our sender-receiver interpretation. Similarly to our discussion in Section \ref{sec: model interpretation}, this assumption can be relaxed to allow for agents with different traits to have different \emph{cardinal utility} over distributions of units, provided that they all have the same \emph{ordinal preference}.} This common-payoff structure is natural, for instance, when workers in a partnership equally share the output generated by their team (\cite{farrell1988partnerships}), when students in a classroom experience the same schooling environment, or when residents in a neighborhood benefit from the same local amenity value.

There is a unit mass of agents, each with a trait $\omega \in \Omega$,  and the population trait distribution is $p \in \Delta(\Omega)$, which has full support. Each group, or unit, is summarized by its composition $\mu \in \Delta(\Omega)$, interpreted as the relative frequencies of traits within that unit, and total mass $m \in [0,1]$. We focus on normalized compositions as we assume that the production technology has constant returns to scale.

Each unit composition $\mu$ can generate a set of feasible values. Formally, let $\mathbf V:\Delta(\Omega)\rightrightarrows \mathbb R$ be a (Kakutani) value-possibility correspondence, where $\mathbf V(\mu)$ is the set of values that can be generated by a unit with composition $\mu$. The interpretation is that, once a unit with composition $\mu$ is formed, some feasible activity, policy, or technology is selected, and this generates a value $v\in \mathbf V(\mu)$ that is equally shared by all agents in that unit and that encompasses all the externalities.

An allocation is a joint distribution $\eta \in \Delta\big(\Delta(\Omega)\times \mathbb R\big)$ which describes both the frequency of different unit compositions and the value assigned to each such unit. Feasibility requires two conditions. First, the allocation must preserve the aggregate trait distribution, that is, $\mathbb{E}_\eta[\mu] = p$.
Second, every realized pair $(\mu,v)$ must lie on the graph of the value correspondence, that is, $\eta(\Gr(\mathbf V))=1$.
Thus, a feasible allocation is simply a way to decompose the population into groups and assign a technologically feasible value to each realized group. 

In order to define our notions of efficiency and fairness, one must also keep track of the payoff received by each trait in expectation. For every feasible allocation $\eta$, the expected payoff of a trait-$\omega$ agent under allocation $\eta$ is
\[
U_\omega(\eta):=\mathbb E_\eta[v\mid \omega]
=\int v\,\frac{\mu(\omega)}{p(\omega)}\, d\eta(\mu,v).
\]
This is the ex-ante payoff received by an agent of trait $\omega$ after averaging over all units in which that trait may end up.

\begin{definition}
A feasible allocation $\eta$ is \emph{Pareto optimal} if there is no feasible allocation $\eta'$ such that
\[
U_\omega(\eta')\ge U_\omega(\eta)\qquad\forall \omega\in\Omega,
\]
with strict inequality for some $\omega'\in\Omega$.
\end{definition}
This is the standard notion of Pareto optimality across traits. It is weaker than \emph{utilitarian efficiency} (i.e., being a solution to the Bayesian persuasion problem) but stronger than the \emph{weak Pareto optimality} property considered in \cite{doval2021information}.\footnote{This definition of Pareto optimality should not be confused with the ex-ante notion used for the sender and the receiver in Section \ref{ssec:accepta_pareto}.}

We next distinguish two notions of fairness.
\begin{definition}
A feasible allocation $\eta$ is \emph{ex-ante fair} if
\[
U_\omega(\eta)=U_{\omega'}(\eta)\qquad\forall \omega,\omega'\in\Omega.
\]
A feasible allocation $\eta$ is \emph{ex-post fair} if $v$ is constant for $\eta$-almost all $(\mu,v)$.
\end{definition}
Our notions of fairness are value-based. Ex-ante fairness requires all traits to obtain the same value \emph{in expectation}, leaving room for realized payoff inequalities. Ex-post fairness is much stronger: it requires all units to yield the same \emph{realized} value. Theorem \ref{Thm: Implementability} and the results in \cite{LR20} can be used to obtain simple moment characterizations of these notions.

\begin{Remark}
   A feasible allocation $\eta$ is ex-ante fair if and only if $\Cov_\eta[v,\mu]=\mathbf 0$, or equivalently, if and only if $\Var_p\!\left[\mathbb{E}_{\eta^{\omega}}[v]\right]=0$. A feasible allocation $\eta$ is ex-post fair if and only if $\Var_\eta[v]=0$.
\end{Remark}

The first condition says that realized value cannot covary with group composition in a way that systematically benefits some traits more than others. As observed in Remark \ref{rk: anova}, this amounts to requiring that the between-trait component of the total variance is $0$. The second says that ex-post fairness eliminates all dispersion in realized values. Under the communication interpretation, these conditions correspond respectively to the truth-telling constraint in mediation and the sender's optimality requirement in cheap talk.

\begin{theorem}
\label{th:impossibility}
A feasible allocation $\eta$ is Pareto optimal and ex-ante fair if and only if it is Pareto optimal and ex-post fair.
\end{theorem}
Theorem \ref{th:impossibility} delivers a sharp incompatibility result: unless efficiency is already compatible with constant realized payoffs, there is no efficient allocation satisfying ex-ante fairness. The proof of this result generalizes that of Theorem \ref{Thm: BP = MD iff BP = CT} but still relies on a similar perturbation argument. The intuition is that, if an allocation is ex-ante fair but still produces unequal realized values across groups, the planner can tilt the allocation toward higher-value groups without systematically favoring any trait. This yields a Pareto improvement.
In addition, the incompatibility result becomes even stronger when the full-dimensionality condition holds.\footnote{In this setting, the full-dimensionality condition holds at almost every trait distribution $p$ provided that the value-possibility correspondence $\mathbf{V}$ is defined as in \eqref{eq:value_corresp} for some \emph{finite} set $A$. This can be interpreted as a model where each group can be assigned to one of finitely many feasible tasks $a \in A$ to complete.}
  \begin{corollary}
  If the full-dimensionality condition holds at $p$, then either $p \in \co \argmax \overline{V}$ or there does not exist a feasible, Pareto optimal, and ex-ante fair allocation at $p$.
    \end{corollary}
   Unless the initial trait composition lies in the convex hull of value-maximizing compositions, an allocation that is both Pareto optimal and ex-ante fair does not exist.

An alternative way to address the efficiency-fairness tradeoff would be to study the efficiency loss when we constrain the social planner problem with fairness requirements. Under the interpretation of this section, mediation and cheap talk exactly coincide with the problem of a social planner with a utilitarian objective that picks the assignment subject to ex-ante and ex-post fairness, respectively. With this, Theorem \ref{Thm: MD vs CT} (for the general model) and Corollary \ref{Cor: MD vs CT finite} (for the finite-task model) characterize those value correspondences $\mathbf{V}$ such that an ex-post fairness constraint would entail a strict loss in efficiency with respect to an ex-ante fairness constraint. 

Here, directional improvability captures the idea that there does not exist a clear direction of monotonicity of preferences over group composition. For example, this happens in a model where groups with a uniform distribution over traits generate lower values, while groups with a composition skewed towards some traits generate larger values, such as the correspondence considered in Section \ref{ssec:lobbying}.

\section{Signaling games with transparent motives}
\label{sec:signaling}
In this section, we extend our analysis by allowing the sender to use payoff-relevant
signals in addition to payoff-irrelevant messages. The formal statements and proofs are relegated to Appendix \ref{app:signaling-proofs}. We follow the belief-based
formulation of signaling games in \cite{koessler2024belief}. The sender observes
\(\omega\), chooses a signal \(x\in X\), and may also send a cheap-talk message;
the receiver observes the signal-message pair, forms a belief, and chooses an
action. We maintain transparent motives: payoffs are \(u_S(x,a)\) and \(u_R(\omega,x,a)\).

The mediator commits to a signaling mechanism
\(\sigma:M_S\to\Delta(X\times M_R)\). After the sender reports \(m_S\), the mechanism
draws a pair \((x,m_R)\), observed by the receiver. This remains within the
framework of \cite{myerson1982optimal}. The smart-contract interpretation from
Section \ref{sec: model interpretation} is especially natural here, since the
mediator may now commit not only to information transmission but also to
payoff-relevant outcomes, such as prices or quantities.\footnote{For example,
this is the role played by smart contracts in \cite{brzustowski2023smart}.}

The only change in the belief-based approach is that the sender's interim value
correspondence must incorporate the signal. For each posterior \(\mu\) and
signal \(x\), let
\begin{equation*}
    \mathbf W(\mu,x):=\co\big(u_S(x,\argmax_{a\in A}\mathbb E_\mu[u_R(\omega,x,a)])\big)
\end{equation*}
be the set of sender payoffs attainable when the receiver best responds to
\((\mu,x)\), and define
\(\mathbf V^S(\mu):=\bigcup_{x\in X}\mathbf W(\mu,x)\) and
\(\overline V^S(\mu):=\max \mathbf V^S(\mu)\). Thus \(\mathbf V^S(\mu)\) records all
sender values attainable at posterior \(\mu\), allowing for the choice of a
payoff-relevant signal.

With this substitution, the characterization of mediated communication is
unchanged in substance. A mediated outcome induces a distribution over triples
\((\mu,x,v)\), or its projection over posterior-value pairs \((\mu,v)\).
Consistency* is still \(\mathbb E[\mu]=p\), obedience* now requires
\(v\in\mathbf W(\mu,x)\) almost everywhere, and honesty* is still
\(\Cov(v,\mu)=\mathbf 0\). The reason is the same as in Theorem
\ref{Thm: Implementability}: all sender types rank
signal-action pairs in the same way, so truth-telling requires every report to
generate the same expected sender payoff. In belief space, this amounts to the
orthogonality of \(v\) and \(\mu\).\footnote{Formally, the proof is the same as the proof of Theorem
\ref{Thm: Implementability}, with the outcome distribution pushed forward to
\((\mu,x,v)\) rather than \((\mu,v)\). Necessity follows from Bayes plausibility,
receiver optimality conditional on the observed signal, and sender indifference
across reports; sufficiency follows by selecting, for each \((\mu,x)\), a mixed
receiver best response that delivers the required value.}

The direct-signaling benchmark has the analogous flatness requirement. Without
mediation, every posterior-signal-message realization used on path must give
the sender the same payoff; otherwise, all sender types would deviate to the more
profitable realization. By Theorem 2 in \cite{koessler2024belief}, the
sender-preferred equilibrium payoff in signaling games with transparent motives
is the quasiconcave envelope of \(\overline V^S\) evaluated at the prior.\footnote{
More precisely, \cite{koessler2024belief} state the result for the
sender-preferred PBE value. For our BNE benchmark, their INTIR condition is not
relevant for the sender-preferred value.} Bayesian persuasion instead
removes the sender's incentive constraint and allows any Bayes-plausible
distribution over posterior-signal pairs.

Mediation improves on direct signaling through the same
countervailing-incentive logic as before. The mediator can improve on direct
signaling when it can assign a payoff above the direct-signaling value to a
posterior closer to the prior and a payoff below that value to a more extreme
posterior, while satisfying Bayes plausibility and zero covariance. Equivalently,
using the signal-contingent di-convexification \(\mathbf W_{CT}\), one looks
for \((\mu^+,x^+,v^+),(\mu^-,x^-,v^-)\in\Gr(\mathbf W_{CT})\) such that
\begin{equation*}
    \mu^+\in(p,\mu^-) \quad \text{and} \quad v^+>\overline V_{CT}(p)>v^-.
\end{equation*}
The high-value
realization raises the sender's payoff; the low-value realization is the
credibility cost that keeps reports incentive compatible. The same averaging argument used in the proof of Theorem
\ref{Thm: MD vs CT} gives necessity under full dimensionality.\footnote{Once outcomes are represented by triples \((\mu,x,v)\), the proof of
Theorem \ref{Thm: MD vs CT} applies to their projection on \((\mu,v)\), with
\(\mathbf V\) replaced by the enlarged correspondence
\(\mathbf V^S(\mu)=\bigcup_x\mathbf W(\mu,x)\). The only additional step is the
measurable selection of signals and receiver best responses delivering the
selected value \(v\in\mathbf V^S(\mu)\).}

This extension clarifies why payoff-relevant signals can make mediation more
powerful. In the baseline model, the mediator moves only across receiver
posteriors. In a signaling game, it can also move vertically within the feasible
sender payoffs at a posterior by choosing the signal. This makes it easier to
create the high- and low-payoff realizations needed for zero covariance.

To illustrate, return to the platform example of Section
\ref{Subsec: Illustrative Example}, but now let the seller choose the price
\(x\ge 0\). Given posterior \(\mu\), type-\(1\) buyers buy whenever
\(\mu\ge x\), while type-\(0\) buyers buy whenever \(1-\mu\ge x\). The upper envelope of $\mathbf{W}(\mu,x)$ is thus
\begin{equation*}
    x\left(\frac{1}{3}\mathbb{I}[\mu \geq x]+\frac{2}{3}\mathbb{I}[1-\mu \geq x]\right),
\end{equation*}
where indifferent buyers purchase. The lower envelope is obtained by replacing the weak inequalities with strict ones.
Direct signaling
requires the seller's payoff to be constant across all posterior-price pairs
used in equilibrium. A mediating platform instead needs only
\(\Cov(v,\mu)=0\) with $v\in\bigcup_{x\in X}\mathbf{W}(\mu,x)$ almost surely. Thus prices can be used to lower the seller's payoff
after some posteriors to make more favorable posterior-price pairs
credible after others. For instance, when $p \in (2/3, 1)$, mediation strictly improves on direct signaling: $1/2 \in \mathbf{W}(1/2,1/2)$ and $0 \in \mathbf{W}(0,0)$ with $0 < 1/2 < p$ and $1/2 > \oV_{CT}(p) = 1/3 > 0$. This improvement is unavailable with a fixed price $1/3$ in Section \ref{Subsec: Illustrative Example}.

Beyond signaling games, our results also have implications for repeated games
with asymmetric information and transparent motives, as in
\cite{hart1985nonzero} and \cite{aumann2003long}. Results of
\cite{forges1985correlated} and \cite{habu2024knowing} relate Nash and
correlated-equilibrium payoffs in these dynamic environments to equilibrium
payoffs of the corresponding static communication games. Hence our comparison
between cheap talk and mediation also translates into a comparison between Nash
and correlated equilibria in the associated repeated games. We defer the formal
statements to Appendix~\ref{app:repeated}.

\begingroup
\small
\begin{spacing}{0.9}
\bibliographystyle{ecta}
\bibliography{reference}
\end{spacing}
\endgroup

\appendix
\section{Appendix}

\subsection{Preliminaries} \label{app: prelim}

The proof of the next ancillary lemma is standard and relegated to Appendix \ref{App: Lemmas}.

\begin{lemma}\label{Lem: Obedience}
    An outcome distribution $\pi \in \D(\W \times A)$ satisfies Obedience if and only if for every measurable $\tilde{a} : A \to A$, $\int u_R(\w,a) \de \pi(\w,a) \geq \int u_R(\w,\tilde{a}(a)) \de \pi(\w,a)$.
\end{lemma}

Given two measurable spaces $(X,\Sigma), (X',\Sigma')$, a measure $\eta$ on $(X,\Sigma)$, and a measurable function $f:X\to X'$, we let
$(f)_{\#}\eta$ denote the pushforward measure of $\eta$ under $f$. A probability kernel from $X$ to $X'$ is a mapping $\nu: X\times \Sigma' \to \R_+$ such that $\nu(\cdot, S)$ is $\Sigma$-measurable for fixed $S \in \Sigma'$ and $\nu(x, \cdot)$ is a probability measure on $X'$ for fixed $x$. We let $\eta \circ \nu$ denote the composition of a probability measure with a kernel, which is a probability measure in $\Delta(X')$ defined as
\[
(\eta\circ \nu)(S) = \int_X \nu(x,S) \de \eta(x)
\]
for every $S \in \Sigma'$. We use $\mathbb{I}$ to denote the indicator function.

Given prior $p \in \D(\W)$, we use $\TT_{BP}(p), \TT_{MD}(p)$ and $\TT_{CT}(p)$ to denote the set of distributions feasible under persuasion, mediation, and cheap talk, respectively.

\subsection{Proofs}
\label{app:proofs}
\begin{proof}[Proof of Theorem \ref{Thm: Implementability}]
    We first show the only if direction. Suppose that $\eta \in \D(\D(\W)\times \R)$ is induced by some communication equilibrium outcome $\pi\in \D(\W\times A)$. Note that $\eta$ is the pushforward measure of $\pi$ under the measurable function $\phi: \W\times A \to \D(\W) \times \R$, where $\phi_1(\w,a) = \pi^a$ is a version of the conditional probability over $\W$ given $a$, and $\phi_2(\w,a) = u_S(a)$. For every $\w\in \W$,
    \begin{align*}
        \int_{\D(\W)\times \R} \mu(\w) \de \eta(\mu, v) =& \int_{\W\times A} \phi_1(\tilde{\w},a)(\w) \de \pi(\tilde{\w}, a) = \int_{\W\times A} \pi^a(\w) \de \pi(\tilde{\w},a)\\ 
        =& \int_{\W \times A} \mathbb{I}[\tilde{\w}=\w] \de \pi(\tilde{\w}, a) = p(\w),
    \end{align*}
    where the first equality follows because $\eta = (\phi)_{\#} \pi$, the second equality follows by definition, the third equality follows by iterated expectations, and the last equality follows by Consistency of $\pi$. Hence, $\eta$ satisfies \eqref{Eq: Bayes-plausibility}.

    Note that by Obedience, for $\pi$-almost all $a \in A$, $u_S(a) \in \BV(\pi^a)$, so $\phi(\w,a) = (\pi^a, u_S(a)) \in \Gr(\BV)$. Hence, $\eta(\Gr(\BV)) = \pi(\phi^{-1}(\Gr(\BV))) = 1$ because $\eta = (\phi)_\# \pi$, so \eqref{Eq: Obedience} is satisfied.

    By Honesty of $\pi$ and the fact that $u_S$ does not depend on $\w$, we have $\mathbb{E}_{\pi^\w}[u_S] = \mathbb{E}_\pi[u_S]$ for every $\w\in\W$. By \eqref{Eq: Bayes-plausibility} we have shown before, for every $\w \in \W$, 
    \begin{align*}
        \Cov_\eta[v,\mu(\w)] =& \int_{\D(\W)\times \R} v \, \mu(\w) \de \eta(\mu, v) - p(\w) \int_{\D(\W)\times \R} v  \de \eta(\mu, v) \\
        =& \int_{\W\times A} u_S(a) \pi^a(\w) \de \pi(\tilde{\w}, a) - p(\w) \int_{\W\times A} u_S(a) \de \pi(\tilde{\w},a)\\
        =&\; p(\w) \, (\mathbb{E}_{\pi^\w}[u_S] - \mathbb{E}_\pi [u_S])  = 0,
    \end{align*}
    where the second equality follows from $\eta = (\phi)_\# \pi$, and the third equality holds because $\frac{\de \pi^\w}{\de \marg_A \pi }(a) = \frac{\pi^a(\w)}{p(\w)}$ by Consistency of $\pi$. Hence, \eqref{eq:zero_cov} holds.

    Next, we show by construction that for every $\eta\in \D(\D(\W)\times \R)$ that satisfies \eqref{Eq: Bayes-plausibility}, \eqref{Eq: Obedience} and \eqref{eq:zero_cov}, there exists a communication equilibrium outcome $\pi$ with $\mathbb{E}_{\eta}[v] = \mathbb{E}_\pi[u_S]$. By \eqref{Eq: Obedience}, the conditional mean $\mathbb{E}_\eta[v \mid \mu]$ is a measurable selector of $\BV$. Hence, Lemma 2 of \cite{LR20} implies that there exists a measurable $\lambda:\D(\W)\to \D(A)$ such that for every $\mu\in \D(\W)$, $\lambda(\mu) \in \argmax_{\alpha\in \D(A)}\mathbb{E}_{\mu\times \alpha}[u_R(\w,a)]$ is a mixed best response for the receiver with posterior $\mu$, and $\int_A u_S(a) \de \lambda(\mu)(a) = \mathbb{E}_{\eta}[v\mid \mu ]$.

    Let $\t = \marg_{\D(\W)}\eta$. Define a probability kernel $\kappa: \Delta(\W) \to \W\times A$ by $\kappa(\mu,\cdot) = \mu \times \lambda(\mu)$, which is the product measure of $\mu\in\D(\W)$ and $\l(\mu)\in \D(A)$. By Lemma 3.1 of \cite{kallenberg2021}, $\kappa$ is well-defined since $\lambda$ is measurable. Let $\pi \coloneqq \tau \circ \kappa$, we show that $\pi$ is a desired communication equilibrium outcome. By construction, for every bounded measurable $u:\W\times A\to \R$,
    \begin{align} \label{Eq: change of var}
        \mathbb{E}_{\pi}[u] = \int_{\D(\W)} \int_{\W \times A} u(\w,a) \de \kappa(\mu, \w, a) \de \t(\mu) 
        =& \int_{\D(\W)} \mathbb{E}_{\mu \times \lambda(\mu)}[u(\w,a)]\de \t(\mu).
    \end{align}

    Therefore, $\mathbb{E}_{\eta}[v] = \mathbb{E}_\pi[u_S]$. By \eqref{Eq: change of var}, for every $\w\in \W$, $\pi(\w, A) = \int_{\D(\W)} \mu(\w) \lambda(\mu)(A) \de \t(\mu) = \int_{\D(\W)\times \R} \mu(\w)\de \eta(\mu,v) = p(\w)$, where the last equality follows from \eqref{Eq: Bayes-plausibility}. Hence, $\pi$ satisfies Consistency. 
    
    To see Obedience, take any measurable $\tilde{a}:A \to A$, by definition of $\lambda$, we have $\mathbb{E}_{\mu \times \lambda(\mu)}[u_R(\w,a)] \geq \mathbb{E}_{\mu\times (\tilde{a})_\# \lambda(\mu)}[u_R(\w,a)]$ for any $\mu \in \D(\W)$. Taking expectation with respect to $\t$ as in \eqref{Eq: change of var}, we have $\int u_R(\w,a) \de \pi(\w,a) \geq \int u_R(\w,\tilde{a}(a)) \de \pi(\w,a)$, and $\pi$ satisfies Obedience by Lemma \ref{Lem: Obedience}.

    Finally, $\pi$ satisfies Honesty since for every $\w \in \W$,
    \begin{align*}
        \mathbb{E}_{\pi^{\w}}[u_S] = \frac{1}{p(\w)} \int_{\D(\W)} \mu(\w) \mathbb{E}_{\lambda(\mu)}[u_S(a)] \de \t(\mu) = \frac{1}{p(\w)} \int_{\D(\W)} \mu(\w) \mathbb{E}_{\eta}[v\mid \mu] \de \t(\mu) = \mathbb{E}_{\pi}[u_S],
    \end{align*}
    where the first equality follows from \eqref{Eq: change of var} and $\pi(\w,A) = p(\w)$, the second follows from definition of $\lambda$, and the last one follows from \eqref{Eq: Bayes-plausibility} and \eqref{eq:zero_cov}. 
\end{proof}

\begin{proof}[Proof of Theorem \ref{Thm: BP = MD iff BP = CT}]
    (Only if) It is immediate. (If) We prove the contrapositive. If $\max \oV - \min \uV=0$, then the statement is obvious. Thus, we now assume that  $B := \max \oV - \min \uV>0$.
    If verifiability has no value, then there exists $\eta \in \TT_{MD}(p)$ with $\mathbb{E}_{\eta}[v] = \VV_{BP}(p)$. Let $\bar v = \mathbb{E}_{\eta}[v]$. Fix $\e \in (0, 1/B)$, define $\eta_\e \in \D(\D(\W)\times \R)$ by 
    \[
    \frac{\de \eta_\e}{\de \eta}(\mu,v) = 1 + \e(v - \bar v) > 0
    \]
    which is a well-defined probability distribution since the Radon-Nikodym derivative $1+\e(v-\bar v)$ is positive $\eta$-almost surely and integrates to one. By construction, $\eta_\e$ has the same support as $\eta$, so \eqref{Eq: Obedience} still holds. By construction,
    \[
    \mathbb{E}_{\eta_\e}[\mu] = \mathbb{E}_{\eta}[\mu] + \e \mathbb{E}_{\eta}[(v-\bar v)\mu] = p + \e \Cov_\eta[v, \mu] = p,
    \]
    where the last equality follows from \eqref{eq:zero_cov}. Therefore, $\eta_\e \in \TT_{BP}(p)$ is feasible for persuasion, with an expected value
    \[
    \mathbb{E}_{\eta_\e}[v] =  \mathbb{E}_{\eta}[v] + \e \mathbb{E}_{\eta}[(v-\bar v)v] =  \mathbb{E}_{\eta}[v] + \e \Var_\eta[v].
    \]
    Therefore, if $\Var_\eta[v] > 0$, then $\mathbb{E}_{\eta_\e}[v]> \mathbb{E}_{\eta}[v]$, contradicting the optimality of $\eta$. Hence, $\Var_\eta[v] = 0$ and $\eta$ is feasible under cheap talk.
\end{proof}

\begin{proof}[Proof of Proposition \ref{Prop: sufficient cond for BP > MD}]
    Suppose there exists $\mu \in H^*(p)$ such that $\oV_{CT}(\mu) > \oV_{CT}(p)$, then there exists $\eta \in \TT_{CT}(\mu)$ with $\mathbb{E}_\eta[v] > \oV_{CT}(p)$. As $\mu \in H^*(p)$, there exists $\mu_0 \in \D(\W)$ such that $\oV_{CT}(p) \in \BV_{CT}(\mu_0)$ and $p = \alpha \mu_0 + (1-\alpha) \mu$ for some $\alpha \in (0,1)$. Hence, there exists $\eta_0 \in \TT_{CT}(\mu_0)$ such that $\mathbb{E}_{\eta_0}[v] = \oV_{CT}(p)$, and $\eta' = \alpha \eta_0 + (1-\alpha) \eta \in \TT_{BP}(p)$. It follows that $\mathbb{E}_{\eta'}[v] = \alpha \mathbb{E}_{\eta_0}[v]  + (1-\alpha) \mathbb{E}_{\eta}[v] > \oV_{CT}(p)$. By Theorem \ref{Thm: BP = MD iff BP = CT}, this implies $\VV_{BP}(p) > \VV_{MD}(p)$.

    Conversely, suppose $\oV_{CT}(p) \geq \max \oV$. Since $\VV_{CT}(p) \leq \VV_{BP}(p) \leq \max \oV$, we have $\VV_{CT}(p) = \VV_{MD}(p)= \VV_{BP}(p)$, so verifiability is not valuable at $p$. 

    Finally, if the full-dimensionality condition holds at $p$, then $\max_{\mu \in H^*(p)} \oV_{CT}(\mu) = \max_{\mu \in \D(\W)} \oV_{CT}(\mu) = \max \oV$. So the necessary and sufficient conditions coincide.
\end{proof}

Recall that $\BV_{CT}: \D(\W)\rightrightarrows \R$ is the correspondence of the sender's payoff under some cheap-talk equilibrium with prior $\mu \in \D(\W)$. 
By Corollary 3 and Section C.2.1 of \cite{LR20}, $\BV_{CT}$ is non-empty, convex, and compact-valued. 
Moreover, by Theorem 1 in \cite{LR20}, $s \ge \oV(p)$ is attainable under cheap talk if and only if $p \in \co\left\{\oV\ge s\right\}$, where we use $\left\{\oV\ge s\right\}$ to denote $\{\mu\in\D(\W): \oV(\mu)\geq s\}$.

We start with a useful lemma that extends Theorem 1 in \cite{LR20}.\footnote{Theorem 1 of \cite{LR20} establishes the weak inequality versions of the first equivalence in Lemma \ref{Lem: Envelope and cvx hull}. We extend this result to strict inequalities.}

\begin{lemma} \label{Lem: Envelope and cvx hull}
For every $s \in \R$, $\oV_{CT}(p) > s$ if and only if $p \in \co\set{\oV > s}$, and $\uV_{CT}(p) < s$ if and only if $p \in \co\set{\uV < s}$.
\end{lemma}

This lemma implies that there exists a cheap-talk equilibrium that attains a \emph{strictly} higher (lower) value than $s$ if and only if the prior is in the convex hull of posteriors with highest (lowest) value strictly above (below) $s$. The proof is relegated to Appendix \ref{App: Lemmas}.

\begin{proof}[Proof of Theorem \ref{Thm: MD vs CT}] Let $s = \oV_{CT}(p)$ in this proof. \leavevmode
    \paragraph{First Statement.} We show this statement by an explicit construction. Suppose cheap talk is hull-directionally improvable at $p$. By definition, there exists $(\mu^+,v^+),(\mu^-,v^-) \in \Gr(\BV_{CT})$ such that $\mu^+ \in (p, \mu^-)$ and $v^+ > s > v^-$. Therefore, there exists $\lambda\in(0,1)$ such that $\mu^+ = \l \mu^- + (1-\l) p$, $\eta^+ \in \TT_{CT}(\mu^+)$ that attains value $v^+$, and $\eta^-\in \TT_{CT}(\mu^-)$ that attains value $v^-$. 

    Let $\xi \coloneqq \tfrac{\tfrac{1}{\lambda}(s-v^-)}{v^+-s + \tfrac{1}{\lambda}(s-v^-)}$. Then,
    \begin{align*}
        \mathbb{E}_{(\xi \eta^+ + (1-\xi) \eta^-)}\left[(v - s) (\mu - p) \right]
        = \xi (v^+-s) (\m^+-p) - (1-\xi) (s-v^-) (\m^- - p)
        = \mathbf{0}.
    \end{align*}
    Moreover, since $\mu^- \in H^*(p)$, there exists $\mu_0\in \D(\W)$ such that $s \in \BV_{CT}(\mu_0)$ and $p\in (\mu_0, \mu^-]$. As $\mu^- \neq p$, there exists $\alpha \in (0,1)$ such that $p = (1-\alpha)  \mu_0 + \alpha (\xi \mu^+ +(1-\xi) \mu^-)$. Since $s \in \BV_{CT}(\mu_0)$, there exists $\eta_0 \in \TT_{CT}(\mu_0)$ that attains $s$.

    Next, consider $\tilde{\eta} \coloneqq (1-\alpha) \eta_0 + \alpha \xi \eta^+ + \alpha(1-\xi)\eta^-$. By construction, $\tilde{\eta}$ satisfies \eqref{Eq: Bayes-plausibility} and \eqref{Eq: Obedience}. It also satisfies \eqref{eq:zero_cov} since 
    \begin{align*}
        \mathbb{E}_{\tilde{\eta}}\left[v(\mu - p)\right] = s \mathbb{E}_{\tilde{\eta}}\left[\mu - p\right] + \alpha \mathbb{E}_{(\xi \eta^+ + (1-\xi) \eta^-)}\left[(v - s) (\mu - p) \right] = \mathbf{0},
    \end{align*}
    where the last equality is by \eqref{Eq: Bayes-plausibility} and our construction of $\xi$. With this, we have
    \begin{equation*}
    \label{eq:bound MD-CT}
\mathbb{E}_{\tilde{\eta}}[v] = s + \alpha \xi (v^+-s) + \alpha (1-\xi) (v^--s)  = s + \alpha(\tfrac{1}{\lambda}-1) \tfrac{(v^+-s) (s-v^-)}{v^+-s + \tfrac{1}{\lambda} (s-v^-)} > s.
    \end{equation*}

\paragraph{Second statement.}
Let $s=\oV_{CT}(p)$, and suppose mediation is valuable at $p$. Then there exists $\eta \in \TT_{MD}(p)$ such that $\mathbb{E}_\eta [v] > s$. By Remark \ref{Rmk: Existence}, we may take $\eta$ to have finite support. Let $H\coloneqq \{(\mu,v) \in \D(\W)\times \R: v > s\}$ and $L\coloneqq \{(\mu,v) \in \D(\W)\times \R: v < s\}$.

Define
\[
A=\int_H (v-s)\,\de\eta(\mu,v), \qquad
B=\int_L (s-v)\,\de\eta(\mu,v) .
\]
Since $\mathbb{E}_\eta[v]>s$, we have $A>B\geq 0$. 

We first show that $B>0$. Suppose not, then
$B=0$, so $v\geq s$ $\eta$-almost surely and $A>0$. By \eqref{eq:zero_cov},
\[
\mathbf{0} = \int_{\D(\W)\times \R} (v-s) (\m-p)\de \eta(\mu, v) = \int_H (v-s) (\m-p)\de \eta(\mu, v).
\]
Therefore,
\[
p = 
\frac{1}{A}\int_H (v-s)\mu\,\de\eta(\mu,v).
\]
For every $(\mu,v)\in H\cap\supp(\eta)$, \eqref{Eq: Obedience} implies
$v\in \BV(\mu)$, and hence $\oV(\mu)>s$. Since $\eta$ has finite support, we have 
$p \in \co\{\oV>s\}$. By Lemma \ref{Lem: Envelope and cvx hull},
this implies $\oV_{CT}(p)>s$, contradicting $s=\oV_{CT}(p)$. Hence, $B>0$.

Now define
\[
\bar\mu^+
\coloneqq
\frac{1}{A}\int_H (v-s)\mu\,\de\eta(\m,v), \qquad
\bar\mu^-
\coloneqq
\frac{1}{B}\int_L (s-v)\mu\,\de\eta(\m,v) .
\]
By \eqref{Eq: Bayes-plausibility} and \eqref{eq:zero_cov}
\[
\mathbf 0=\int_{\D(\W)\times \R} (v-s)(\mu-p)\,\de\eta(\m,v)
  =A(\bar\mu^+-p)-B(\bar\mu^- -p).
\]
Therefore
\[
\bar\mu^+-p=\frac{B}{A}(\bar\mu^- -p).
\]
As before, since $\eta$ has finite support, $\bar \mu^+$ is a convex combination of posteriors $\m$ with $\oV(\mu) > s$. Lemma \ref{Lem: Envelope and cvx hull} then implies $\oV_{CT}(\bar\m^+) > s$. Similarly, $\uV_{CT}(\bar\mu^-)<s$. Moreover, $\bar \mu^+ \neq p$ as otherwise it contradicts $s = \oV_{CT}(p)$. Since $0 < B/A < 1$, we have $\bar\mu^+\in (p,\bar\mu^-)$ and thereby cheap talk is directionally improvable at $p$.
\end{proof}

\begin{proof}[Proof of Proposition \ref{prop: opt MD binary}]
    Normalize $\oV_{CT}(p) = 0$, we prove the result in two steps.

\textbf{Step 1.} There is an optimal solution $\eta^* = \sum_{i=1}^3 \eta_i^* \delta_{(\mu_i^*,v_i^*)}$ such that $\mu_1^* < p < \mu_2^* < \mu_3^*$ and $v_3^* < v_1^* \leq 0 < v_2^*$.

Since $\W$ is binary, Remark \ref{Rmk: Existence} implies that an optimal solution exists with support size at most three; atoms with the same posterior can be merged because $\BV$ is convex-valued. Any feasible distribution in \eqref{Eq: MD problem} with at most two distinct posteriors is feasible under cheap talk: if its support is $\{(\mu_1,v_1),(\mu_2,v_2)\}$ with $\mu_1<p<\mu_2$, then \eqref{Eq: Bayes-plausibility} and \eqref{eq:zero_cov} imply $\mathbb E[(v-v_2)(\mu-p)]=0$, hence $v_1=v_2$. Since $\VV_{MD}(p)>\VV_{CT}(p)$, every optimal solution has at least three distinct posteriors.

Let $\eta^* = \sum_{i=1}^3 \eta_i^* \delta_{(\mu_i^*,v_i^*)}$ be an optimal solution with three distinct posteriors $\mu_1^*<\mu_2^*<\mu_3^*$. By \eqref{Eq: Bayes-plausibility}, $\mu_1^*<p<\mu_3^*$. If $\mu_2^*\leq p$, then the normalization in subsection \ref{ssec: opt MD binary state} implies $v_1^*,v_2^*\leq 0$. Since the solution yields a value above $0$, we must have $v_3^*>0$. But then $v_i^*(\mu_i^*-p)\geq 0$ for all $i$, with strict inequality for $i=3$, contradicting \eqref{eq:zero_cov}. Hence, $\mu_1^*<p<\mu_2^*<\mu_3^*$.

\eqref{Eq: Bayes-plausibility} and \eqref{eq:zero_cov} imply
\begin{align} \label{eq: system binary state}
    \eta_1^*(p-\mu_1^*)=\eta_2^*(\mu_2^*-p)+\eta_3^*(\mu_3^*-p), \quad v_1^*\eta_1^*(p-\mu_1^*)
    =
    v_2^*\eta_2^*(\mu_2^*-p)+v_3^*\eta_3^*(\mu_3^*-p).
\end{align}
Thus, $v_1^*$ is a weighted average of $v_2^*$ and $v_3^*$. Since the solution is not feasible under cheap talk, the values are not all equal, so either $v_3^*<v_1^*<v_2^*$ or $v_2^*<v_1^*<v_3^*$.

The second case is impossible. If $v_2^*<v_1^*<v_3^*$, then $v_1^*\le 0$ because $\mu_1^*<p$. By \eqref{eq: system binary state},
\begin{align*}
    v_1^* = v_2^* + (v_3^* - v_2^*) \frac{\eta_3^* (\mu_3^*-p)}{\eta_2^*(\mu_2^*-p) + \eta_3^*(\mu_3^*-p)} > v_2^* + (v_3^* - v_2^*) \frac{\eta_3^*}{\eta_2^*+ \eta_3^*},
\end{align*}
where the strict inequality follows from $\mu_3^*-p>\mu_2^*-p >0$. Therefore, $\sum_{i=1}^3\eta_i^* v_i^* < v_1^*\leq 0,$ contradicting strict improvement. Hence $v_3^*<v_1^* \leq 0 <v_2^*$. 

\textbf{Step 2.} An optimal solution with the form obtained in Step 1 also satisfies the frontier requirements in the proposition.

For any $((\mu_i,v_i))_{i=1}^3 \in ([0,1]\times\R)^3$ with $\mu_1 < p < \mu_2 < \mu_3$ and $v_3 < v_1 < v_2$, there exists a unique solution $(\eta_i)_{i=1}^3$ with $\sum_i \eta_i = 1$ to the system \eqref{eq: system binary state}. This pins down a distribution $\sum_{i=1}^3 \eta_i \delta_{(\mu_i,v_i)}$ which is feasible under \eqref{Eq: MD problem} if each $(\mu_i,v_i)\in \Gr(\BV)$.

Let $S((\mu_i,v_i)_{i=1}^3)\coloneqq \sum_{i=1}^3\eta_i v_i$ denote the induced objective value. A direct calculation, reported in Appendix \ref{sec: OLA Binary}, shows that, on this domain,
\[
    \frac{\partial S}{\partial \mu_1}<0,\qquad
    \frac{\partial S}{\partial \mu_2}<0,\qquad
    \frac{\partial S}{\partial \mu_3}>0,
\]
and
\[
    \frac{\partial S}{\partial v_1}>0,\qquad
    \frac{\partial S}{\partial v_2}>0,\qquad
    \frac{\partial S}{\partial v_3}<0.
\]

Take the optimal solution $\eta^*$ from Step 1. If $v_1^*<\oV(\mu_1^*)$, then replacing
$v_1^*$ by $\oV(\mu_1^*)$ preserves the ordering of values and strictly increases $S$,
a contradiction. Thus $v_1^*=\oV(\mu_1^*)$. Similarly, $v_2^*=\oV(\mu_2^*)$ and $v_3^*=\uV(\mu_3^*).$

If there were
$\mu<\mu_1^*$ with $\oV(\mu)\ge v_1^*$, then replacing
$(\mu_1^*,v_1^*)$ by $(\mu,\oV(\mu))$ would strictly increase $S$, since
$S$ decreases in $\mu_1$ and increases in $v_1$. Hence, $\oV(\mu)<v_1^*$ for all $\mu < \mu_1^*$. Similarly, $\oV(\mu)<v_2^*$ for all $\mu\in(p,\mu_2^*)$ and $\uV(\mu)>v_3^*$ for all $\mu>\mu_3^*$.
\end{proof}

We next state and prove a lemma deriving the payoff correspondence $\BV$ considered in the lobbying application of Section \ref{ssec:lobbying} in the main text.

\begin{lemma}
\label{lm:lobbying}
    Action $i > 0$ is a best response for the receiver if and only if $\mu \in D_i = \co(\{\delta_i\}\cup I_i)$. The receiver is indifferent between $i$ and $0$ if and only if $\mu \in I_i$.
\end{lemma}
\begin{proof}
    The receiver is indifferent between actions $i$ and $0$ at $\mu \in \Delta(\W)$ if and only if $\lambda = \mathbb{E}_\mu [\phi(|\w-i|)]$, and the indifference set is pinned down by the intersection of this hyperplane and $\Delta(\W)$. By construction, the receiver is indifferent between $i$ and $0$ at every binary belief $\beta_{i}^k$ with $k\neq i$. This collection of $n-1$ affinely independent beliefs spans $I_i$, hence it is exactly where the receiver is indifferent between $i$ and $0$. Direct calculation shows that the receiver weakly prefers $i$ to $0$ if and only if $\mu \in D_i$.

    Next, we show that for any $\mu \in D_i$, the receiver strictly prefers $i$ to $j \neq 0$ as well. That is, $\mathbb{E}_{\mu}[\phi(|\w-i|)] < \mathbb{E}_{\mu}[\phi(|\w-j|)]$. It suffices to show each extreme point of $D_i$ satisfies this. $\delta_i$ clearly is. For every $\beta_{i}^k$, $\mathbb{E}_{\beta_{i}^k}[\phi(|\w-i|)] = \lambda$ and 
    \begin{align*}
        \mathbb{E}_{\beta_{i}^k}[\phi(|\w-j|)] =&\, \tfrac{\lambda}{\phi(|i-k|)}\phi(|k-j|) + (1-\tfrac{\lambda}{\phi(|i-k|)}) \phi(|i-j|) \\
        > & \, \tfrac{\lambda}{\phi(|i-k|)}\phi(|k-j|) + (1-\tfrac{\lambda}{2\lambda}) \phi(|i-j|) \\ 
        > & \, \tfrac{\lambda}{\phi(|i-k|)}\phi(|k-j|) + \tfrac{1}{2} 2\lambda \geq \lambda,
    \end{align*}
    where the inequalities hold because $i\neq j$ and $i\neq k$, so $\phi(|i-j|) \geq \phi(1) > 2\lambda$ and $\phi(|i-k|) > 2\lambda$. 

    The preceding arguments show that $i$ is the receiver's best response if $\mu \in D_i$. The only if direction then follows from the fact that if $\mu \notin D_i$, then $0$ strictly dominates $i$. 
\end{proof}

We now prove the main result of Section \ref{ssec:lobbying}.

\begin{proof}[Proof of Proposition \ref{prop: lobbying}]
  If the prior $p \in S_n$, no disclosure leads to the global maximum value $n$, so $\mathcal{V}_{BP}(p) = \mathcal{V}_{MD}(p) =\mathcal{V}_{CT}(p)$. If $p \in S_{n-1} \setminus S_n$ is non-boundary, then $\oV_{CT}(p) = n-1$ and the full-dimensionality condition holds at $p$. For any $(\mu^+,v^+), (\mu^-,v^-) \in \Gr(\BV)$ with $v^+ > \oV_{CT}(p)$ and $\mu^+ \in (p,\mu^-)$, we have $\mu^+ \in S_n$ and $v^+ = n$, and thereby $\mu^-$ is in  $S_n \setminus I_n$, so $v^- = n$. Therefore, cheap talk is not directionally improvable at $p$, so $\mathcal{V}_{BP}(p) > \mathcal{V}_{MD}(p) =\mathcal{V}_{CT}(p)$.

    To show (iii), suppose $p \in \D(\W) \setminus S_{n-1}$ is non-boundary, so the full-dimensionality condition holds and $\oV_{CT}(p) < n-1$. Let $\mu^-\in \D(\W)$ be the belief supported on $\{n-1, n\}$ with $\mu^-(n-1) = \mu^-(n) = 1/2$. Since $\phi(1) > 2\l$, we have $\beta_{n-1}^n(n) = \frac{\l}{\phi(1)} < \mu^-(n)$ and $\beta_{n}^{n-1}(n-1) = \frac{\l}{\phi(1)} < \mu^-(n-1)$. Hence, $\mu^-$ is not in $D_{n-1}$ or $D_n$ (and $\mu^-(i)=0$ rules out $D_i$ for $i\leq n-2$), and thereby $\uV_{CT}(\mu^-) = 0 <  \oV_{CT}(p)$. Direct calculation shows that there exists a small enough $\alpha \in (0,1)$ such that $\mu^+ = \alpha p + (1-\alpha)\mu^-$ lies in
    $S_{n-1}$.
    Therefore, $\oV_{CT}(\mu^+) = n-1 > \oV_{CT}(p)$, and cheap talk is directionally improvable at $p$. By Proposition \ref{Prop: sufficient cond for BP > MD} and Theorem \ref{Thm: MD vs CT}, we have $\mathcal{V}_{BP}(p) > \mathcal{V}_{MD}(p) >\mathcal{V}_{CT}(p)$. 
\end{proof}

\begin{proof}[Proof of Proposition \ref{pro:accep_Pareto}]
First, observe that $\psi(z)\coloneqq H(G^{-1}(z))$ is convex. By definition
    \begin{equation*}
        \psi(z)=H(G^{-1}(z))= \int_0^1 \max\left\{G^{-1}(t), G^{-1}(z)\right\} \de t=G^{-1}(z)z+\int_z^1G^{-1}(t)dt,
    \end{equation*}
  hence that
  \begin{align*}
      \psi'(z) = z (G^{-1})'(z) = \frac{z}{g(G^{-1}(z))}=\frac{G(G^{-1}(z))}{g(G^{-1}(z))}=\frac{G}{g} \circ G^{-1}(z), 
  \end{align*}
  where $g$ is the density of $G$. By log-concavity of $G$, $\tfrac{G}{g}$ is increasing. It follows that $\psi' $ is the composition of two increasing functions, hence it is an increasing function. Moreover, $\psi$ is also strictly increasing since it is a composition of two strictly increasing functions.
    
    Since $G$ is strictly increasing, $v(x) = G(R(x))$ is strictly quasiconvex and minimally edge non-monotone given $T$. Because $\min_{x\in X_T} R(x) < \max_{x\in X}R(x)$, we have the following observation, the proof of which is given in Online Appendix \ref{ssec:mean_meas}.
    \begin{claim}\label{claim: existence of full-dim region}
        There exists an $(n-1)$--simplex $\tilde{\D} \subseteq \D(\W)$ such that for all $p \in \interior \tilde{\D}$, $\VV_{MD}(p) > \VV_{CT}(p)$.
    \end{claim}
    
    Fix a prior $p \in \interior \tilde{\D}$, and let $\tau^* \in \D(\D(\W))$ be a sender-optimal feasible distribution over posteriors under mediation. Then, for every $\tau$ feasible under cheap talk,
    \[
    \mathbb{E}_{\t^*}[v_R\circ T] = \mathbb{E}_{\t^*}[\psi \circ v\circ T] \geq \psi(\mathbb{E}_{\t^*}[v\circ T]) > \psi(\mathbb{E}_{\t}[v\circ T]) = \mathbb{E}_{\t}[\psi \circ v\circ T] = \mathbb{E}_{\t}[v_R\circ T].
    \]
    The first inequality follows from Jensen's inequality, the strict inequality follows from the strict monotonicity of $\psi$ and $\VV_{MD}(p) > \VV_{CT}(p)$, and the second equality from the fact that $v\circ T$ must be constant over the support of $\t$ since it is feasible under cheap talk. Hence, mediation is strictly (ex-ante) Pareto improving at every $p \in \interior \tilde{\D}$.
\end{proof}

\begin{proof}[Proof of Proposition \ref{prop: project investment}]
    (i) Recall that $R(x) = \frac{1}{2}\|x\|^2$, so $\underline{\w} \in \arg\min_{\w\in \W} \|\w\|$ if and only if $\underline{\w} \in \arg\min_{x\in X_{T}}R(x)$. For any $\underline{\w} \in \arg\min_{\w\in \W}\|\w\|$ and $\w \neq \underline{\w}$, the one-dimensional function
    \begin{align*}
        \hat{R}_{\w}(\l) = R(\l \w + (1-\l) \underline{\w}) = \frac{1}{2}(\|\w - \underline{\w}\|^2 \l^2 + 2(\langle \w, \underline{\w}\rangle - \|\underline{\w}\|^2)\l + \|\underline{\w}\|^2),
    \end{align*}
    is non-monotone in $\l \in [0,1]$ if and only if $(\|\underline{\w}\|^2 - \langle\w , \underline{\w}\rangle) /\|\w - \underline{\w}\|^2 \in (0,1)$, which is equivalent to $\langle \w, \underline{\w}\rangle < \|\underline{\w}\|^2$. By definition, $R$ is minimally edge non-monotone if and only if there exists $\underline{\w} \in \arg\min_{\w\in\W} \|\w\|$ such that $\langle \w, \underline{\w}\rangle < \|\underline{\w}\|^2$ for every state $\w \neq \underline{\w}$. (ii) then follows immediately from Proposition \ref{pro:accep_Pareto}. To see (iii), note that suppose each $\w \in \W$ is a scalar multiple of some coordinate vector and $\underline{\w} \in \arg\min_{\w\in\W} \|\w\|$ is not Pareto dominated, then no other state has its positive coordinate in the same project as $\underline{\w}$. Therefore, $\langle \w, \underline{\w} \rangle = 0$ for every $\w \neq \underline{\w}$, so (i) holds.
\end{proof}

\begin{proof}[Proof of Theorem \ref{th:impossibility}]
    The if direction is immediate. For the only if direction, suppose a feasible $\eta \in \D(\D(\W)\times \R)$ is Pareto optimal and ex-ante fair. Let $\bar v = \mathbb{E}_{\eta}[v]$ and we focus on the nontrivial case $B = \max \oV - \min \uV > 0$. 
    As in the proof of Theorem \ref{Thm: BP = MD iff BP = CT}, fix $\e \in (0, 1/B)$, and define $\eta_\e \in \D(\D(\W)\times \R)$ by $ \frac{\de \eta_\e}{\de \eta}(\mu,v) = 1 + \e(v - \bar v) > 0$.
    This is a feasible allocation: it satisfies \eqref{Eq: Bayes-plausibility} and \eqref{Eq: Obedience}. The payoff change for trait $\w$ is:   
    \begin{align*}
        U_\w(\eta_\e) - U_\w(\eta) = &\;\mathbb{E}_{\eta_\e}\left[v\frac{\mu(\w)}{p(\w)}\right] - \mathbb{E}_{\eta}\left[v\frac{\mu(\w)}{p(\w)}\right] 
        = \frac{\e}{p(\w)} \mathbb{E}_{\eta}\left[(v-\bar v)v\mu(\w)\right] \\
        = &\; \frac{\e}{p(\w)} \mathbb{E}_{\eta}\left[(v-\bar v)^2\mu(\w)\right] + \frac{\e \bar v}{p(\w)} \mathbb{E}_{\eta}\left[(v-\bar v)\mu(\w)\right],
    \end{align*}
    where the first equality follows from the definition, the second equality is by construction of $\eta_\e$, and the last one is simply rewriting. Note that $\mathbb{E}_{\eta}\left[(v-\bar v)\mu(\w)\right] = 0$ by ex-ante fairness \eqref{eq:zero_cov}. Therefore,
    \[
        U_\w(\eta_\e) - U_\w(\eta) = \frac{\e}{p(\w)} \mathbb{E}_{\eta}\left[(v-\bar v)^2\mu(\w)\right] \geq 0.
    \]
    Moreover, if $\eta$ does not satisfy ex-post fairness, then
    \[
     \mathbb{E}_p[U_\w(\eta_\e) - U_\w(\eta)] = \e \mathbb{E}_{\eta}\left[(v-\bar v)^2\right] > 0, 
    \]
    so at least one trait strictly benefits, and $\eta_\e$ Pareto dominates $\eta$, a contradiction. 
\end{proof}

\section{Optimal mediation with binary states} \label{sec: OLA Binary}
In this section, we provide some calculation details for the proof of Proposition \ref{prop: opt MD binary} and an additional first-order necessary condition when $\BV$ is singleton-valued and differentiable. Take any triple of posterior-value pairs $\{(\mu_i,v_i)\}_{i=1}^3 \in ([0,1]\times\R)^3$ with $\mu_1 < p < \mu_2 < \mu_3$ and $v_3 < v_1 < v_2$. Consider the following system of equations
\begin{align*}
        \underbrace{\begin{pmatrix}
            1 & 1 & 1\\
            \m_1 & \m_2 & \m_3 \\
            v_1(\m_1-p) & v_2 (\m_2-p) & v_3(\m_3-p)
        \end{pmatrix}}_{A}
        \begin{pmatrix}
            \eta_1 \\
            \eta_2 \\
            \eta_3
        \end{pmatrix}
        =
        \begin{pmatrix}
            1 \\
            p \\
            0
        \end{pmatrix}
\end{align*}
Note that $\det A = (v_3 - v_1)(\m_3-p)(\m_2 - \m_1) - (v_2 - v_1)(\m_2 - p) (\m_3 - \m_1) \neq 0$ since $v_3 - v_1$ and $v_2 - v_1$ are of different signs. Cramer's rule yields a unique solution: 
    \begin{align*}
        \eta_1 = \frac{(v_3 - v_2)(\m_3 - p)(\m_2 - p )}{\det A}, 
        \eta_2 = \frac{(v_1 - v_3)(\m_1 - p)(\m_3 - p )}{\det A}, 
        \eta_3 = \frac{(v_2 - v_1)(\m_2 - p)(\m_1 - p )}{\det A}.
    \end{align*}
Since $v_3 < v_1 < v_2$, $\eta_i > 0$ for all $i=1,2,3$. Therefore, each $\{(\mu_i,v_i)\}_{i=1}^3 \in ([0,1]\times\R)^3$ with $\mu_1 < p < \mu_2 < \mu_3$ and $v_3 < v_1 < v_2$ induces a unique distribution $\sum_{i=1}^3 \eta_i \delta_{(\mu_i,v_i)}$ which is feasible under \eqref{Eq: MD problem} if each $(\mu_i,v_i)\in \Gr(\BV)$.
The corresponding expected value is $S((\mu_i,v_i)_{i=1}^3) = \sum_{i=1}^3 \eta_i v_i$. Substituting the solution $\eta_i$ above into $S$, we have 
\begin{align*}
        \frac{\partial S}{\partial \mu_i} = \frac{- (v_1-v_2)(v_1-v_3)(v_2-v_3)(p-\m_{i+1})(p-\m_{i+2})(\m_{i+1}-\m_{i+2})}{\left[(v_1 - v_2)(\m_1 \m_2 + p \, \m_3) + (v_3 - v_1)(\m_1 \m_3 + p \, \m_2) + (v_2 - v_3)(\m_2 \m_3 + p \, \m_1) \right]^2}\\
        \frac{\partial S}{\partial v_i} = \frac{(v_{i+1}-v_{i+2})^2 (\m_i - \m_{i+1})(\m_i - \m_{i+2}) (p-\m_{i+1}) (p-\m_{i+2})}{\left[(v_1 - v_2)(\m_1 \m_2 + p \, \m_3) + (v_3 - v_1)(\m_1 \m_3 + p \, \m_2) + (v_2 - v_3)(\m_2 \m_3 + p \, \m_1) \right]^2}
\end{align*}
for $i = 1,2,3$, where the subscripts are considered in mod 3. This implies the desired signs of the partial derivative in Step 2 of the proof of Proposition \ref{prop: opt MD binary}.

Recall the maintained assumption in subsection \ref{ssec: opt MD binary state}, that is, $\oV(\mu) \leq \oV_{CT}(p)$ for all $\mu \in [0,p]$ and mediation strictly improves on cheap talk. If we further assume that $\BV = V$ is singleton-valued and differentiable, the solution of \eqref{Eq: MD problem} also satisfies the following first-order necessary conditions.
\begin{proposition} \label{Prop: opt MD cont}
    \eqref{Eq: MD problem} admits a solution $\eta^*$ with $\supp \eta^* = \{(\mu_i, V(\m_i))\}_{i=1}^3$ such that
    \begin{enumerate}
        \item $\mu_1 < p < \mu_2 < \mu_3$ and $V(\m_3) < V(\m_1) \leq \oV_{CT}(p) < V(\m_2)$;
        \item for every $i \in \{1,2,3\}$ and $\{j,k\} = \{1,2,3\} \setminus \{i\}$ with $j<k$,
        \[
            (\mu_k - \mu_j) (V(\m_i) - V(\m_k)) (V(\m_i) - V(\m_j)) + (\m_j - \m_i)(\m_i - \m_k) V'(\m_i)(V(\m_k)-V(\m_j)) = 0
        \]
        with the boundary modifications that ``$\,=0$'' is replaced by ``$\,\leq 0$'' when $\mu_1=0$ (for $i=1$) and ``$\,\ge 0$'' when $\mu_3=1$ (for $i=3$).
    \end{enumerate} 
\end{proposition}
\begin{proof}
    Part 1 follows from Proposition \ref{prop: opt MD binary}. Following the same calculation as before, each tuple $\{\mu_i\}_{i=1}^3 \in [0,1]^3$ with $\mu_1 < p < \mu_2 < \mu_3$ and $V(\m_3) < V(\m_1) \leq \oV_{CT}(p) < V(\m_2)$ induces a unique distribution $\sum_{i=1}^3 \eta_i\delta_{(\mu_i, V(\mu_i))}$ that is feasible under \eqref{Eq: MD problem}, where
    \[
    \eta_i = \frac{(V(\mu_{i+2}) - V(\mu_{i+1}))(\mu_{i+2} - p)(\mu_{i+1}-p)}{(V(\mu_3)-V(\mu_1)) (\mu_3-p) (\mu_2-\mu_1) - (V(\mu_2)-V(\mu_1))(\mu_2-p)(\mu_3-\mu_1)}
    \]
    and the subscripts are considered in mod 3. For every $i\in \{1,2,3\}$, differentiating $S$ after substituting $v_i = V(\mu_i)$, we obtain 
    \[
    \frac{\partial \mathbb{E}_\eta[V]}{\partial \mu_i} = \frac{\partial S}{\partial \mu_i} + V'(\mu_i) \frac{\partial S}{\partial v_i} = c_i F_i
    \]
    where $c_i \in \R$ and $F_i$ denotes the expression in part 2. Direct calculation shows that $c_1, c_3 > 0 > c_2$. Hence, at an interior optimum, $F_i = 0$ for each $i$. If $\mu_1 = 0$, optimality requires $\frac{\partial \mathbb{E}_\eta[V]}{\partial \mu_1} \leq 0$, equivalently, $F_1 \leq 0$. If $\mu_3 = 1$, optimality requires $F_3 \geq 0$.
\end{proof}
To see how the first-order conditions are used, consider a smooth version of the example in Section \ref{Subsec: Illustrative Example}. Let \(\BV=V\) be singleton-valued and differentiable, with \(V(0)=2\) and \(V(1)=1\). Suppose that \(V(\mu)<2\) on \((0,\bar\mu)\), crosses above \(2\) at \(\bar\mu\), increases up to a unique maximizer \(\hat\mu\), and then decreases toward \(V(1)=1\). For any \(p\in(0,\bar\mu)\), Propositions \ref{prop: opt MD binary} and \ref{Prop: opt MD cont} imply that an optimal solution can be chosen with support $\{(0,2),(\mu^+,V(\mu^+)), (1,1)\}$, where \(\mu^+\in(\bar\mu,\hat\mu]\) is pinned down by the first-order condition.

\newpage
\begin{bibunit}[ecta]

\setcounter{page}{1}
{\renewcommand{\thefootnote}{\fnsymbol{footnote}}
\begin{center}
    \Large{\textbf{Online Appendix to ``The Bounds of Mediated Communication''}} \\
    \large{\textbf{Roberto Corrao\footnote[1]{Department of Economics, Stanford, \href{mailto:rcorrao@stanford.edu}{rcorrao@stanford.edu}} and Yifan Dai\footnote[2]{Department of Economics, MIT, \href{mailto:yfdai@mit.edu}{yfdai@mit.edu}}}\\
    FOR ONLINE PUBLICATION ONLY}
\end{center}}

\section{Omitted proofs of technical lemma} \label{App: Lemmas}

\begin{proof}[Proof of Lemma \ref{Lem: Obedience}]
    The only if direction follows from the law of iterated expectations and the definition of Obedience. For every measurable $\tilde{a} : A\to A$, we have
    \begin{align*}
    \int_{\W\times A} u_R(\w,a) \de \pi(\w,a) &= \int_A \mathbb{E}_{\pi^a}[u_R(\w,a)] \de\marg_A \pi(a) \\
    &\geq \int_A \mathbb{E}_{\pi^a}[u_R(\w,\tilde{a}(a))] \de\marg_A \pi(a) =\int_{\W\times A} u_R(\w,\tilde{a}(a)) \de \pi(\w,a).
    \end{align*}
    For the if direction, suppose Obedience is not satisfied, then there exists a measurable $S \subseteq A$ with $\marg_A \pi (S) > 0$ such that $ \mathbb{E}_{\pi^a}[u_R(\w,a)] < \max_{a'\in A}  \mathbb{E}_{\pi^a}[u_R(\w,a')]$ for all $a \in S$. By the measurable maximum theorem \cite[Theorem 18.19]{aliprantis06}, there exists a measurable $\hat{a}:A \to A$ with $\hat{a}(a) \in \argmax_{a'\in A}  \mathbb{E}_{\pi^a}[u_R(\w,a')]$. We then have $\int u_R(\w,a) \de \pi(\w,a) < \int u_R(\w,\hat{a}(a)) \de \pi(\w,a)$, contradiction.
\end{proof}

\begin{proof}[Proof of Lemma \ref{Lem: Envelope and cvx hull}]
    For any $s \geq \oV(p)$, the first equivalence follows from Theorem 1 of \cite{LR20}. For the only if direction, suppose $\oV_{CT}(p) > s$, then there exists $\eta\in \TT_{CT}(p)$ that attains a value $s' > s$. Theorem 1 of \cite{LR20} implies that $p \in \co\set{\oV \geq s'} \subseteq \co \set{\oV > s}$. For the if direction, suppose that $p \in \co \set{\oV > s}$. Then there exists finitely many points $\set{\mu_i}_{i=1}^{k} \subseteq \set{\oV > s}$ such that $p = \sum_{i=1}^k \alpha_i \mu_i$ for some $\set{\alpha_i}_{i=1}^k \subseteq  [0,1]$, $\sum_{i=1}^k \alpha_i =1$. Let $\hat{s} \coloneqq \min_i \oV(\mu_i)$, so that $p \in \co \set{\oV \geq \hat{s}}$. Theorem 1 of \cite{LR20} then implies that $\oV_{CT}(p) \geq \hat{s} > s$.
    For any $s < \oV(p)$, the first equivalence is true as both $\oV_{CT}(p) \geq \oV(p) > s$ and $p\in \co\set{\oV > s}$ are true. The second equivalence follows from a symmetric argument.\footnote{See footnote 15 of \cite{LR20}.} 
\end{proof}

\section{Existence and value of optimal mediation}\label{App: existence and value}
Let $g \in \mathbb{R}^n$ denote an arbitrary Lagrange multiplier for \eqref{eq:zero_cov} such that $\langle g,p\rangle=1$  and define the corresponding \emph{virtual} indirect value function of the sender as 
$V^g(\mu):= \max_{v \in \BV(\mu)}v\langle g,\mu\rangle$.
Each $V^g(\mu)$ is the belief-based version of the \emph{virtual utility} in \cite{myerson1997game} and \cite{salamanca2021value} and, like those, takes into account a fixed shadow price $g$ of the constraint \eqref{eq:zero_cov}.\footnote{Recall that the virtual utilities in both \cite{myerson1997game} and \cite{salamanca2021value} are defined on outcomes as opposed to beliefs.}
We next use these objects to characterize the optimal value of mediation. For any measurable function $U:\D(\W) \to \R$, let $\cav(U)$ denote its concavification.

\begin{proposition} \label{Prop: value of MD}
The sender's optimal value under mediation is
    \begin{equation*}
        \VV_{MD}(p)=  \inf_{g \in \mathbb{R}^{n}:\langle g,p\rangle=1} \cav(V^g)(p).
    \end{equation*}
and \eqref{Eq: MD problem} admits a solution $\eta^*$ supported on no more than $2n-1$ points. 
\end{proposition}
\begin{proof}
    Let $I = [\min_{\mu\in\D(\W)} \uV(\mu), \max_{\mu\in \D(\W)} \oV(\mu)]$, where the min and max are well-defined by semi-continuity of $\oV$ and $\uV$. Since every feasible \(\eta\) is supported on \(\Gr(\BV)\), it is supported on \(\D(\W)\times I\). Thus it is enough to work on the compact space \(\D(\W)\times I\).

    We first show that the feasible set $\TT_{MD}(p)$ is closed in $\D(\D(\W)\times I)$ in the weak topology. Take a sequence $\{\eta_n\}$ in $\TT_{MD}(p)$ that converges weakly to $\eta$. By \eqref{Eq: Bayes-plausibility} and the continuity of the integrand, we have $p = \int \mu \de \eta_n \to \int \mu \de \eta$. By \eqref{Eq: alternative zerocov}, for every $\w \in \W$,
    \begin{align*}
        0 = \Cov_{\eta_n}[v, \mu(\w)] =  \int_{\D(\W)\times \R} v (\mu(\w) -p(\w)) \de \eta_n(\mu,v) \\
    \to \int_{\D(\W)\times \R} v (\mu(\w) -p(\w)) \de \eta(\mu,v) = \Cov_{\eta}[v, \mu(\w)],
    \end{align*}
    where the weak convergence follows from the continuity of the integrand. Since $\BV$ is upper hemi-continuous and closed-valued, $\Gr(\BV)$ is closed and hence $1 = \limsup_{n} \eta_n(\Gr(\BV)) \leq \eta(\Gr(\BV))$ by the Portmanteau Theorem. Therefore, $\eta(\Gr(\BV)) = 1$, and $\eta \in \TT_{MD}(p)$. Therefore \(\TT_{MD}(p)\) is closed. Since it is a closed subset of the compact space \(\D(\D(\W)\times I)\), it is compact. The maximum of \eqref{Eq: MD problem} is attained as the objective is continuous.

    Since the feasible set is also convex, Bauer's maximum principle implies that there exists a solution $\eta'$ which is an extreme point. Theorem 2.1 of \cite{Win88} then implies the size of the support of $\eta'$ is bounded by the number of linearly independent moment constraints plus one, that is,  $|\supp(\eta')|\leq 2(n-1)+1 = 2n-1$.

    Finally, rewrite the value of the problem using a Lagrange multiplier $g \in \R^n$ on the truth-telling constraint\footnote{For the Lagrangian approach to constrained information design, see \cite{doval2018constrained}.}
    \begin{equation}
    \label{eq:sup_inf}
         \sup_{\eta \in \TT_{BP}(p)} \inf_{g\in \R^n} \int_{\D(\W)\times \R} v (1+\langle g, \mu - p\rangle) \de \eta(\m,v) = \sup_{\eta \in \TT_{BP}(p)} \inf_{g\in \R^n: \langle g, p \rangle = 1} \int_{\D(\W)\times \R} v \langle g, \mu\rangle \de \eta(\m,v),
    \end{equation}
    where the equality holds because $\langle g + c \mathbf{1}, \mu - p\rangle = \langle g, \mu - p\rangle$ for any $g \in \R^n$, $\mu \in \D(\W)$ and $c \in \R$.\footnote{$\mathbf{1} \in \R^n$ denotes the vector with all entries equal to 1.}
    Let $M(\eta, g) \coloneqq \int v \langle g,\mu \rangle\de \eta$. It is continuous in $\eta$ for every $g$ and continuous in $g$ for every $\eta$. It is also affine in both variables. Note that the set $\TT_{BP}(p)$ is compact and convex, so we can apply Sion's minimax theorem to change the sup and inf in \eqref{eq:sup_inf}. Therefore, the value can be rewritten as $ \inf_{g\in \R^n: \langle g, p \rangle = 1} \sup_{\eta \in \TT_{BP}(p)} \int v \langle g, \mu\rangle \de \eta = \inf_{g\in \R^n: \langle g, p \rangle = 1} \cav(V^g) (p)$, where $V^g(\mu) = \max_{v\in \BV(\mu)} v \langle g, \mu \rangle$, and the last equality follows from \cite{kamenica2011bayesian}.
\end{proof}

\section{Moment-measurable mediation: Quasiconvex case} \label{ssec:mean_meas}
In this section, we apply the results from Section \ref{Sec: MD vs CT} to \emph{moment-measurable mediation}. For $1\leq k \leq n-1$, a $k$-dimensional moment is a linear function $T:\D(\W)\to \R^k$ such that the set of relevant moments $X = T(\D(\W))$ has dimension $k$. We assume that $\BV(\mu)= \{v(T(\mu))\}$ is singleton-valued for some continuous $v:\R^k \to \R$. Here, we focus on the multidimensional case ($k>1$) under the assumption that $v(x)$ is strictly quasiconvex. This is the main case considered in past works on multidimensional cheap talk under transparent motives \citep{chakraborty2010persuasion,LR20}.\footnote{Quasiconvex sender's utilities play an important role also in the informed information design model of \cite{koessler2021information}.} Throughout this section, we write $V = v\circ T$. Since $\BV$ is singleton-valued, it suffices to consider distributions over posteriors in $\D(\D(\W))$.

When $v(x)$ is strictly quasiconvex and the full-dimensionality condition holds at $p$, only two extreme cases can happen: 
\begin{theorem} \label{Thm: multidim quasiconvex}
    Assume that $\BV(\mu) = \{v(T(\mu))\}$ for some $k$-dimensional moment $T$ ($k\geq 2$) and continuous and strictly quasiconvex $v(x)$. If the full-dimensionality condition holds at $p$, then exactly one of these cases holds: 
    \begin{description}
        \item[(1)] $\max V = \VV_{BP}(p) = \VV_{MD}(p) = \VV_{CT}(p) > V(p)$;
        \item[(2)] $\max V > \VV_{BP}(p) > \VV_{MD}(p) >
        \VV_{CT}(p) > V(p)$.
    \end{description}
\end{theorem}
\begin{proof}
    By Corollary 6 of \cite{LR20}, when $T$ is multi-dimensional and $v$ strictly quasiconvex, no disclosure is suboptimal under cheap talk. Suppose the full-dimensionality condition holds at $p$, by Proposition \ref{Prop: sufficient cond for BP > MD}, $\VV_{BP}(p) = \VV_{MD}(p)$ if and only if $\set{V > \VV_{CT}(p)} = \emptyset$, which means that cheap talk attains the global maximum value. This leads to the dichotomy in the theorem statement: If $\max V = \VV_{CT}(p)$, then (1) holds trivially. It suffices to show  $\max V > \VV_{CT}(p)$ implies (2). 
    
    Note that if $\VV_{BP}(p) = \max V$, it must be the case that $V(\mu) = \max V$ for all $\mu$ in the support of any optimal distribution over posteriors under Bayesian persuasion, which implies $\VV_{BP}(p) = \VV_{CT}(p)$, yielding a contradiction. Hence, what remains to show is that $\VV_{MD}(p) > \VV_{CT}(p)$. 

    Let $D_+ = \set{x\in X: v(x) > \VV_{CT}(p)}$ and $D_- = \set{x\in X: v(x) < \VV_{CT}(p)}$, both are open by continuity of $v$. Since $\max V > \VV_{CT}(p)$, we have  $D_+ \neq \emptyset$. Take any open ball in $D_+$, there exist two points $x_1, x_2$ in this open ball such that $x_1, x_2$, and $T(p)$ are not colinear. Note that by strict quasiconvexity, no disclosure is suboptimal under cheap talk, so $T(p) \in D_-$. Moreover, there exists a unique $\lambda_i \in (0,1)$ such that $v(\lambda_i x_i + (1-\lambda_i) T(p)) = \VV_{CT}(p)$ for $i=1,2$ since $v$ is continuous and strictly quasiconvex. Here, existence follows from the intermediate value theorem, whereas strict quasiconvexity implies uniqueness. By strict quasiconvexity, $\tfrac{1}{2}(\lambda_1 x_1 + \lambda_2 x_2) + (1 - \tfrac{1}{2}(\lambda_1+\lambda_2)) T(p) \in D_-$. Since $D_-$ is open, there exists $\e > 0$ such that $\tfrac{1}{2}(\lambda_1(1+\e) x_1 + \lambda_2 (1+ \e) x_2) + (1 - \tfrac{1}{2}(\lambda_1+\lambda_2)(1 + \e)) T(p) \in D_-$.

    Let $x_i' = \lambda_i (1+\e) x_i + (1-\lambda_i (1+ \e)) T(p)$.
    Take any $\mu_i \in \D(\W)$ such that $T(\mu_i) = x_i'$ for $i = 1,2$. By construction, $T(\tfrac{1}{2}\mu_1 + \tfrac{1}{2}\mu_2) = \tfrac{1}{2}x_1' + \tfrac{1}{2}x_2' \in D_-$, so $V(\tfrac{1}{2}\mu_1 + \tfrac{1}{2}\mu_2) < \VV_{CT}(p)$. Let 
    \[
    \tilde{\mu}_i = \frac{1+\e/2}{1 + \e} \mu_i + \frac{\e/2}{1+\e} p
    \]
    for each $i = 1,2$. By linearity, $T(\tilde{\mu}_i) = \lambda_i (1+\e/2) x_i + (1-\lambda_i (1+ \e/2)) T(p) \in D_+$. Therefore, $V(\tilde{\mu}_i) > \VV_{CT}(p)$ for $i=1,2$ and hence $\oV_{CT}(\tfrac{1}{2}\tilde{\m}_1 + \tfrac{1}{2}\tilde{\m}_2) > \VV_{CT}(p)$ by Lemma \ref{Lem: Envelope and cvx hull}. It follows that cheap talk is directionally improvable at $p$, so $\VV_{MD}(p) > \VV_{CT}(p)$ by Theorem \ref{Thm: MD vs CT} and the full-dimensionality condition.
\end{proof}

While Theorem \ref{Thm: multidim quasiconvex} dramatically simplifies the comparison among communication protocols in the present setting, it still relies on the full-dimensionality condition. The minimally edge non-monotonicity condition we introduced in Definition \ref{def: minimally edge non-mon} ensures the existence of a non-trivial set of priors that satisfy full dimensionality when $v$ is strictly quasiconvex. Recall that $X_T = \left\{T(\delta_{\w}) \in \mathbb{R}^k: \w \in \W\right\}$.

\begin{proposition}
 \label{Pro: Suff cond for full-dim multidim quasi-cvx}
    Assume that $\BV(\mu) = \{v(T(\mu))\}$ for some $k$-dimensional moment $T$ ($k\geq 2$) and that $v(x)$ is continuous, strictly quasiconvex, and minimally edge non-monotone given $T$. Then there exists an $(n-1)$--simplex $\tilde{\D} \subseteq \D(\W)$ such that the full-dimensionality condition holds for all $p \in \interior\tilde{\D}$. For every such $p$, point (2) of Theorem \ref{Thm: multidim quasiconvex} holds if and only if $\min_{x\in X_T}v(x) < \max v$.
\end{proposition}
\begin{proof}
    Since $v$ is minimally edge non-monotone, there exists a state $\underline{\w} \in \argmin_{\w \in \W} V(\delta_{\w})$ such that for any $\w \in \W \setminus \set{\underline{\w}}$, $f_{\w}(\lambda) \coloneqq V(\lambda \delta_{\w} + (1-\lambda) \delta_{\underline{\w}})$ is not weakly increasing in $\lambda\in [0,1]$. 
    
    We show that $f_{\w}$ is strictly quasiconvex on $[0,1]$. Note that for any $\lambda \neq \lambda'\in [0,1]$
    \begin{align*}
        f_\w(\alpha \lambda + (1-\alpha) \lambda')
        &= v(\alpha T(\mu) + (1-\alpha)T(\mu'))\\
        &\leq \max \set{v(T(\mu)), v(T(\mu'))} = \max \set{f_{\w}(\lambda), f_{\w}(\lambda')},
    \end{align*}
    where $\mu = \lambda \delta_{\w} + (1-\lambda) \delta_{\underline{\w}}$, $\mu' = \lambda' \delta_{\w} + (1-\lambda') \delta_{\underline{\w}}$. The first equality is by definition and linearity of $T$, the inequality is by (strict) quasiconvexity of $v$, and the last equality is by definition. The inequality is strict if and only if $T(\mu) \neq T(\mu')$. Suppose $T(\mu) = T(\mu')$, then by linearity of $T$, $T(\delta_{\w}) = T(\delta_{\underline{\w}})$, which means $f_{\w}$ is a constant on $[0,1]$. This contradicts with the assumption that $f_{\w}$ is non-monotone, hence $T(\mu) \neq T(\mu')$ and $f_{\w}$ is strictly quasiconvex.

    As $f_\w$ is strictly quasiconvex and non-monotone, there must be a unique $\lambda_{\w} \in (0,1]$ such that $f_{\w} (\lambda_\w) = f_{\w}(0)$. Suppose $f_\w(\lambda) > f_\w(0)$ for all $\lambda > 0$, then there exists $\lambda_2 > \lambda_1 >0$ such that $f_{\w}(\lambda_1) > f_{\w}(\lambda_2) > f_{\w}(0)$ (otherwise $f_\w$ is weakly increasing). But $\lambda_1 \in (0,\lambda_2)$, so $f_{\w}(\lambda_1) > f_{\w}(\lambda_2) > f_{\w}(0)$ violates the strict quasiconvexity, yielding a contradiction. So there must be a $\hat{\lambda}_\w \in (0,1]$ such that $f_\w(\hat{\lambda}_{\w}) \leq f_{\w}(0)$. By continuity of $v$, there exists $\lambda_\w \in [\hat{\lambda}_\w, 1]$ such that $f_\w(\lambda_\w) = f_{\w}(0)$. The uniqueness is by strict quasiconvexity.

    The argument above holds for any $\w \in \W \setminus\set{\underline{\w}}$. Let $\mu_{\w} \coloneqq \lambda_\w \delta_{\w} + (1-\lambda_\w) \delta_{\underline{\w}}$, we have $V(\mu_{\w}) = V(\delta_{\underline{\w}})$ for any $\w \in \W \setminus\set{\underline{\w}}$. Set $\tilde{\D} \coloneqq \co \set{\delta_{\underline{\w}}, \set{\mu_{\w}: \w \in \W \setminus\set{\underline{\w}}}}$. This is an ($n-1$)--simplex as $\set{\delta_{\underline{\w}}, \set{\mu_{\w}: \w \in \W \setminus\set{\underline{\w}}}}$ is affinely independent with $n$ points. Moreover, for any $p \in \interior\tilde{\D}$, there is a cheap talk distribution over posterior beliefs that supports on $\set{\delta_{\underline{\w}}, \set{\mu_{\w}: \w \in \W \setminus\set{\underline{\w}}}}$ that attains $V(\delta_{\underline{\w}})$, so $\VV_{CT}(p) \geq V(\delta_{\underline{\w}})$. Since $v(\cdot)$ is strictly quasiconvex, the composition $V = v\circ T$ is quasiconvex, hence $V(\mu) \leq V(\delta_{\underline{\w}})$ for any $\mu \in \tilde{\D}$. This shows that $\set{V >  V(\delta_{\underline{\w}})}$ is contained in $\D(\W) \setminus \tilde{\D}$, which is convex by construction. By Lemma \ref{Lem: Envelope and cvx hull}, $\VV_{CT}(p) \leq V(\delta_{\underline{\w}})$ for any $p\in \tilde{\D}$. Therefore, the full-dimensionality condition holds for all priors $p \in \interior\tilde{\D}$ as $\VV_{CT}$ is locally constant.
    Moreover, if $ V(\delta_{\underline{\w}}) < \max_{\mu\in\D(\W)} V(\mu)$, then for any $p\in \interior\tilde{\D}$, $\VV_{CT}(p) < \max V$. As the full-dimensionality condition holds, Theorem \ref{Thm: multidim quasiconvex} yields that  $\max V > \VV_{BP}(p) > \VV_{MD}(p) > \VV_{CT}(p) > V(p)$.
\end{proof}

\begin{proof}[Proof of Claim \ref{claim: existence of full-dim region}]
    Since $G$ is strictly increasing and $R$ is strictly quasiconvex and minimally edge non-monotone given $T$, the composition $v = G \circ R$ is strictly quasiconvex and minimally edge non-monotone given $T$. Since the trivial case was excluded ($\min_{x\in X_T} R(x) < \max_{x\in X} R(x)$), the claim then follows from Proposition \ref{Pro: Suff cond for full-dim multidim quasi-cvx}.
\end{proof}

\section{Correlated equilibria in long cheap talk and repeated games}
\label{app:repeated}
In this appendix, we discuss the implications of our results for the comparison of correlated and Nash equilibria in long cheap talk and repeated games with asymmetric information where the sender's payoff is state independent.

Fix a finite set of states $\W$, a finite action set $A$, and utility functions $u_R(\w,a)$ and $u_S(a)$ for the receiver and the sender respectively. Following the notation in \cite{forges2020games}, let $DP_0(p)$ denote the basic decision problem described by the previous primitive objects.

The long cheap talk game is an extension of the basic decision problem $DP_0(p)$ by allowing the sender and receiver to exchange messages simultaneously for several rounds before the receiver takes an action. Formally, let two finite sets $M_S$ and $M_R$ be the sender's and receiver's message spaces, respectively. Following \cite{LR20}'s notation, we let $H_{<\infty} \coloneqq \bigsqcup_{t=0}^{\infty} (M_S \times M_R)^t$ and $H_{\infty} \coloneqq (M_S \times M_R)^\mathbb{N}$.  The sender observes the realized state $\w \in \W$ at $t=0$. Then at each time $t=1,2,\ldots$, the sender sends message $m_t \in M_S$ and the receiver sends $\tilde{m}_t \in M_R$ simultaneously. Finally, after seeing the sequence of messages $h_{\infty} \in H_{\infty}$, the receiver chooses an action $a \in A$. A strategy for the sender is a measurable function $\sigma: \W\times H_{<\infty}\to \D M_S$ and a strategy for the receiver is a pair of measurable functions $\tilde{\sigma}: H_{<\infty}\to \D M_R$ and $\rho: H_{\infty}\to \D A$. We denote the long cheap talk game as $CT_{\infty}(p)$.

Under transparent motives, Proposition 4 of \cite{LR20} shows that every sender payoff attainable in a Nash equilibrium of $CT_{\infty}(p)$ is also attainable in a perfect Bayesian equilibrium of the one-shot cheap-talk game. Therefore, the highest sender's expected payoff that is induced by a Nash equilibrium of $CT_{\infty}(p)$ coincides with the one-shot highest cheap talk value $\VV_{CT}(p)$. A correlated equilibrium of $CT_{\infty}(p)$ is a Nash equilibrium of an extension of $CT_{\infty}(p)$ where the players privately receive correlated signals before the beginning of the game. \cite{forges1985correlated} shows that the set of correlated equilibrium payoffs of the long cheap talk game $\mathcal{C}(CT_{\infty}(p))$ is the same as the set of all communication equilibrium payoffs of the basic decision problem $\mathcal{M}(DP_0(p))$. Therefore, the highest sender's expected payoff induced by a correlated equilibrium of $CT_{\infty}(p)$ coincides with the sender's payoff induced by the sender's preferred communication equilibrium $\VV_{MD}(p)$.

A different class of games we consider is a simplified version of the infinitely repeated sender-receiver game introduced in \cite{hart1985nonzero}. There are two action sets $A_S, A_R$ for the sender and receiver, respectively. The sender observes the realized state $\w \in \W$ at $t=0$. Then at each time $t=1,2,\ldots$, the sender chooses action $a_t \in A_S$ and the receiver chooses $\tilde{a}_t \in A_R$ simultaneously. The action of the receiver is the only one that is payoff-relevant, and the sender's payoff does not depend on the state. That is, the sender's payoff at time $t$ is $u_S(\tilde{a}_t)$ and the receiver's payoff at time $t$ is $u_R(\w, \tilde{a}_t)$. The actions are observed every period, and players have perfect recall. The players' overall payoffs are defined as the liminf of the expected time average of the one-period payoffs. That is, $U_S \coloneqq \liminf_{T\to\infty} \mathbb{E}[\tfrac{1}{T}\sum_{t=1}^T u_S(\tilde{a}_t)]$ and $U_R \coloneqq \liminf_{T\to\infty} \mathbb{E}[\tfrac{1}{T}\sum_{t=1}^T u_R(\w, \tilde{a}_t)]$. This is the transparent-motive case of the repeated games of \emph{pure information transmission} as defined in \cite{forges2020games}, and we denote it as $\Gamma_\infty(p)$.

The correlated equilibria of $\Gamma_\infty(p)$ are defined similarly, and \cite{forges1985correlated} shows that the set of correlated equilibrium payoffs of this game $\mathcal{C}(\Gamma_{\infty}(p))$ coincides with the set of communication equilibrium payoffs of the basic decision problem $\mathcal{M}(DP_0(p))$. Therefore, the highest sender's expected payoff induced by a correlated equilibrium of $\Gamma_{\infty}(p)$ is the same as  the sender's payoff in a sender's preferred communication equilibrium $\VV_{MD}(p)$. Moreover, Proposition 1 of \cite{habu2024knowing} implies that the set of sender’s Nash-equilibrium payoff of $\Gamma_{\infty}(p)$ coincides with the set of sender’s payoff in a one-stage cheap talk equilibrium.

\section{Additional analysis} \label{App: Examples}
\subsection{Beyond transparent motives and multiple receivers} \label{subApp: trilemma}
Theorem \ref{Thm: BP = MD iff BP = CT} implies that the following three properties cannot hold at the same time: (1) The mediator publicly communicates with multiple receivers (i.e., there exists a unique common posterior); (2) The payoff of the sender is state-independent; (3) Mediation is optimal under persuasion and strictly better than cheap talk. This subsection gives two examples in which property (3) holds once one of the first two assumptions is relaxed. To present these examples, we first slightly extend our setting to deal with the sender's state-dependent preferences and multiple receivers with private posterior beliefs.

\paragraph{Beyond transparent motives}
The main analysis focused on the case of the state-independent sender's payoff function. Without this assumption, it is still possible to express the Honesty constraint purely in terms of the unconditional distribution of beliefs. 

Suppose that the sender's indirect payoff at state $\w$ and the receiver's posterior $\m$ is uniquely given by $V(\m,\w)$. It is easy to show (see for example \cite{doval2018constrained}) that the truth-telling constraint can be written as
\begin{equation}
\label{eq:TT_stat_dep}
    \int_{\D(\W)} V(\m,\w)\left(\frac{\m(\w)}{p(\w)}-\frac{\m(\w')}{p(\w')}\right)\de \t(\m)\ge 0 \qquad \forall \w,\w' \in \W.
\end{equation}
These are $n(n-1)$ moment constraints, hence there exists an optimal mediation plan with no more than $n^2$ signals.
Equivalently, using Bayes plausibility,
\[
\Cov_{\tau}\!\left(
V(\mu,\omega),
\frac{\mu(\omega)}{p(\omega)}
-
\frac{\mu(\omega')}{p(\omega')}
\right)
\geq 0
\qquad \forall \omega,\omega'\in\Omega .
\]
Thus, state-dependent preferences replace the zero-covariance condition in the
main text by a system of statewise moment inequalities. For each true
state \(\omega\), the covariance between the state-\(\omega\) value \(V(\cdot,\w)\) and the
posterior likelihood ratio of the truthful report must be weakly larger than the
corresponding covariance under any deviation report \(\omega'\).

We now provide an example that Theorem \ref{Thm: BP = MD iff BP = CT} may fail with a state-dependent sender's payoff. Consider a binary state space $\W = \set{0,1}$ and the prior on $\w = 1$ is $p = 1/2$. The sender's payoff is state-dependent and singleton-valued $V(\mu, \w) = G(\mu) - \tfrac{\w}{\mu}$, where
\begin{align*}
    G(\mu) = 
    \begin{cases}
        4\mu & \text{if } \mu \in [0, 1/4)\\
        -2\mu + 3/2 & \text{if } \mu \in [1/4, 1/2)\\
        2\mu - 1/2 & \text{if } \mu \in [1/2, 3/4)\\
        -4\mu + 4  & \text{if } \mu \in [3/4, 1]
    \end{cases}
\end{align*}

We show that $\tilde{\t} = \tfrac{1}{2}\delta_{1/4} + \tfrac{1}{2} \delta_{3/4}$ is feasible under mediation, that it solves Bayesian persuasion at  $p$, and that cheap talk is strictly worse than mediation.

The distribution $\tilde{\t}$ is optimal if it solves 
\[
\max_{\t \in \TT_{BP}(p)} p \int_0^1 V(\mu, 1) \de \t^1(\mu) + (1-p) \int_0^1 V(\mu, 0) \de \t^0(\mu).
\]
Bayes-plausibility implies that the objective function becomes $\int_0^1 G(\mu) \de \t - 1$, hence $\tilde{\t}$ is the unique solution of this maximization problem because it is supported on the global maximum of $G$.

Note that $\int_0^1 \tfrac{1}{\mu} \de \tilde{\t}^0(\mu) = 10/3 > 2 = \int_0^1 \tfrac{1}{\mu} \de \tilde{\t}^1(\mu)$ and $\int_0^1 G(\mu) \de \tilde{\t}^0(\mu) = \int_0^1 G(\mu) \de \tilde{\t}^1(\mu) = 1$. The truth-telling constraints for mediation $\int V(\mu, 0) \de \tilde{\t}^0(\mu) \geq \int V(\mu, 0) \de \tilde{\t}^1(\mu)$ and $\int V(\mu, 1) \de \tilde{\t}^1(\mu) \geq \int V(\mu, 1) \de \tilde{\t}^0(\mu)$ are satisfied. So $\tilde{\t}$ is implementable under mediation. However, cheap talk with state-dependent utility requires $V(\mu,\w) = V(\mu',\w)$ for all $\w\in \set{0,1}$ and $\mu,\mu' \in \supp(\tilde{\t}^{\w})$. So $\tilde{\t}$ is not feasible under cheap talk because $V(1/4, 1) = -3 \neq -1/3 = V(3/4, 1)$. As $\tilde{\t}$ is the unique solution of persuasion and $\tilde{\t}$ is not feasible under cheap talk, cheap talk attains a strictly lower value than mediation.

\paragraph{Multiple receivers and private communication.}
We next show that the conclusion can also fail under transparent motives when
private communication generates different posterior beliefs for different
receivers. We focus on a simple separable case: there are two states $\W=\{0,1\}$ and two receivers, each solves an isolated
decision problem, and the sender's payoff is additively separable in the two
receivers' posterior beliefs. Let \(\mu_i\in[0,1]\) denote receiver \(i\)'s
posterior on state \(1\), and let \(\tau_i\) denote the marginal distribution of
\(\mu_i\). Both marginals are Bayes-plausible, so $\mathbb{E}_{\tau_i}[\mu_i] = p = 1/2$ for $i=1,2$.

Suppose the sender's indirect payoff is $V(\mu_1,\mu_2)=V_1(\mu_1)+V_2(\mu_2)$.
The same Bayes-rule argument as in Theorem \ref{Thm: Implementability} gives the
aggregate truth-telling condition
\begin{equation}
\label{eq:tt_mult_rec_general}
    \int_0^1 V_1(\mu_1)(\mu_1-p)\,d\tau_1(\mu_1)
    +
    \int_0^1 V_2(\mu_2)(\mu_2-p)\,d\tau_2(\mu_2)
    =0 .
\end{equation}
Unlike the public-posterior case, the two covariance terms need not vanish
separately. Private communication allows the mediator to use the posterior of
one receiver to offset the sender's incentive generated by the posterior of the
other receiver.

Now consider an example where $V_1$ is strictly increasing and strictly convex with $V_1(1)-V_1(0)=1$, and $V_2(\mu_2) = -\rho \mu_2$ with $\rho > 1$. Under Bayesian persuasion, the sender optimally fully discloses the
state to receiver \(1\), while any Bayes-plausible information policy for
receiver \(2\) is optimal because the second term is linear. Hence the
persuasion value can be attained by taking $\tau_1^*=\frac12\delta_0+\frac12\delta_1$
and any Bayes-plausible \(\tau_2\).

For mediation, condition \eqref{eq:tt_mult_rec_general} becomes
\begin{equation}
\label{eq:tt_mult_rec}
    \int_0^1 V_1(\mu_1)(\mu_1-\tfrac12)\,d\tau_1(\mu_1)
    -
    \rho\int_0^1\mu_2(\mu_2-\tfrac12)\,d\tau_2(\mu_2)=0 .
\end{equation}
Given \(\tau_1^*\), the first term equals \(1/4\). Moreover, $\int_0^1\mu_2(\mu_2-\tfrac12)\,d\tau_2(\mu_2)$
ranges from \(0\), under no disclosure, to \(1/4\), under full disclosure. Since
the set of Bayes-plausible distributions is convex, there exists a
Bayes-plausible \(\tau_2^*\) for which this term equals \(1/(4\rho)\). Thus
\((\tau_1^*,\tau_2^*)\) satisfies \eqref{eq:tt_mult_rec} and attains the
Bayesian persuasion value.

This outcome cannot be sustained by private cheap talk. If receiver \(1\) is
fully informed, there must be messages inducing posterior \(0\) and posterior
\(1\) for receiver \(1\). Fix any message of receiver \(2\) inducing posterior
\(\mu_2\). After a message inducing posterior \(0\) for receiver \(1\), the
sender can deviate only in the message sent to receiver \(1\), replacing it by a
message inducing posterior \(1\), while keeping receiver \(2\)'s message fixed.
This deviation changes the sender's payoff from \(V_1(0)-\rho\mu_2\) to
\(V_1(1)-\rho\mu_2\), which is strictly higher because \(V_1\) is strictly
increasing. Hence full disclosure to receiver \(1\) is not implementable under
private cheap talk, while it is implementable under mediation. Therefore
mediation attains the persuasion value and strictly improves on cheap talk.

\subsection{Receiver's utility and informativeness} \label{subApp: informativeness}

Mediation need not have a
monotone implication for receiver welfare or informativeness. In some
applications, such as the acceptance-game application in Section
\ref{ssec:accepta_pareto}, the sender's gain from mediation is also an ex-ante
Pareto improvement. In general, however, receiver welfare is governed by a
different payoff function from the one appearing in the sender's truth-telling
constraint. Informativeness is likewise ambiguous: optimal mediation may be more
informative than cheap talk, as in the illustrative example when $p\in(0,1/3)$, but the reverse can
also occur.
For instance, consider a binary state space $\W = \set{0,1}$, $A = \{-1,0,1\}$, and let $\mu \in [0,1]$ denote the posterior belief on state $1$. The sender's payoff equals $a$. The receiver's payoff is such that $a = 0$ is optimal when $\mu \in [0,1/4] \cup [3/4,1]$, $a=1$ is optimal when $\mu \in [1/4,1/2]$, and $a=-1$ is optimal when $\mu \in [1/2,3/4]$. For priors \(p\in(0,1/4)\), full disclosure is optimal under cheap talk, while Theorem \ref{Thm: MD vs CT} implies that mediation is valuable, so full disclosure is suboptimal under optimal mediation.

\subsection{Nonattainment in the dual problem}
\label{app:fail_dual}

Suppose $\W = \{0,1\}$ and let $\mu \in [0,1]$ denote the posterior on state $1$. Assume that $\BV=V$ is singleton-valued. The dual problem of \eqref{Eq: MD problem} is:
\begin{align*}
    &\inf_{f_0,f_1,g\in\R}  f_1 p+f_0\\
    \textnormal{subject to:}&\; f_1\mu+f_0 \geq (1+g(\mu-\tfrac{1}{2}))V(\mu) \qquad \forall \mu\in[0,1].
\end{align*}

Suppose $V(\mu) = 4\mu(\mu-1/2)+1/4$ and $p=1/2$, the corresponding dual problem of mediation does not have a solution. Let $V^g(\mu)= (1+g(\mu-\tfrac{1}{2}))V(\mu)$. Note that when $g<0$, the lowest line above $V^g$ is a tangent line of $V^g$ at $\mu^*=\tfrac{1}{2} - \tfrac{1}{2g}$ that passes through $(0,V^g(0))$. That is, $f_1 = \tfrac{g}{4} - \tfrac{1}{g}$ and $f_0 = \tfrac{1}{4}(1-\tfrac{g}{2}) = V^g(0)$. Then the value $f_1/2+f_0 = \tfrac{1}{4} - \tfrac{1}{2g} \downarrow \tfrac{1}{4}$ as $g\to -\infty$. Also observe that $g\geq 0$ is never an optimal solution of the dual, since $(V^g(0) + V^g(1))/2 = \tfrac{5}{4} + \tfrac{g}{2} \geq \tfrac{5}{4}$. Therefore, the infimum value of this dual problem cannot be attained.

\section{Infinite state space}
\label{Sec: infinite state}
In this appendix, we show our main theorems extend to the case when $\W$ is a compact metric space endowed with its Borel sigma algebra. The other parts of the model are the same as Section \ref{sec:model}. 

We first introduce some useful mathematical notions. Let $\mathbb{M}(\W)$ denote the space of finite signed countably additive measures over $\W$. Let $\D(\W)$ denote the space of Borel probability measures over $\W$ and endow it with the topology of weak convergence, and let $C(\W)$ denote the space of continuous real functions over $\W$. Recall that a function $\psi : \D(\W) \times \R \to \mathbb{M}(\W)$ is Gelfand integrable with respect to some $\eta \in \D(\D(\W) \times \R)$ if and only if, for every continuous function $h \in C(\W)$ the map $(\m,v) \mapsto \langle h, \psi(\mu,v) \rangle$ is Lebesgue integrable with respect to $\eta$. In this case, the Gelfand integral of $\psi$ with respect to $\eta$  is denoted by $\int_{\D(\W)\times \R}\psi(\m,v) \de \eta(\m,v)$ and corresponds to the unique element of $\mathbb{M}(\W)$ such that
\begin{equation*}
    \Big \langle h,  \int_{\D(\W)\times \R}\psi(\m,v) \de \eta(\m,v) \Big \rangle =\int_{\D(\W)\times \R} \langle h, \psi(\mu,v) \rangle \de \eta(\m,v)
\end{equation*}
for all $h \in C(\W)$, where the integral on the right-hand side is a Lebesgue integral. For more details on the Gelfand integral, see Section 11.9 of \cite{aliprantis06}.

We now go back to our model. Let $p \in \D(\W)$ denote the common prior over $\W$ and assume that  $\supp(p)=\W$. By Revelation Principle, we focus on the CE outcomes $\pi \in \D(\W\times A)$ that satisfy Consistency, Obedience, and Honesty. The first two properties are defined analogously to the main text. In the case of Honesty, recall that in the main text this property is equivalent to impose that the conditional expectation $\mathbb{E}_{\pi^\w}[u_S(a)]$ is equal to the same constant for all $\w \in \W$. We extend the notion of Honesty in the current infinite setting by requiring that there exists a constant $K \in \mathbb{R}$ such that for $p$-almost-all $\w \in \W$, $\mathbb{E}_{\pi^\w}[u_S(a)] = K$, where $\pi^{\w}$ is a version of the conditional probability of given $\w  \in \W$. We define the indirect value correspondence $\BV: \D(\W) \rightrightarrows \R$ as in Section \ref{sec:imple}, which is upper hemi-continuous, compact, convex, and non-empty
valued, and the upper (lower) envelopes are denoted as $\oV$ ($\uV$). 
As Definition \ref{def: induced by CE}, we say a distribution $\eta\in\D(\D(\W)\times \R)$ is induced by some CE outcome $\pi \in \D(\W\times A)$ if $\eta = (\phi^{\pi})_{\#} \pi$, where $\phi^{\pi}: \W\times A \to \D(\W) \times \R$, with $\phi^{\pi}_1(\w,a) = \pi^a$ is a version of the conditional probability over $\W$ given $a$, and $\phi^{\pi}_2(\w,a) = u_S(a)$.

Theorem \ref{Thm: Implementability} in the main text can be extended as follows:
\begin{thmp}{1*}
    If $\eta\in\D(\D(\W) \times \R)$ is induced by some CE outcome, then it satisfies 
    \begin{itemize}
        \item [(i)]  Consistency*:
            \[
            \int_{\D(\W)\times \R} \m \de \eta(\m, v)=p;
            \]
            \item [(ii)] Obedience*: $\eta(\Gr(\BV)) = 1$;
            \item [(iii)] Honesty*:
            \[
            \int_{\D(\W) \times \R} v(\mu - p) \de \eta(\m, v) = \mathbf{0},
            \]
    where $ \mathbf{0} \in \mathbb{M}(\W)$ denotes the measure that is identically zero.\footnote{The integrals in (i) and (iii) are Gelfand integrals. Since $\W$ is compact, any $h \in C(\W)$ is bounded. Moreover, since $u_S$ is bounded, there exists $\bar v\in \R$ such that $|v|\leq \bar v$ for every $(\mu,v)\in \Gr(\BV)$. It follows that $(\mu,v)\mapsto \langle h,\mu\rangle$ and $(\mu,v)\mapsto v\langle h,\mu-p\rangle$ are bounded functions on $\Gr(\BV)$, and hence Lebesgue integrable with respect to $\eta\in \D(\Gr(\BV))$.}
    \end{itemize}
    Conversely, if $\eta$ satisfies (i), (ii), and (iii), then there exists a CE outcome $\pi\in \D(\W\times A)$ such that $\mathbb{E}_{\eta}[v] = \mathbb{E}_{\pi}[u_S]$.    
\end{thmp}
\begin{proof}
    Suppose $\eta \in \D(\D(\W)\times \R)$ is induced by some communication equilibrium outcome $\pi\in \D(\W\times A)$. For every $h\in C(\W)$,
    \begin{align*}
        \int_{\D(\W)\times \R} \langle h, \mu \rangle \de \eta(\mu,v) =& \int_{\W\times A} \langle h, \phi_1^{\pi}(\w,a)\rangle \de\pi(\w, a) = \int_{\W\times A} \langle h, \pi^a \rangle \de \pi(\w, a)\\ 
        =& \int_{\W \times A} h(\w) \de \pi(\w, a) = \langle h, p \rangle.
    \end{align*}
    The first equality is by $\eta = (\phi^{\pi})_{\#} \pi$, the second equality is by definition, the third one is by the law of iterated expectations, and the last one is by Consistency of $\pi$. Hence, $\eta$ satisfies Consistency*.

    By Obedience of $\pi$ and $\eta = (\phi^{\pi})_{\#}\pi$, $\eta(\Gr(\BV)) = \pi((\phi^{\pi})^{-1}(\Gr(\BV))) = 1$, so Obedience* is satisfied.

    By Honesty of $\pi$ and the fact that $u_S$ does not depend on $\w$, we have $\mathbb{E}_{\pi^\w}[u_S] = \mathbb{E}_{\pi}[u_S]$ $p$-almost surely. For all $h \in C(\W)$,
    \begin{align*}
        &\int_{\W\times A} u_S(a)h(\w) \de \pi(\w, a) =  \int_{\W}h(\w) \mathbb{E}_{\pi^{\w}}[u_S] \de p(\w) = \langle h, p \rangle \mathbb{E}_{\pi}[u_S] = \langle h, p \rangle \int_{\D(\W)\times \R} v \de \eta(\mu,v),
    \end{align*}
    where the first equality is by iterated expectation and Consistency, the second one follows from Honesty, and the last one is by the fact that $\eta$ is induced by $\pi$. We also have
    \begin{align*}
        &\int_{\W\times A} u_S(a)h(\w) \de \pi(\w, a) =  \int_{\W\times A} u_S(a) \langle h, \pi^a\rangle \de \pi(\w, a) = \int_{\D(\W)\times \R} v \langle h, \mu\rangle \de \eta(\mu, v) 
    \end{align*}
    where the first equality is by iterated expectation and the second one is by $\eta = (\phi^{\pi})_{\#} \pi$. Therefore, 
    \[
    \int_{\D(\W)\times \R} v \langle h, \mu - p\rangle \de \eta(\mu,v) = 0
    \]
    for every $h \in C(\W)$, so Honesty* holds.
    
    Next, we show by construction that for any $\eta\in \D(\D(\W)\times \R)$ that satisfy Consistency*, Obedience* and Honesty*, there exists a communication equilibrium outcome $\pi$ with $\mathbb{E}_{\eta}[v] = \mathbb{E}_\pi[u_S]$. 
    By Obedience*, the conditional mean $\mathbb{E}_{\eta}[v\mid \mu]$ is a measurable selector of $\BV$. Hence, Lemma 2 of \cite{LR20} implies that there exists a measurable $\lambda:\D(\W)\to \D(A)$ such that for all $\mu\in \D(\W)$, $\lambda(\mu) \in \argmax_{\alpha\in \D(A)}\mathbb{E}_{\mu\times \alpha}[u_R(\w, a)]$ is a mixed best response for the receiver with posterior $\mu$, and  $\mathbb{E}_{\eta}[v\mid \mu] = \int_A u_S(a) \de \lambda(\mu)(a)$.

    Let $\t = \marg_{\D(\W)}\eta$. Define a probability kernel $\kappa: \D(\W) \to \W \times A$ by $\kappa(\mu,\cdot) = \mu \times \l(\mu)$, which is the product measure of $\mu \in \D(\W)$ and $\l(\m) \in \D(A)$. Since $\l$ is measurable, $\kappa$ is well-defined by Lemma 3.1 of \cite{kallenberg2021}. Let $\pi = \tau \circ \kappa$, we show that $\pi$ is a desired communication equilibrium outcome. By construction, for every bounded measurable $u: \W\times A \to \R$,
    \begin{align} \label{Eq: change of var infinite}
        \mathbb{E}_{\pi}[u] = \int_{\D(\W)} \int_{\W \times A} u(\w,a) \de \kappa(\mu, \w, a) \de \t(\mu) 
        =& \int_{\D(\W)} \mathbb{E}_{\mu \times \lambda(\mu)}[u(\w,a)]\de \t(\mu).
    \end{align}

    This implies that $\mathbb{E}_{\pi}[u_S] = \mathbb{E}_{\eta}[v]$.
    The same argument as in the finite-state case (in the proof of Theorem \ref{Thm: Implementability}) shows that $\pi$ satisfies Consistency and Obedience.
    It remains to verify Honesty. Let $\pi^\w$ be a version of the conditional distribution of $\pi$ given $\w$. For every Borel $W \subseteq \W$, 
    \[
    \int_{W} \mathbb{E}_{\pi^\w}[u_S] \de p(\w) = \int_{\D(\W)} \mu(W) \mathbb{E}_{\lambda(\mu)}[u_S(a)] \de \t(\mu) = \int_{\D(\W)\times \R}\mu(W) v \de \eta(\mu,v),
    \]
    where the first equality follows from \eqref{Eq: change of var infinite} and iterated expectation, the second follows from the definition of $\l$ and iterated expectation. By Honesty*, the Gelfand integral $\int_{\D(\W)\times \R} v(\mu-p) \de \eta(\m,v)$ is the zero measure. Hence, evaluating this signed measure at the Borel set $W$,
    \[
    \int_{\D(\W)\times \R}\mu(W) v \de \eta(\mu,v) = p(W)\int_{\D(\W)\times \R} v \de \eta(\mu,v) = p(W) \mathbb{E}_{\pi}[u_S].
    \]
    This then implies that 
    \[
    \int_{W} (\mathbb{E}_{\pi^{\w}}[u_S] - \mathbb{E}_{\pi}[u_S]) \de p(\w) = 0
    \]
    for every Borel $W \subseteq \W$. Hence, $\mathbb{E}_{\pi^{\w}}[u_S] = \mathbb{E}_{\pi}[u_S]$ $p$-almost surely, $\pi$ satisfies Honesty.
\end{proof}

The attainment part in Proposition \ref{Prop: value of MD} extends as follows: Fix any sequence of feasible $\eta_n$ that converges weakly to $\eta$, for every $h \in C(\W)$, $0 = \int \langle h, \mu - p \rangle \de \eta_n(\mu, v) \to \int \langle h, \mu - p \rangle \de \eta(\mu, v)$ and $0 = \int v \langle h, \mu - p \rangle \de \eta_n(\mu, v) \to \int v \langle h, \mu - p \rangle \de \eta(\mu, v)$. So the feasibility set of the auxiliary program is compact.

Next, we extend Theorem \ref{Thm: BP = MD iff BP = CT}, Proposition \ref{Prop: sufficient cond for BP > MD}, and the first statement of Theorem \ref{Thm: MD vs CT} to a continuum of states. Recall the cheap talk hull at $p$ is defined as
\[
H^*(p) = \{\mu \in \D(\W):  \exists \mu_0 \in \D(\W) \text{ such that }  \oV_{CT}(p) \in \BV_{CT}(\mu_0) \text{ and } p\in (
        \mu_0,\mu]\}.
\]

\begin{thmp}{2*} \label{Thm: BP vs MD infinite}
   $\VV_{BP}(p) = \VV_{MD}(p) $ if and only if $ \VV_{BP}(p) = \VV_{CT}(p)$. 
\end{thmp}
\begin{proof}
    It suffices to show the only if direction. If $\max \oV - \min \uV=0$, then the statement is obvious. Thus, we now assume that  $B := \max \oV - \min \uV>0$.
    If verifiability has no value, then by the generalized attainment result, there exists $\eta \in \TT_{MD}(p)$ with $\mathbb{E}_{\eta}[v] = \VV_{BP}(p)$. Let $\bar v = \mathbb{E}_{\eta}[v]$. Fix $\e \in (0, 1/B)$, define $\eta_\e \in \D(\D(\W)\times \R)$ by 
    \[
    \frac{\de \eta_\e}{\de \eta}(\mu,v) = 1 + \e(v - \bar v) > 0
    \]
    which is a well-defined probability distribution since the Radon-Nikodym derivative $1+\e(v-\bar v)$ is positive $\eta$-almost surely and integrates to one. By construction, $\eta_\e$ has the same support as $\eta$, so Obedience* still holds. For every $h\in C(\W)$,
    \begin{align*}
        \int_{\D(\W)\times \R} (v-\bar v) \langle h, \mu \rangle \de \eta(\mu,v) =&\; \int_{\D(\W)\times \R} v \langle h, \mu \rangle \de \eta(\mu,v) - \bar v \int_{\D(\W)\times \R} \langle h, \mu \rangle \de \eta(\mu,v) \\
    =& \;\bar v \langle h, p\rangle - \bar v \langle h, p\rangle = 0,
    \end{align*}
    where the second equality follows from Honesty* and Consistency* of $\eta$. Therefore, for every $h\in C(\W)$,
    \[
    \int_{\D(\W)\times \R}  \langle h, \mu \rangle \de \eta_\e(\mu,v) = \int_{\D(\W)\times \R}  \langle h, \mu \rangle \de \eta(\mu,v)+ \e\int_{\D(\W)\times \R} (v-\bar v) \langle h, \mu \rangle \de \eta(\mu,v) = \langle h, p \rangle,
    \]
    so $\eta_\e$ satisfies Consistency* and thereby $\eta_\e \in \TT_{BP}(p)$. Finally, the expected sender payoff under $\eta_\e$ is
\[
    \mathbb{E}_{\eta_\e}[v] =  \mathbb{E}_{\eta}[v] + \e \mathbb{E}_{\eta}[(v-\bar v)v] =  \mathbb{E}_{\eta}[v] + \e \Var_\eta[v].
    \]
    Therefore, if $\Var_\eta[v] > 0$, then $\mathbb{E}_{\eta_\e}[v]> \mathbb{E}_{\eta}[v]$, contradicting the optimality of $\eta$. Hence, $\Var_\eta[v] = 0$ and $\eta$ is feasible under cheap talk.
\end{proof}

Proposition \ref{Prop: sufficient cond for BP > MD} is extended by the same proof as in Appendix \ref{app:proofs}, except for the parts about full-dimensionality, which we discuss in the next subsection. 

The definition of hull-directional improvability of cheap talk in Definition \ref{def:improv} remains the same in the current setting.

\begin{thmp}{3*}  \label{Thm: MD vs CT infinite}
    If cheap talk is hull-directionally improvable at $p$, then $\VV_{MD}(p) > \VV_{CT}(p)$.
\end{thmp}
\begin{proof}
    By the same construction as in the proof of the first statement of Theorem \ref{Thm: MD vs CT}.
\end{proof}

\subsection{Full-dimensionality under infinite state space}
A candidate for the notion of full dimensionality for the current infinite-state setting is $H^*(p)=\D(\W)$, which corresponds to the full-dimensionality notion when $\W$ is finite. However, this condition implies that \emph{for all} $\m \in \D(\W)$ there exists $\alpha>1$ such that $\alpha p +(1-\alpha)\m$ is a probability measure, a quite demanding condition.\footnote{Indeed, this is impossible when $\W$ is uncountable. Suppose for all $\m \in \D(\W)$ there exists $\alpha>1$ such that $\alpha p +(1-\alpha)\m$ is a probability measure. For every measurable $W \subseteq \W$, $p(W) = 0$ implies $\mu(W) = 0$. It follows that every singleton $\{\w\}$ is an atom for $p$, so $\W$ cannot be uncountable.} In this section, we propose a weaker notion of full dimensionality for the infinite-state setting that still collapses to the full-dimensionality notion in the main text when $\W$ is finite and allows us to extend  Theorem \ref{Thm: multidim quasiconvex}.

For every full support $p \in \D(\W)$, let $\D_p^*(\W)$ denote the set of Borel probability measures $\m$ that admit a \emph{continuous} Radon-Nikodym density $\frac{\de \m}{\de p}(\w)$ with respect to $p$. Define the continuous cheap talk hull at $p$ as
\begin{equation*}
\tilde{H}^*(p) = \{\m \in \D_p^*(\W): \exists \mu_0 \in \D(\W) \text{ such that }  \oV_{CT}(p) \in \BV_{CT}(\mu_0) \text{ and } p\in (\mu_0,\mu]\}.
\end{equation*}
Clearly, $\tilde{H}^*(p)\subseteq H^*(p)$ and the two sets coincide when $\W$ is finite. Moreover, we remark the following property of $\D_p^*(\W)$ which motivates our definition of $\tilde{H}^*(p)$.
\begin{lemma}
    Fix any $p \in \D(\W)$ with full support. Then for all $\m \in \D_p^*(\W)$, there exists $\alpha>1$ such that $\alpha p +(1-\alpha)\m \in \D(\W)$.
\end{lemma}
\begin{proof}
    Fix any $\m \in \D_p^*(\W)$ distinct from $p$ and define $M=\max_{\w \in \W}\frac{\de \m}{\de p}(\w)$. Moreover, fix any $\alpha \in \left(1,\frac{M}{M-1}\right)$. Next, observe that $\alpha p +(1-\alpha)\m \in \mathbb{M}(\W)$ and that $(\alpha p +(1-\alpha)\m)(\W)=1$. Therefore, we are left to show that $\alpha p +(1-\alpha)\m$ is positive. It is sufficient to show that $\int_{\W}h\de(\alpha p +(1-\alpha)\m) \ge 0$ for every weakly positive continuous function $h\in C(\W)$. Fix one such $h$ and observe that 
    \begin{align*}
        \int_{\W}h(\w)\de(\alpha p +(1-\alpha)\m)(\w) &=\int_{\W}h(\w)\left(\alpha+(1-\alpha)\frac{\de \m}{\de p}(\w)\right)\de p(\w) \\
        &\ge \int_{\W}h(\w)\left(\alpha+(1-\alpha)M\right)\de p(\w) \ge 0,
    \end{align*}
    where the first inequality follows from the definition of $M$ and the fact that $\alpha>1$, while the second inequality follows from the fact that $h(\w) \ge0$ for all $\w \in \W$ and $\alpha+(1-\alpha)M >0$ by construction. Given that $h$ was arbitrarily chosen, the desired result follows.
\end{proof}
We now extend the full-dimensionality condition to the current setting.

\begin{defp}{2*}
    The full-dimensionality condition holds at $p$ if $\tilde{H}^*(p)=\D_p^*(\W)$.
\end{defp}

To extend Theorem \ref{Thm: multidim quasiconvex}, we first extend our definition of moments. Let $T$ be a continuous linear map from $\D(\W)$ to a locally convex space. We say $T$ is $k$-dimensional if $X = T(\D(\W))$ has  dimension $k$.

\begin{thmp}{5*}
\label{th:quasi_star}
 Assume that $\BV(\mu) = \{v(T(\mu))\}$ for some $k$-dimensional moment $T$ ($k\geq 2$) and continuous and strictly quasiconvex $v: X \to \R$. If the full-dimensionality condition holds at $p$, then exactly one of these cases holds:
    \begin{description}
        \item[(1)] $\max V = \VV_{BP}(p) = \VV_{MD}(p) = \VV_{CT}(p) > v(T(p))$;
        \item[(2)] $\max V > \VV_{BP}(p) > \VV_{MD}(p) >
        \VV_{CT}(p) > v(T(p))$.
    \end{description}
\end{thmp}
\begin{proof}
    As $T$ is multidimensional and $v$ is strictly quasiconvex, Corollary 6 of \cite{LR20} implies no disclosure is suboptimal under cheap talk, that is $\VV_{CT}(p) > V(p)$. Because $T$ is linear and continuous and $v$ is continuous, $V(\m)\coloneqq v(T(\m))$ is also continuous. By Theorem \ref{Thm: BP vs MD infinite} and the fact that $\tilde{H}^*(p) = \D_p^*(\W)$, $\VV_{BP}(p) = \VV_{MD}(p)$ if and only if  $\{V > \VV_{CT}(p) \} = \emptyset$. To see this, if there exists $\mu\in \D(\W)$ such that  $V(\mu) > \VV_{CT}(p)$, by Theorem 3.1 of \cite{cerreia2019characterization}, there exists a sequence $\{\mu_m\}_{m \in \mathbb{N}} \subseteq \D_p^*(\W)$ such that $\mu_m \to \mu$. By continuity of $V$, there exists $\mu_k \in \tilde{H}^*(p) \subseteq H^*(p)$ in this sequence such that $\oV_{CT}(\mu_k) \geq V(\mu_k) > \oV_{CT}(p)$. It follows that $\VV_{BP}(p) > \VV_{CT}(p)$ by Theorem \ref{Thm: BP vs MD infinite}. Conversely, if $\max V = \VV_{CT}(p)$, then $\VV_{BP}(p) = \VV_{CT}(p) = \max V$ trivially holds. 
    
    This leads to the dichotomy in the theorem statement: If $\max V = \VV_{CT}(p)$, then (1) holds trivially. It suffices to show  $\max V > \VV_{CT}(p)$ implies (2). Note that if $\VV_{BP}(p) = \max V$, it must be the case that $V(\mu) = \max V$ for all $\mu$ in the support of any optimal distribution over posteriors under Bayesian persuasion, which implies $\VV_{BP}(p) = \VV_{CT}(p)$, yielding a contradiction. Hence, what remains to show is that $\VV_{MD}(p) > \VV_{CT}(p)$.

    To show that $\VV_{MD}(p) > \VV_{CT}(p)$, we follow the same construction as in the proof of Theorem \ref{Thm: multidim quasiconvex}. In particular, there exists $\m_1, \m_2 \in \D(\W)$ such that $T(\mu_1),T(\mu_2) \in D_+$ and $T(\tfrac{1}{2}\m_1+ \tfrac{1}{2}\m_2) \in D_-$, where $D_+ = \{x \in X: v(x) > \oV_{CT}(p)\}$ and $D_- = \{x \in X: v(x) < \oV_{CT}(p)\}$. By continuity of $v$, $D_+$ and $D_-$ are open, and so is $\co D_+$. Therefore, there exists $\e > 0$ such that for any $\mu_1' \in B_{\e}(\m_1)$, an $\e$-ball around $\m_1$ and any $\mu_2' \in B_{\e}(\m_2)$, we have $T(\tfrac{1}{2}\mu_1' + \tfrac{1}{2}\mu_2') \in D_-$. By Theorem 3.1 of \cite{cerreia2019characterization} and the full-dimensionality condition, there exists $\hat\mu_i \in \tilde{H}^*(p) \cap B_{\e}(\m_i)$ for $i = 1,2$ and $\hat{\mu} \coloneqq \tfrac{1}{2}\hat\mu_1 + \tfrac{1}{2}\hat\mu_2$ has moment $T(\hat\mu) \in D_-$. As 
    $\tilde{H}^*(p) $ is convex, $\hat{\mu} \in \tilde{H}^*(p) $. Because $\co D_+$ is open, there exists $\lambda \in (0,1)$ such that $T(\lambda \hat{\mu} + (1-\l) p) \in \co D_+$. We may partially extend Lemma \ref{Lem: Envelope and cvx hull} to show that for every $\mu \in \D(\W)$, $T(\mu) \in \co D_+$ implies $\oV_{CT}(\mu) >\oV_{CT}(p)$ and $T(\mu) \in \co D_-$ implies $\uV_{CT}(\mu) < \oV_{CT}(p)$. Therefore, $\oV_{CT}(\lambda \hat{\mu} + (1-\l) p) > \oV_{CT}(p) > \uV_{CT}(\hat{\mu})$. As $\hat{\mu} \in \tilde{H}^*(p) $, it follows that 
    $\VV_{MD}(p) > \VV_{CT}(p)$ by Theorem \ref{Thm: MD vs CT infinite}.
\end{proof}

As for the finite state case, it would be important to establish when the full-dimensionality condition holds in this infinite-dimensional setting. For example, in the context of Theorem \ref{th:quasi_star}, full dimensionality holds at $p$ if and only if, for all $\m \in \D_p^*(\W)$, there exists $\alpha>1$ such that $\overline{V}_{CT}(\alpha p+(1-\alpha)\m)\ge \overline{V}_{CT}(p)$, which is the case when the sender's cheap talk value is constant within an $\varepsilon$-ball of $p$ with respect to the Kullback-Leibler divergence. We leave the analysis of more primitive conditions such as minimal edge non-monotonicity in the infinite-dimensional setting to future research.

\section{Signaling games with transparent motives}
\label{app:signaling-proofs}

This section proves the signaling analogues of Theorems
\ref{Thm: Implementability}, \ref{Thm: BP = MD iff BP = CT}, and
\ref{Thm: MD vs CT}. Let \(X\) be a compact metric space of payoff-relevant
signals. Payoffs are \(u_S(x,a)\) and \(u_R(\omega,x,a)\), with \(u_S\)
independent of \(\omega\). For each \((\mu,x)\in\D(\W)\times X\), define
\[
    \mathbf W(\mu,x)
    :=
    \co\left(
    u_S\left(x,\argmax_{a\in A}\mathbb E_\mu[u_R(\omega,x,a)]\right)
    \right),
\]
and let
\[
    \mathbf V^S(\mu):=\bigcup_{x\in X}\mathbf W(\mu,x),
    \qquad
    \overline V^S(\mu):=\max\mathbf V^S(\mu),
    \qquad
    \underline V^S(\mu):=\min\mathbf V^S(\mu).
\]
A mediated signaling outcome is represented by
\(\pi\in\D(\W\times X\times A)\). Let \(\pi^{x,a}\) be a version of the
conditional distribution of \(\omega\) given \((x,a)\). The outcome \(\pi\)
induces \(\gamma^\pi\in\D(\D(\W)\times X\times\R)\) by the pushforward $(\omega,x,a)\mapsto \left(\pi^{x,a},x,u_S(x,a)\right)$,
and we write \(\eta^\pi:=\marg_{\D(\W)\times\R}\gamma^\pi\).

\begin{theorem}[Signaling analogue of Theorem \ref{Thm: Implementability}]
\label{thm:signaling-implementability-app}
If \(\gamma\in\D(\D(\W)\times X\times\R)\) is induced by a mediated signaling
communication-equilibrium outcome at prior \(p\), then, for
\(\eta=\marg_{\D(\W)\times\R}\gamma\),
\[
    \mathbb E_\eta[\mu]=p,\qquad
    \gamma(\Gr(\mathbf W))=1,\qquad
    \Cov_\eta[v,\mu]=\mathbf 0.
\]
Conversely, if \(\gamma\in\D(\D(\W)\times X\times\R)\) satisfies these three
conditions, then there is a mediated signaling communication-equilibrium
outcome \(\pi\in\D(\W\times X\times A)\) such that $\mathbb E_\pi[u_S(x,a)]=\mathbb E_\gamma[v]$.
\end{theorem}

\begin{proof}
The proof is the proof of Theorem \ref{Thm: Implementability} with \(x\) added
as a payoff-relevant coordinate.

Suppose first that \(\gamma=\gamma^\pi\). For each \(\omega\in\W\),
\[
    \int \mu(\omega)\,d\eta(\mu,v)
    =
    \int \pi^{x,a}(\omega)\,d\pi(\tilde\omega,x,a)
    =
    \int \mathbf 1_{\{\tilde\omega=\omega\}}\,d\pi(\tilde\omega,x,a)
    =
    p(\omega),
\]
so the analogue of \eqref{Eq: Bayes-plausibility} holds. Obedience implies
that, for \(\pi\)-almost every \((x,a)\), the action \(a\) is optimal at
posterior \(\pi^{x,a}\) given signal \(x\); hence
\(u_S(x,a)\in\mathbf W(\pi^{x,a},x)\), so the analogue of
\eqref{Eq: Obedience} holds.

Finally, Honesty and transparent motives imply
\(\mathbb E_{\pi^\omega}[u_S(x,a)]=\mathbb E_\pi[u_S(x,a)]\) for every
\(\omega\). Since
\(d\pi^\omega/d\marg_{X\times A}\pi(x,a)=\pi^{x,a}(\omega)/p(\omega)\),
\[
\begin{aligned}
    \Cov_\eta[v,\mu(\omega)]
    &=
    \int v\mu(\omega)\,d\eta(\mu,v)
    -
    p(\omega)\int v\,d\eta(\mu,v) \\
    &=
    \int u_S(x,a)\pi^{x,a}(\omega)\,d\pi(\tilde\omega,x,a)
    -
    p(\omega)\int u_S(x,a)\,d\pi(\tilde\omega,x,a)\\
    &=
    p(\omega)\left(
    \mathbb E_{\pi^\omega}[u_S(x,a)]-\mathbb E_\pi[u_S(x,a)]
    \right)
    =
    0.
\end{aligned}
\]
Thus the analogue of \eqref{eq:zero_cov} holds.

Conversely, let \(\tau=\marg_{\D(\W)\times X}\gamma\). Since
\(\gamma(\Gr(\mathbf W))=1\), the conditional mean
\(\bar v(\mu,x):=\mathbb E_\gamma[v\mid \mu,x]\) belongs to
\(\mathbf W(\mu,x)\) for \(\tau\)-almost every \((\mu,x)\). By the same
measurable-selection argument used in the proof of Theorem
\ref{Thm: Implementability}, there is a measurable kernel
\(\lambda:\D(\W)\times X\to\D(A)\) such that \(\lambda(\mu,x)\) is supported on
receiver best replies to \((\mu,x)\) and
\[
    \int_A u_S(x,a)\,d\lambda(\mu,x)(a)=\bar v(\mu,x).
\]
Define \(\pi\in\D(\W\times X\times A)\) by
\[
    \int f(\omega,x,a)\,d\pi
    =
    \int_{\D(\W)\times X}
    \int_{\W\times A} f(\omega,x,a)\,d\mu(\omega)d\lambda(\mu,x)(a)
    \,d\tau(\mu,x)
\]
for every bounded measurable \(f\). Then
\(\mathbb E_\pi[u_S(x,a)]=\mathbb E_\gamma[v]\), and
\(\marg_\W\pi=p\) follows from \(\mathbb E_\eta[\mu]=p\). Obedience follows
from the fact that \(\lambda(\mu,x)\) is supported on best replies. Finally,
for each \(\omega\),
\[
\begin{aligned}
    \mathbb E_{\pi^\omega}[u_S(x,a)]
    =
    \frac{1}{p(\omega)}
    \int \mu(\omega)\bar v(\mu,x)\,d\tau(\mu,x)
    =
    \frac{1}{p(\omega)}
    \int \mu(\omega)v\,d\gamma(\mu,x,v)
    =
    \mathbb E_\gamma[v],
\end{aligned}
\]
where the last equality uses \(\mathbb E_\eta[\mu]=p\) and
\(\Cov_\eta[v,\mu]=\mathbf 0\). Hence Honesty holds.
\end{proof}

Let \(\mathcal T_{BP}^S(p)\), \(\mathcal T_{MD}^S(p)\), and
\(\mathcal T_{CT}^S(p)\) denote the feasible distributions over
\((\mu,v)\) under, respectively, verifiable signaling, mediated signaling, and
direct signaling with cheap talk. Thus \(\mathcal T_{BP}^S(p)\) imposes
Bayes plausibility and \(v\in\mathbf V^S(\mu)\), while
\(\mathcal T_{MD}^S(p)\) additionally requires a lift satisfying
\(\Cov[v,\mu]=0\), as in Theorem
\ref{thm:signaling-implementability-app}. Direct signaling imposes payoff
flatness across on-path posterior-signal-message realizations. By the
belief-based characterization of \cite{koessler2024belief}, under transparent
motives the sender-preferred direct-signaling value is the quasiconcave envelope
of \(\overline V^S\). Let the corresponding values be
\(\mathcal V_{BP}^S(p)\), \(\mathcal V_{MD}^S(p)\), and
\(\mathcal V_{CT}^S(p)\).

\begin{theorem}[Signaling analogue of Theorem \ref{Thm: BP = MD iff BP = CT}]
\label{thm:signaling-bp-md-ct-app}
For every prior \(p\), $\mathcal V_{BP}^S(p)>\mathcal V_{MD}^S(p)$ if and only if $\mathcal V_{BP}^S(p)>\mathcal V_{CT}^S(p)$.
\end{theorem}

\begin{proof}
The only-if direction follows from
\(\mathcal V_{BP}^S(p)\ge \mathcal V_{MD}^S(p)\ge \mathcal V_{CT}^S(p)\).

For the converse, suppose \(\mathcal V_{BP}^S(p)=\mathcal V_{MD}^S(p)\), and
let \(\eta\in\mathcal T_{MD}^S(p)\) attain the persuasion value. The statement is obvious when $\max \oV^S = \min \uV^S$, so we assume that $B\coloneqq \max \oV^S - \min \uV^S > 0$. Let $\bar v = \mathbb{E}_\eta[v]$ and fix $\e \in (0,1/B)$. Consider the same construction as in the proof of Theorem \ref{Thm: BP = MD iff BP = CT}, define $\eta_\e \in \D(\D(\W)\times \R)$ by $\frac{\de \eta_\e}{\de \eta}(\mu,v) = 1+\e(v-\bar v) > 0$. By construction, $\eta_\e$ has the same support as $\eta$, so it satisfies the signaling analogue of \eqref{Eq: Obedience}. Since $\eta$ is feasible under mediation, $\Cov_\eta[v,\mu]=0$. Therefore,
\[
\mathbb{E}_{\eta_\e}[\mu] = \mathbb{E}_{\eta}[\mu] + \e \mathbb{E}_{\eta}[(v-\bar v)\mu] = p + \e \Cov_\eta[v,\mu] = p,
\]
so $\eta_\e$ satisfies the signaling analogue of \eqref{Eq: Bayes-plausibility}. It follows that $\eta_\e \in \TT_{BP}^S(p)$, with an expected value $\mathbb{E}_{\eta_\e}[v] = \mathbb{E}_\eta[v]+\e \Var_\eta[v]$. Therefore, if $\Var_\eta[v] >0$, then $\mathbb{E}_{\eta_\e}[v] > \mathbb{E}_\eta[v]$, contradicting the optimality of $\eta$. Hence, $v$ is constant over the support of $\eta$. 
By the flatness
characterization of \cite{koessler2024belief}, $\eta$ is feasible under direct signaling with cheap talk. Thus
\(\mathcal V_{CT}^S(p)=\mathcal V_{BP}^S(p)\).
\end{proof}

Let \(\mathbf V_{CT}^S\) be the direct-signaling cheap-talk correspondence
generated by the flat convexification of \(\mathbf V^S\), and let
\(\overline V_{CT}^S\) and \(\underline V_{CT}^S\) be its upper and lower
selections. The strict hull lemma used in the proof of Theorem
\ref{Thm: MD vs CT} applies to \(\mathbf V^S\): for every \(q\) and \(s\),
\[
    \overline V_{CT}^S(q)>s
    \quad\Longleftrightarrow\quad
    q\in\co\{\mu:\overline V^S(\mu)>s\},
\]
and
\[
    \underline V_{CT}^S(q)<s
    \quad\Longleftrightarrow\quad
    q\in\co\{\mu:\underline V^S(\mu)<s\}.
\]
This follows from the same proof as Lemma \ref{Lem: Envelope and cvx hull},
with a caveat that $\BV^S(\mu)$ need not be convex-valued. This is addressed by the truncation argument in Claim 2 of \cite{koessler2024belief} and its symmetric counterpart. For the upper-envelope equivalence, we work with 
\[
\BV^S_{up}(\mu) \coloneqq \bigcup_{x\in X}\left(\mathbf W(\mu,x)\cap \left[\max_{x'\in X}\underline{\mathbf W}(\mu,x'),+\infty\right)\right),
\]
which is interval-valued and has the same upper selection as $\BV^S$. For the lower-envelope equivalence, we use
\[
\BV^S_{down}(\mu) \coloneqq \bigcup_{x\in X}\left(\mathbf W(\mu,x)\cap \left(-\infty,\min_{x'\in X}\overline{\mathbf W}(\mu,x')\right]\right),
\]
which is interval-valued and has the same lower selection as $\BV^S$. Applying the original proof to these two truncated correspondences gives the two equivalences.

Define the signaling cheap-talk hull at \(p\) by
\[
    H^{S,*}(p)
    :=
    \left\{
    \mu\in\D(\W):
    \exists \mu_0\in\D(\W)\text{ such that }
    \overline V_{CT}^S(p)\in\mathbf V_{CT}^S(\mu_0)
    \text{ and }p\in(\mu_0,\mu]
    \right\}.
\]
Full dimensionality holds if \(H^{S,*}(p)=\D(\W)\). Direct signaling is
directionally improvable at \(p\) if there are
\((\mu^+,v^+),(\mu^-,v^-)\in\Gr(\mathbf V_{CT}^S)\) such that
\[
    \mu^+\in(p,\mu^-),
    \qquad
    v^+>\overline V_{CT}^S(p)>v^-.
\]
It is hull-directionally improvable if, in addition, \(\mu^-\in H^{S,*}(p)\).

\begin{theorem}[Signaling analogue of Theorem \ref{Thm: MD vs CT}]
\label{thm:signaling-md-ct-app}
The following hold.
\begin{enumerate}
    \item If direct signaling is hull-directionally improvable at \(p\), then
    mediation is valuable.
    \item If mediation is valuable at \(p\), then direct signaling is
    directionally improvable at \(p\).
\end{enumerate}
Moreover, if full dimensionality holds at \(p\), then mediation is valuable at
\(p\) if and only if direct signaling is directionally improvable at \(p\).
\end{theorem}

\begin{proof}
Let \(s=\overline V_{CT}^S(p)\).

For the first statement, suppose there are
\((\mu^+,v^+),(\mu^-,v^-)\in\Gr(\mathbf V_{CT}^S)\) such that
\(\mu^+\in(p,\mu^-)\), \(v^+>s>v^-\), and \(\mu^-\in H^{S,*}(p)\). Choose
\(\lambda\in(0,1)\) with
\(\mu^+=\lambda\mu^-+(1-\lambda)p\), and let
\(\eta^+\in\mathcal T_{CT}^S(\mu^+)\) and
\(\eta^-\in\mathcal T_{CT}^S(\mu^-)\) attain \(v^+\) and \(v^-\). Set $\xi
    :=
    \frac{\frac{1}{\lambda}(s-v^-)}
    {v^+-s+\frac{1}{\lambda}(s-v^-)}$.
Then $\mathbb E_{\xi\eta^+ +(1-\xi)\eta^-}
    [(v-s)(\mu-p)]
    =
    \mathbf 0$.
Since \(\mu^-\in H^{S,*}(p)\), there exist \(\mu_0\) and
\(\alpha\in(0,1)\) such that
\(\overline V_{CT}^S(p)\in\mathbf V_{CT}^S(\mu_0)\) and
\[
    p=(1-\alpha)\mu_0+\alpha\big(\xi\mu^+ +(1-\xi)\mu^-\big).
\]
Let \(\eta_0\in\mathcal T_{CT}^S(\mu_0)\) attain \(s\), and define
\[
    \tilde\eta
    :=
    (1-\alpha)\eta_0+\alpha\xi\eta^+
    +\alpha(1-\xi)\eta^-.
\]
Then \(\tilde\eta\) satisfies the signaling analogues of
\eqref{Eq: Bayes-plausibility}, \eqref{Eq: Obedience}, and
\eqref{eq:zero_cov}, and hence is mediation-feasible by Theorem
\ref{thm:signaling-implementability-app}. Its payoff is
\[
    \mathbb E_{\tilde\eta}[v]
    =
    s+
    \alpha\left(\frac{1}{\lambda}-1\right)
    \frac{(v^+-s)(s-v^-)}
    {v^+-s+\frac{1}{\lambda}(s-v^-)}
    >
    s,
\]
so \(\mathcal V_{MD}^S(p)>\mathcal V_{CT}^S(p)\).

For the second statement, suppose mediation is valuable. By the finite-support
argument in Remark \ref{Rmk: Existence}, there is
\(\eta\in\mathcal T_{MD}^S(p)\) with finite support and
\(\mathbb E_\eta[v]>s\). Let
\(H=\{(\mu,v):v>s\}\) and \(L=\{(\mu,v):v<s\}\), and define
\[
    A=\int_H (v-s)\,d\eta(\mu,v),
    \qquad
    B=\int_L (s-v)\,d\eta(\mu,v).
\]
Then \(A>B\ge 0\). We first show \(B>0\). If \(B=0\), then \(v\ge s\)
\(\eta\)-almost surely and, using \eqref{eq:zero_cov},
\[
    \mathbf 0
    =
    \int (v-s)(\mu-p)\,d\eta(\mu,v)
    =
    \int_H (v-s)(\mu-p)\,d\eta(\mu,v).
\]
Thus
\[
    p=\frac{1}{A}\int_H (v-s)\mu\,d\eta(\mu,v).
\]
Since \(\eta\) has finite support and every point in \(H\cap\supp(\eta)\)
satisfies \(\overline V^S(\mu)>s\), the strict hull lemma implies
\(\overline V_{CT}^S(p)>s\), a contradiction. Hence \(B>0\).

Define
\[
    \bar\mu^+
    :=
    \frac{1}{A}\int_H (v-s)\mu\,d\eta(\mu,v),
    \qquad
    \bar\mu^-
    :=
    \frac{1}{B}\int_L (s-v)\mu\,d\eta(\mu,v).
\]
By Bayes plausibility and zero covariance,
\[
    A(\bar\mu^+-p)-B(\bar\mu^- -p)=\mathbf 0,
\]
so $\bar\mu^+-p=\frac BA(\bar\mu^- -p)$.
Since \(0<B/A<1\), \(\bar\mu^+\in(p,\bar\mu^-)\). Moreover,
\(\bar\mu^+\in\co\{\mu:\overline V^S(\mu)>s\}\) and
\(\bar\mu^-\in\co\{\mu:\underline V^S(\mu)<s\}\). By the strict hull lemma, $\overline V_{CT}^S(\bar\mu^+)>s$ and $\underline V_{CT}^S(\bar\mu^-)<s.$
Hence there are
\((\bar\mu^+,v^+),(\bar\mu^-,v^-)\in\Gr(\mathbf V_{CT}^S)\) with
\(v^+>s>v^-\) and \(\bar\mu^+\in(p,\bar\mu^-)\). Direct signaling is
directionally improvable.

Under full dimensionality, directional improvability and hull-directional
improvability coincide, so the two statements imply the final equivalence.
\end{proof}

\begingroup
\small
\begin{spacing}{0.9}
\setlength{\bibsep}{0pt plus 0.2ex}
\renewcommand{\bibsection}{%
  \section*{References for the Online Appendix}
  \vspace{-0.5em}
}
\putbib[reference]
\end{spacing}
\endgroup

\end{bibunit}

\end{document}

%% file: theorem3_simplex_figure_input.tex
\begin{tikzpicture}[
    x=1cm,
    y=1cm,
    line join=round,
    line cap=round,
    simplex/.style={draw=black, line width=0.75pt, fill=black!2},
    direction/.style={draw=black, line width=0.75pt},
    extension/.style={draw=black!55, line width=0.6pt},
    benchmark/.style={draw=black!60, line width=0.7pt, densely dashed},
    valuebar/.style={draw=black!70, line width=0.55pt},
    valuecurve/.style={draw=black, line width=0.7pt, densely dashed},
    valuepoint/.style={circle, draw=black, fill=white, inner sep=1.25pt},
    beliefpoint/.style={circle, fill=black, inner sep=1.45pt},
    belieflabel/.style={font=\footnotesize, inner sep=1.2pt, fill opacity=.92, text opacity=1},
    valuelabel/.style={font=\footnotesize, inner sep=1.2pt, fill opacity=.92, text opacity=1}
]

\coordinate (A) at (0,0);
\coordinate (B) at (7.1,0);
\coordinate (C) at (2.75,4.65);
\draw[simplex] (A) -- (B) -- (C) -- cycle;

\coordinate (L) at (-0.35,0.49);
\coordinate (R) at (7.35,3.03);
\draw[extension] (L) -- (R);

\coordinate (p)       at (1.45,1.08);
\coordinate (muplus)  at (3.05,1.61);
\coordinate (muminus) at (4.05,1.94);

\draw[direction] (p) -- (muminus);

\coordinate (vp)       at (1.45,2.00);
\coordinate (vmuplus)  at (3.05,3.16);
\coordinate (vmuminus) at (4.05,2.37);

\draw[valuebar] (p) -- (vp);
\draw[valuebar] (muplus) -- (vmuplus);
\draw[valuebar] (muminus) -- (vmuminus);

\coordinate (benchL) at (-0.05,1.50);
\coordinate (benchR) at (5.95,3.48);
\draw[benchmark] (benchL) -- (benchR);


\node[valuepoint] at (vp) {};
\node[valuepoint] at (vmuplus) {};
\node[valuepoint] at (vmuminus) {};

\node[beliefpoint] at (p) {};
\node[beliefpoint] at (muplus) {};
\node[beliefpoint] at (muminus) {};

\node[belieflabel, anchor=north] at (1.42,0.95) {$p$};
\node[belieflabel, anchor=north] at (3.05,1.53) {$\mu^+$};
\node[belieflabel, anchor=north] at (4.13,1.85) {$\mu^-$};

\node[valuelabel, anchor=south west] at (3.10,3.18)
    {$ v^+$};
\node[valuelabel, anchor=south east] at (1.40,2.02)
    {$\overline V_{CT}(p)$};
\node[valuelabel, anchor=south west] at (4.13,2.40)
    {$v^-$};


\end{tikzpicture}